\documentclass[pra,aps,nopacs,onecolumn,twoside,superscriptaddress]{revtex4}

\usepackage{soul} 
\usepackage{multirow}
\usepackage{amsmath,amsfonts,amssymb,caption,color,epsfig,graphics,graphicx,hyperref,latexsym,mathrsfs,revsymb,theorem,url,verbatim,epstopdf,enumerate}
\usepackage{amsmath,amsfonts,amssymb,caption,hyperref,color,epsfig,graphics,graphicx,latexsym,mathrsfs,revsymb,theorem,url,verbatim,epstopdf,cleveref}
\usepackage{fontenc}
\hypersetup{colorlinks,linkcolor={blue},citecolor={blue},urlcolor={red}}

\usepackage{cleveref}

\newtheorem{definition}{Definition}
\newtheorem{proposition}[definition]{Proposition}
\newtheorem{lemma}[definition]{Lemma}

\newtheorem{theorem}[definition]{Theorem}
\newtheorem{corollary}[definition]{Corollary}
\newtheorem{conjecture}[definition]{Conjecture}

\newtheorem{remark}[definition]{Remark}
\newtheorem{example}[definition]{Example}
\newtheorem{question}[definition]{Question}
\newtheorem{memo}[definition]{Memo}

\def\squareforqed{\hbox{\rlap{$\sqcap$}$\sqcup$}}
\def\qed{\ifmmode\squareforqed\else{\unskip\nobreak\hfil
\penalty50\hskip1em\null\nobreak\hfil\squareforqed
\parfillskip=0pt\finalhyphendemerits=0\endgraf}\fi}
\def\endenv{\ifmmode\;\else{\unskip\nobreak\hfil
\penalty50\hskip1em\null\nobreak\hfil\;
\parfillskip=0pt\finalhyphendemerits=0\endgraf}\fi}
\newenvironment{proof}{\noindent \textbf{{Proof.~} }}{\qed}
\def\Dbar{\leavevmode\lower.6ex\hbox to 0pt
{\hskip-.23ex\accent"16\hss}D}
\makeatletter
\def\url@leostyle{%
  \@ifundefined{selectfont}{\def\UrlFont{\sf}}{\def\UrlFont{\small\ttfamily}}}
\makeatother
\def\bcj{\begin{conjecture}}
\def\ecj{\end{conjecture}}
\def\bcr{\begin{corollary}}
\def\ecr{\end{corollary}}
\def\bd{\begin{definition}}
\def\ed{\end{definition}}
\def\bea{\begin{eqnarray}}
\def\eea{\end{eqnarray}}
\def\bem{\begin{enumerate}}
\def\eem{\end{enumerate}}
\def\bex{\begin{example}}
\def\eex{\end{example}}
\def\bim{\begin{itemize}}
\def\eim{\end{itemize}}
\def\bl{\begin{lemma}}
\def\el{\end{lemma}}
\def\bma{\begin{bmatrix}}
\def\ema{\end{bmatrix}}
\def\bpf{\begin{proof}}
\def\epf{\end{proof}}
\def\bpp{\begin{proposition}}
\def\epp{\end{proposition}}
\def\bqu{\begin{question}}
\def\equ{\end{question}}
\def\br{\begin{remark}}
\def\er{\end{remark}}
\def\bt{\begin{theorem}}
\def\et{\end{theorem}}
\def\bmm{\begin{memo}}
\def\emm{\end{memo}}

\def\btb{\begin{tabular}}
\def\etb{\end{tabular}}

\newcommand{\nc}{\newcommand}

\nc{\as}{{\cal AS}}
\nc{\app}{{\cal AP}}

\def\a{\alpha}
\def\b{\beta}
\def\g{\gamma}
\def\d{\delta}

\def\l{\lambda}

\def\n{\nu}

\def\r{\rho}

\def\c{\chi}

 \nc{\bbA}{\mathbb{A}} \nc{\bbB}{\mathbb{B}} \nc{\bbC}{\mathbb{C}}
 \nc{\bbD}{\mathbb{D}} \nc{\bbE}{\mathbb{E}} \nc{\bbF}{\mathbb{F}}
 \nc{\bbG}{\mathbb{G}} \nc{\bbH}{\mathbb{H}} \nc{\bbI}{\mathbb{I}}
 \nc{\bbJ}{\mathbb{J}} \nc{\bbK}{\mathbb{K}} \nc{\bbL}{\mathbb{L}}
 \nc{\bbM}{\mathbb{M}} \nc{\bbN}{\mathbb{N}} \nc{\bbO}{\mathbb{O}}
 \nc{\bbP}{\mathbb{P}} \nc{\bbQ}{\mathbb{Q}} \nc{\bbR}{\mathbb{R}}
 \nc{\bbS}{\mathbb{S}} \nc{\bbT}{\mathbb{T}} \nc{\bbU}{\mathbb{U}}
 \nc{\bbV}{\mathbb{V}} \nc{\bbW}{\mathbb{W}} \nc{\bbX}{\mathbb{X}}
 \nc{\bbZ}{\mathbb{Z}}

 \nc{\bA}{{\bf A}} \nc{\bB}{{\bf B}} \nc{\bC}{{\bf C}}
 \nc{\bD}{{\bf D}} \nc{\bE}{{\bf E}} \nc{\bF}{{\bf F}}
 \nc{\bG}{{\bf G}} \nc{\bH}{{\bf H}} \nc{\bI}{{\bf I}}
 \nc{\bJ}{{\bf J}} \nc{\bK}{{\bf K}} \nc{\bL}{{\bf L}}
 \nc{\bM}{{\bf M}} \nc{\bN}{{\bf N}} \nc{\bO}{{\bf O}}
 \nc{\bP}{{\bf P}} \nc{\bQ}{{\bf Q}} \nc{\bR}{{\bf R}}
 \nc{\bS}{{\bf S}} \nc{\bT}{{\bf T}} \nc{\bU}{{\bf U}}
 \nc{\bV}{{\bf V}} \nc{\bW}{{\bf W}} \nc{\bX}{{\bf X}}
 \nc{\bZ}{{\bf Z}}

\nc{\cA}{{\cal A}} \nc{\cB}{{\cal B}} \nc{\cC}{{\cal C}}
\nc{\cD}{{\cal D}} \nc{\cE}{{\cal E}} \nc{\cF}{{\cal F}}
\nc{\cG}{{\cal G}} \nc{\cH}{{\cal H}} \nc{\cI}{{\cal I}}
\nc{\cJ}{{\cal J}} \nc{\cK}{{\cal K}} \nc{\cL}{{\cal L}}
\nc{\cM}{{\cal M}} \nc{\cN}{{\cal N}} \nc{\cO}{{\cal O}}
\nc{\cP}{{\cal P}} \nc{\cQ}{{\cal Q}} \nc{\cR}{{\cal R}}
\nc{\cS}{{\cal S}} \nc{\cT}{{\cal T}} \nc{\cU}{{\cal U}}
\nc{\cV}{{\cal V}} \nc{\cW}{{\cal W}} \nc{\cX}{{\cal X}}
\nc{\cZ}{{\cal Z}}

\nc{\hA}{{\hat{A}}} \nc{\hB}{{\hat{B}}} \nc{\hC}{{\hat{C}}}
\nc{\hD}{{\hat{D}}} \nc{\hE}{{\hat{E}}} \nc{\hF}{{\hat{F}}}
\nc{\hG}{{\hat{G}}} \nc{\hH}{{\hat{H}}} \nc{\hI}{{\hat{I}}}
\nc{\hJ}{{\hat{J}}} \nc{\hK}{{\hat{K}}} \nc{\hL}{{\hat{L}}}
\nc{\hM}{{\hat{M}}} \nc{\hN}{{\hat{N}}} \nc{\hO}{{\hat{O}}}
\nc{\hP}{{\hat{P}}} \nc{\hR}{{\hat{R}}} \nc{\hS}{{\hat{S}}}
\nc{\hT}{{\hat{T}}} \nc{\hU}{{\hat{U}}} \nc{\hV}{{\hat{V}}}
\nc{\hW}{{\hat{W}}} \nc{\hX}{{\hat{X}}} \nc{\hZ}{{\hat{Z}}}

\nc{\hn}{{\hat{n}}}

\def\diag{\mathop{\rm diag}}

\def\max{\mathop{\rm max}}

\def\dg{\dagger}

\def\ra{\rightarrow}

\newcommand{\ket}[1]{|#1\rangle}
\newcommand{\proj}[1]{| #1\rangle\!\langle #1 |}

\def\Dbar{\leavevmode\lower.6ex\hbox to 0pt
{\hskip-.23ex\accent"16\hss}D}

\begin{document}

\title{Extreme points of absolutely PPT states with exactly three distinct eigenvalues}

\date{\today}


\author{Nalan Wang}\email[]{nalanwang@buaa.edu.cn}
\affiliation{LMIB(Beihang University), Ministry of Education, and School of Mathematical Sciences, Beihang University, Beijing 100191, China}

\author{Lin Chen}\email[]{linchen@buaa.edu.cn (corresponding author)}
\affiliation{LMIB(Beihang University), Ministry of Education, and School of Mathematical Sciences, Beihang University, Beijing 100191, China}

\author{Zhiwei Song}\email[]{zhiweisong@buaa.edu.cn (corresponding author)}
\affiliation{LMIB(Beihang University), Ministry of Education, and School of Mathematical Sciences, Beihang University, Beijing 100191, China}


\begin{abstract}
Whether the sets of absolutely separable (AS) and absolutely two-qutrit positive-partial-transpose (AP) states are the same has been an open problem in entanglement theory for decades. Since they are both convex sets, we investigate the boundary and extreme points of full-rank two-qutrit AP states with exactly three distinct eigenvalues. We show that every boundary point is an extreme point, with exactly one exception. We explicitly characterize the expressions of such points, each of which turns out to contain at most one parameter in some intervals. When the parameter approaches the ends of intervals, most points become the known extreme points of exactly two distinct eigenvalues. We present our results by tables and figures.  
\end{abstract}

\maketitle


\tableofcontents

\section{Introduction}

The separability problem is a long-standing open problem arising from quantum information theory, and known as an NP-hard problem in computational complexity. Determining whether a quantum state is separable (equivalently, non-entangled) is a key step to confirm the quantumness of many practically useful protocols such as quantum computing, cryptography and teleportation. From a theoretical perspective, non-separability may demonstrate the non-locality which is one of the fundamental features of quantum mechanics. As a result, much efforts have been devoted into the study of separability problem in the past decades, by developing various tools from linear and operator algebra. The most important tool is the positive partial transpose (PPT) criterion \cite{peres1996separability,horodecki1996necessary}. This is a linear positive map under which a separable state on $\bbC^m\otimes\bbC^n$ remains positive. The criterion is also sufficient when $mn\le6$. However, there exist entangled states which are positive partial transpose (PPT) for $mn>6$ \cite{horodecki1997separability}. The construction and detecting of such states have attracted lots of attentions in the quantum-information community. The separability of quantum states of simplest system, namely two-qubit states have been widely studied from more aspects such as volume measure \cite{zyczkowski1998volume,kus2001geometry,zyczkowski2003hilbert,huong2024separability}, eigenvalues \cite{slater2009eigenvalues,tanaka2014determining}, and operator monotone \cite{slater2021quasirandom}. Three-qubit separable states have also been investigated \cite{han2017separability}.
The boundary and length of separable states have also been studied \cite{chen2013dimensions,chen2015boundary,halder2021characterizing}. Separable and PPT states have also been studied numerically \cite{leinaas2007extreme,leinaas2010numerical}. 

As a proper subset of separable (resp. PPT) states, absolutely separable (resp. PPT) states were firstly proposed as a kind of separable (resp. PPT) states close to the maximally mixed states \cite{gurvits2002largest,verstraete2001maximally,adhikari2021constructing}. It turns out that the properties of absolute separability (AS) and absolute PPT (AP) are essentially determined by the spectrum of density matrices \cite{knill2003separability,gurvits2005better,hildebrand2007entangled,2015Positive}. AS and PPT states have been studied in terms of maps and channels \cite{filippov2017absolutely}. The difference has been shown between absolute separability and positive
partial transposition in the symmetric subspace \cite{louvet2025nonequivalence,serrano2024absolute}.
There have been lots of efforts for improving the bounds of absolute separability  \cite{abellanet2024improving}. The connection between resource theory and absolute separability was also constructed \cite{patra2023resource}.

By definition, AS and AP states are respectively separable and PPT up to any global unitary operation. Such states exist, say the maximally mixed states. Further, AS states are a subset of AP states. Whether AS and AP states coincide has been an open problem for decades \cite{arunachalam2014absolute}. The problem has a positive answer for qubit-qudit system \cite{Johnston2013Separability}. It is known that the sets of AS and AP states are both convex and compact. Hence they are convex hulls of extreme points. The open problem would have a positive answer when they have the same extreme points. Recently, the extreme points of qubit-qudit AP states have been totally characterized \cite{2025Extreme}. In the same paper, the extreme points of two-qutrit AP states with exactly two distinct eigenvalues have also been analytically constructed. 

In this paper, we shall continue in this vein to study the extreme points of full-rank two-qutrit AP states with exactly three distinct eigenvalues $a,b$ and $c$ in the decreasing order. We begin by formally introducing the convex sets of separable, PPT, absolutely separable and absolutely PPT states, as well as their extreme points as defined in Definition \ref{absolutely_separable_rho}. We then explore interior points from the perspective of line segments, as revealed by the lemma \ref{equivalent_definition_of_interiorpoint}. It helps readers independently check whether the points constructed in the tables later are boundary or interior points. We also review the necessary and sufficient conditions, by which a two-qutrit state lies in $\app_{3,3}$ and becomes a boundary/extreme point in Lemmas \ref{AP3,3_judge_positive_semi-definite} and \ref{boun-extr}. They actually have been used to generate all extreme points of two-qutrit AP states of exactly two distinct eigenvalues, as we show in Lemma \ref{le:two distinct eigenvalues}. From Sec. \ref{sec:mu(a)=1}, we introduce our findings of characterizing extreme points of $\app_{3,3}$. As the simplest case, we begin by investigating the boundary points of $\app_{3,3}$ with exactly one maximal eigenvalue 
in Lemmas \ref{le:a1b7c,a4b4c} and \ref{le:3cbbbbbccc}. They are mostly extreme points of $\app_{3,3}$, with the only exceptional point $\nu_{1,5,3}:=\diag(3c,b,b,b,b,b,c,c,c)$, which turns out to be a non-extreme point in Lemma \ref{le:nu(1,5,3)}. We further show that almost all boundary points of $\app_{3,3}$ with exactly three distinct eigenvalues are extreme points except $\nu_{1,5,3}$ by Theorem \ref{all_3eigen}. Then we proceed to characterize the expressions of above-mentioned extreme points in terms of the numbers of largest eigenvalues $\mu(a)$ in Sec. \ref{sec:mu(a)=1}-\ref{sec:mu(a)>2}. Most extreme points have exactly one parameter in some intervals. When the parameter approach the two ends of interval, the extreme points become the known extreme points of exactly two distinct eigenvalues. For the convenience of readers, we present our findings for $\mu(a)=1,2,...,7$ in Table \ref{tab:performance1}-\ref{tab:performance7}, respectively. We also show the connection of extreme points in Figure \ref{fig:fig1}-\ref{fig:fig7} and \ref{fig:umbrella}. The last figure shows the Umbrella Model concluding the main findings of this paper. Some calculations are presented in appendix \cref{Appendix}.

The rest of this paper is organized as follows. In Sec. \ref{sec:pre} we introduce the preliminary facts used in this paper. In Sec. \ref{sec:res} we study the existence of boundary and extreme points of two-qutrit AP states with various numbers of distinct eigenvalues. They are further characterized in Sec. \ref{sec:mu(a)=1}-\ref{sec:mu(a)>2}. Finally we conclude in Sec. \ref{sec:con}.

\section{Preliminaries}
\label{sec:pre}

In this section we introduce the preliminary knowledge of this paper. A bipartite state $\r$ is called separable if it is the convex sum of product states. It is known that a separable state is PPT though the converse usually fails. The definitions of absolute separable (AS) and absolute PPT (AP) states are presented as follows.

\begin{definition}
\label{absolutely_separable_rho}
    A separable (resp. PPT) state $\r$ is called absolutely separable (resp. PPT) if it remains separable (PPT) regardless of what global unitary operation is applied. That is, for any unitary matrix $U$, the state $U\r U^\dg$ is separable (resp. PPT). We shall respectively refer to $\as_{m,n}$ and $\app_{m,n}$ as the sets of $m\times n$ AS and AP states. 
\qed
\end{definition}
The set of separable and PPT states are both convex sets. An extreme point of a convex set cannot be expressed as a non-trivial convex combination of two  distinct points in the convex set. 
A state $\rho \in \as_{m,n}$ (resp. $\app_{m,n}$) is called an interior point if there exists $\epsilon > 0$ s.t. $\frac{1}{1-\epsilon} (\rho - \epsilon \frac{1}{mn} I_{mn}) \in \as_{m,n}$ (resp. $\app_{m,n}$). Otherwise, $ \rho$ is called a boundary point. So any extreme point of $\as_{m,n}$ (resp. $\app_{m,n}$) is a boundary point. We can explore interior points from the perspective of line segments, as revealed by the following lemma. It helps readers check whether the points constructed in Table \ref{tab:performance1}-\ref{tab:performance7} are boundary or interior points, independently from the argument used in subsequent paragraphs. 
\begin{lemma}
\label{equivalent_definition_of_interiorpoint}
    A state $\rho \in \as_{m,n}$ (resp. $\app_{m,n}$) is an interior point if and only if for any state $v \in \as_{m,n}$ (resp.\quad $\app_{m,n}$), there exists $\epsilon > 0$ s.t. for any $0 < \delta < \epsilon$, we have 
    \begin{equation}
        \frac{1}{1+\delta}(\rho + \delta v) \in \as_{m,n} (resp.\quad  \app_{m,n}) \ and \  \frac{1}{1-\delta}(\rho - \delta v) \in \as_{m,n} (resp.\quad  \app_{m,n}).
    \end{equation}
\end{lemma}
\begin{proof}
    Firstly, we prove the only if part. Suppose $\rho$ is an interior point of $\as_{m,n}$(resp. $\app_{m,n}$). From the definition of interior points, there exists $\epsilon_0 > 0$ s.t. 
    \begin{eqnarray}
        \rho_{\epsilon_0} := \frac{1}{1-\epsilon_0} \left( \rho - \epsilon_0 \frac{I_{mn}}{mn} \right) \in \mathcal{AS}_{m,n}(resp. \quad \app_{m,n}).
    \end{eqnarray}
    Take an arbitrary $v \in \mathcal{AS}_{m,n}$(resp. $\app_{m,n}$), and for sufficiently small $\delta$, we consider
    \begin{eqnarray}
    \label{eq:rho+delta v}\frac{\rho + \delta v}{1+\delta} = \frac{ (1-\epsilon_0)\rho_{\epsilon_0} + \epsilon_0 \frac{I_{mn}}{mn} + \delta v}{1+\delta}
    \end{eqnarray}
    where $\delta>0$. This can be expressed as a convex combination of $\rho_{\epsilon_0}, \frac{I_{mn}}{mn}, v$, with nonnegative coefficients. As long as $\delta$ is small enough, \eqref{eq:rho+delta v} belongs to the convex set $\mathcal{AS}_{m,n}$(resp. $\app_{m,n}$), and the following is also an AS (resp. AP) state,
\begin{eqnarray}
         \frac{\rho - \delta v}{1-\delta} = \frac{ (1-\epsilon_0)\rho_{\epsilon_0} + \epsilon_0 \frac{I_{mn}}{mn} - \delta v}{1-\delta}.
    \end{eqnarray}
    Therefore, for sufficiently small $\epsilon > 0$, we have
    \begin{equation}
        \frac{1}{1+\delta}(\rho + \delta v) \in \as_{m,n} (resp.\quad  \app_{m,n}) \ and \  \frac{1}{1-\delta}(\rho - \delta v) \in \as_{m,n} (resp.\quad  \app_{m,n})
    \end{equation}
    for all $0<\delta<\epsilon$. 

    Secondly, we prove the if part. In particular, take $v = \frac{I_{mn}}{mn}$, which clearly belongs to $\mathcal{AS}_{m,n}$(resp. \quad $\app_{m,n}$). Then there exists some $\delta_0 > 0$ s.t.
    \begin{eqnarray}
         \frac{1}{1-\delta_0} \left( \rho - \delta_0 \frac{I}{mn} \right) \in \mathcal{AS}_{m,n}(resp.\quad  \app_{m,n}).
    \end{eqnarray}
    Therefore, $\rho$ must be an interior point. The proof is complete.
\end{proof}

In the following, we review some necessary facts from \cite{2025Extreme}. They will be used in the derivation of novel boundary and extreme points of $\app_{3,3}$ in the next sections. We shall use frequently the following two $3\times3$ real symmetric matrices,
\begin{eqnarray}
&&
L_1(x):=
\bma 
2x_9 & x_8-x_1 & x_6-x_2 \\
x_8-x_1 & 2x_7 & x_5-x_3 \\
x_6-x_2 & x_5-x_3 & 2x_4 
\ema,
\\&&
l_1(x):=\det L_1(x),
\\&&
L_2(x):=
\bma 
2x_9 & x_8-x_1 & x_7-x_2 \\
x_8-x_1 & 2x_6 & x_5-x_3 \\
x_7-x_2 & x_5-x_3 & 2x_4 
\ema,
\\&&
l_2(x):=\det L_2(x),
\end{eqnarray}
where $x_1\ge \cdots \ge x_9>0$. 

\begin{lemma}
\label{AP3,3_judge_positive_semi-definite}
(i) The full-rank two-qutrit state $\r$ of eigenvalue vector $\l$ belongs to $\mathcal{AP}_{3,3}$ if and only if $L_1(\lambda), \ L_2(\lambda)$ are both positive semi-definite matrices. That is, $l_1(\l)\ge0$, $l_2(\l)\ge0$, and $4x_7x_9\ge(x_8-x_1)^2$. 

(ii) If $\r$ is a boundary point of $\mathcal{AP}_{3,3}$, then at least one of $l_1(\lambda), l_2(\lambda)$ equals zero. 
\end{lemma}

\begin{lemma}
\label{boun-extr}
    Let $\r$ be a full-rank boundary point of $\mathcal{AP}_{3,3}$ and $U, V$ be the order-three real orthogonal matrices s.t.
    \begin{eqnarray}
    \label{U-orth}
        U^{T} L_1(\lambda) U = D_1, 
    \end{eqnarray}
    \begin{eqnarray}
    \label{V-orth}
        V^{T} L_2(\lambda) V = D_2,
    \end{eqnarray}
    where $U := [\vec{u_1} \quad \vec{u_2} \quad \vec{u_3}]$, $V := [\vec{v_1} \quad \vec{v_2} \quad \vec{v_3}]$, and $D_1, D_2$ are positive semi-definite and diagonal matrices with diagonal entries in nonincreasing order.
    
    (i) Suppose $l_1(\lambda) = 0,\ l_2(\lambda) > 0$. Then $\r$ is an extreme point if and only if the following linear equations in terms of $t_1, \ldots, t_9$ have only the trivial solution,
    \begin{eqnarray}
    \label{sum_ti}
        \sum_{i=1}^{9} t_i = 0, 
    \end{eqnarray}
    \begin{eqnarray}
    \label{t_k=t_k+1}
        t_k = t_{k+1} \text{ whenever } \lambda_k = \lambda_{k+1}, 
    \end{eqnarray}
    \begin{eqnarray}
    \label{u-equa}
        \begin{bmatrix}2t_9 & t_8 - t_1 & t_6 - t_2 \\ t_8 - t_1 & 2t_7 & t_5 - t_3 \\ t_6 - t_2 & t_5 - t_3 & 2t_4 \end{bmatrix} \cdot \vec{u_3} =  \begin{bmatrix} 0 \\ 0 \\ 0 \end{bmatrix}. 
    \end{eqnarray}

    (ii) Suppose $l_1(\lambda) > 0, l_2(\lambda) = 0$. Then $\r$ is an extreme point if and only if \eqref{sum_ti} and \eqref{t_k=t_k+1} hold, and the following linear equation in terms of $t_1, \ldots, t_9$ has only the trivial solution:
    \begin{eqnarray}
    \label{v-equa}
        \begin{bmatrix} 2t_9 & t_8 - t_1 & t_7 - t_2 \\ t_8 - t_1 & 2t_6 & t_5 - t_3 \\ t_7 - t_2 & t_5 - t_3 & 2t_4 \end{bmatrix} \cdot \vec{v_3} =  \begin{bmatrix} 0 \\ 0 \\ 0 \end{bmatrix}. 
    \end{eqnarray}

    (iii) Suppose $l_1(\lambda) = 0, l_2(\lambda) = 0$. Then $\r$ is an extreme point if and only if \eqref{sum_ti}-\eqref{v-equa} in terms of $t_1, \ldots, t_9$ have only the trivial solution.
\end{lemma}
Using the above lemmas, the authors in \cite{2025Extreme} have derived the sets of two-qutrit boundary and extreme points of $\app_{3,3}$ with exactly two distinct eigenvalues. That is, the two sets are the same. We present them as follows.
\begin{lemma}
\label{le:two distinct eigenvalues}
Let $\rho \in \mathcal{AP}_{3,3}$ has exactly two distinct eigenvalues. Then $\rho$ is an extreme point if and only if it is unitarily equivalent to one of the following eight states:
\begin{align}
\label{zeta_1}
\zeta_1 &= \frac{1}{11}\diag\{3, 1, 1, 1, 1, 1, 1, 1, 1\},\\
\label{zeta_2}
\zeta_2 &= \frac{1}{9 + 2\sqrt{2}}\diag\{\sqrt{2} + 1, \sqrt{2} + 1, 1, 1, 1, 1, 1, 1, 1\},\\
\label{zeta_3}
\zeta_3 &= \frac{1}{12}\diag\{2, 2, 2, 1, 1, 1, 1, 1, 1\},\\
\label{zeta_4}
\zeta_4 &= \frac{1}{10 + \sqrt{17}}\diag\{\frac{5 + \sqrt{17}}{4}, \frac{5 + \sqrt{17}}{4}, \frac{5 + \sqrt{17}}{4}, \frac{5 + \sqrt{17}}{4}, \frac{5 + \sqrt{17}}{4}, 1, 1, 1, 1\},\\
\label{zeta_5}
\zeta_5 &= \frac{1}{5x + 4}\diag\{x, x, x, x, x, 1, 1, 1, 1\},\\
\label{zeta_6}
\zeta_6 &= \frac{1}{21}\diag\{3, 3, 3, 3, 3, 3, 1, 1, 1\},\\
\label{zeta_7}
\zeta_7 &= \frac{1}{23 + 14\sqrt{2}}\diag\{3 + 2\sqrt{2}, \ldots, 3 + 2\sqrt{2}, 1, 1\},\\
\label{zeta_8}
\zeta_8 &= \frac{1}{8}\diag\{1, 1, 1, 1, 1, 1, 1, 1, 0\}.
\end{align}
where $x \approx 2.70928$ is a root of the equation $x^3 - x^2 - 5x + 1 = 0$.
\qed
\end{lemma}

\subsection{An inequality}

In this subsection, we shall use Lemma \ref{AP3,3_judge_positive_semi-definite} (i) to determine whether a given two-qutrit state is AP. To simplify Lemma \ref{AP3,3_judge_positive_semi-definite} (i), we show that
$l_1(\l)\ge0$ and $l_2(\l)\ge0$ imply the last inequality of Lemma \ref{AP3,3_judge_positive_semi-definite} (i), i.e.,  
\begin{eqnarray}
\label{an inequality} 4\lambda_7\lambda_9-(\lambda_1-\lambda_8)^2 \geq 0. 
\end{eqnarray}
If this is true then we can ignore the above inequality in subsequent sections.

We focus on the case of three distinct eigenvalues $a, b, c$, which means \eqref{an inequality} can be reformulated as 
\begin{eqnarray}
\label{an inequality_2}
    4c\lambda_7-(a-\lambda_8)^2 \geq 0.
\end{eqnarray}

\begin{lemma}
Suppose $\r=\diag(\l_1,....,\l_9)$ is full-rank with three distinct eigenvalues $a>b>c>0$. If the determinants $l_1(\lambda) \geq 0, \ l_2(\lambda) \geq 0$ and at least one of $l_1(\lambda), \ l_2(\lambda)$  is zero, then the inequality \eqref{an inequality_2} holds. In this case, we obtain $\r \in \mathcal{AP}_{3,3}$. 
\end{lemma}
\begin{proof}
We outline the proof. We begin by studying $\lambda_8$, and then  classify  $\lambda_7$.  We prove by contradiction to show the inequality \eqref{an inequality_2}.
Because $a>b>c>0$, it holds that $\lambda_8<a$. We have two cases, namely (i) $\lambda_8=b$ and (ii) $\lambda_8=c$.

(i) Suppose $\lambda_8 = b$, then $\lambda_7 = a$ or $\lambda_7 = b$. If $\lambda_7 = a$, it suffices to consider \eqref{eq:rho711}. Here, $l_1(\lambda) = l_2(\lambda)$, so the conditions of the lemma implies $l_1(\lambda) = l_2(\lambda) = 0$. Suppose \eqref{an inequality_2} does not hold, then we obtain the system
\begin{eqnarray}
    4c\lambda_7-(a-\lambda_8)^2 = 4 \frac{1 - b - c}{7} c - (\frac{1 - b - c}{7} - b)^2 < 0,
\end{eqnarray}
with $0<c<b<\frac{1 - b - c}{7}$. By computation, this system has no solution. If $\lambda_7 = b$, then we assume that there are $k\in[2,7]$ eigenvalues $b$.
Then the left hand side of \eqref{an inequality_2} can be reformulated as 
\begin{eqnarray}
\label{f(k)1}
    f(k):=4bc - (\frac{1-kb-c}{8-k}-b)^2.
\end{eqnarray}
We have 
\begin{eqnarray}
    f'(k) = -2\frac{(1-8b-c)^2}{(8-k)^3}<0.
\end{eqnarray}
Therefore, the minimum of $f(k)$ in \eqref{f(k)1} is achieved at $k=7$. To prove \eqref{an inequality_2}, it suffices to show that $f(7)\ge0$ using
\eqref{eq:r_a7bc}. Here, $l_1(\lambda) = l_2(\lambda)$, so the condition of the lemma implies $l_1(\lambda) = l_2(\lambda) = 0$, namely the lhs of
\eqref{an inequality_2} equals zero. We have proven \eqref{an inequality_2}. 

(ii) Suppose $\lambda_8 = c$. Then $\lambda_7 = b$ or $\lambda_7 = c$. If $\lambda_7 = b$, then we also assume that there are $k\in[1,6]$ eigenvalues $b$.
Then the lhs of \eqref{an inequality_2} can be reformulated as 
\begin{eqnarray}
\label{f(k)2}
    g(k):=4bc - (\frac{1-kb-2c}{7-k}-c)^2. 
\end{eqnarray}
We have 
\begin{eqnarray}
    g'(k) = -2\frac{\left[1-kb-(9-k)c\right](1-7b-2c)}{(7-k)^3}<0. 
\end{eqnarray}
Therefore, the minimum of $g(k)$ in \eqref{f(k)1} is achieved at $k=6$. To prove \eqref{an inequality_2}, it suffices to show that $g(6)\ge0$ using
\eqref{eq:r_a6bc}. Here, $l_1(\lambda) = l_2(\lambda)$, so the condition of the lemma implies $l_1(\lambda) = l_2(\lambda) = 0$, namely the lhs of
\eqref{an inequality_2} equals zero. We have proven \eqref{an inequality_2}.  

If $\lambda_7 = c$, then we assume that there are $k_b\in[1,5]$ eigenvalues $b$ and $k_c\in[3,7]$ eigenvalues $c$, such that $4\le k_b+k_c\le8$. Then the lhs of \eqref{an inequality_2} can be reformulated as 
\begin{eqnarray}
\label{f(k)3}
    f(k_b,k_c):=4c^2 - (\frac{1-k_bb-k_cc}{9-k_b-k_c}-c)^2. 
\end{eqnarray}
We have 
\begin{eqnarray}
    \frac{\partial f}{\partial k_b} = -2 \frac{[1-k_bb-(9-k_b)c][1-(9-k_c)b-k_cc]}{(9-k_b-k_c)^3}<0,
\end{eqnarray}
\begin{eqnarray}
    \frac{\partial f}{\partial k_c} = -2 \frac{[1-k_bb-(9-k_b)c]^2}{(9-k_b-k_c)^3}<0.
\end{eqnarray}
To prove \eqref{an inequality_2}, it suffices to show that $f(k_b,8-k_b)\ge0$. We have
\begin{eqnarray}
    \frac{df}{dk_b} = -2[1-k_bb-(9-k_b)c](c-b)>0. 
\end{eqnarray}
To prove \eqref{an inequality_2}, it suffices to show that $f(1,7)\ge0$ using
$\lambda(\r)=(a,b,c,c,c,c,c,c,c)$. Here, $l_1(\lambda) = l_2(\lambda)$, so the condition of the lemma implies $l_1(\lambda) = l_2(\lambda) = 0$. Suppose \eqref{an inequality_2} does not hold, then we obtain the system
\begin{eqnarray}
    f(1,7)<0,l_1(\lambda) = 0.
\end{eqnarray}
By computation, this system has no solution. We have proven \eqref{an inequality_2}. 

To conclude, the proof is complete. 
\end{proof}

\section{Existence of extreme points with three distinct eigenvalues}
\label{sec:res}

In this section, we show that some boundary points $\n_{1,k,8-k}:=\diag(a,b,...,b,c,...,c)$ of two-qutrit AP states are extreme points of distinct eigenvalues $a>b>c>0$. We also characterize their expressions. Here $k\in[2,6]$ counts the multiplicity of $b$ being the eigenvalues of $\n_{1,k,8-k}$. Some facts are extended to AP states of two maximum eigenvalues. In Lemma \ref{le:a1b7c,a4b4c}, we consider the relatively simple cases, i.e.,  $\lambda(\r) = (a, b, c, c, c, c, c, c, c)$ and $\lambda(\r) = (a, b, b, b, b, c, c, c, c)$, whose proof illustrates a general methodology. We examine a system of linear equations in terms of $t_1, \cdots  t_9$ under a unitarily equivalent form of a state $\r$, and then apply similar techniques to other cases of $\r$ in Lemma \ref{le:3cbbbbbccc}. Next for the special point $\nu_{1,5,3}$, we construct the one-to-one correspondence between the set of them and the interior of convex hull of $\zeta_1$ and $\zeta_6$ in lemma \ref{le:nu(1,5,3)}. In addition, we obtain a more general conclusion applicable to other cases of boundary points in Theorem \ref{all_3eigen}. 

\begin{lemma}
\label{le:a1b7c,a4b4c}
    Suppose $\r \in \mathcal{AP}_{3,3}$ is full-rank with three distinct eigenvalues, and the largest eigenvalue of $\r$ is simple. If $\r$ is a boundary point of $\app_{3,3}$ with $\lambda(\r) = (a, b, c, c, c, c, c, c, c)$ or $\lambda(\r) = (a, b, b, b, b, c, c, c, c)$ where $a>b>c>0$, then $\r$ is an extreme point of $\app_{3,3}$.
\end{lemma}
\begin{proof}
First, it is known that the boundary point $\r \in \mathcal{AP}_{3,3}$ with $\lambda(\r) = (a, b, c, c, c, c, c, c, c)$ $(a>b>c>0)$ is an extreme point of $\app_{3,3}$ in \cite{2025Extreme}.

Second, given the boundary point 
\begin{eqnarray}
\label{eq:r_abbbbc}
    \lambda(\r) = (a, b, b, b, b, c, c, c, c) \quad (a>b>c>0),
\end{eqnarray}
the condition implies that $l_{1}(\lambda) = l_{2}(\lambda) = 0$ because $\lambda_6 = \lambda_7$. According to Lemma \ref{boun-extr}, to prove $\r$ in \eqref{eq:r_abbbbc} is an extreme point, it suffices to prove that the following equations have only the trivial solution, 
\begin{eqnarray}
    t_{1} + 4t_{2} + 4t_{9} = 0, 
    \quad
    t_2 = t_3 = t_4 = t_5, \quad t_6 = t_7 = t_8 =t_9, 
\end{eqnarray}
\begin{eqnarray}
\label{vec_u3_eq4b}
    \begin{bmatrix} 2t_{9} & t_{9} - t_{1} & t_{9} - t_{2} \\ t_{9} - t_{1} & 2t_{9} & 0 \\ t_{9} - t_{2} & 0 & 2t_{2} \end{bmatrix} \cdot \vec{u_3} = \begin{bmatrix} 0 \\ 0 \\ 0 \end{bmatrix}. 
\end{eqnarray}
Let $\vec{u_3} = \begin{bmatrix} u_{13} \\ u_{23} \\ u_{33} \end{bmatrix}$. Eq. \eqref{vec_u3_eq4b} can be reformulated as
\begin{eqnarray}
\label{vec_t_eq4b}
    \begin{bmatrix} -u_{23} & -u_{33} & 2u_{13} + u_{23} + u_{33} \\ -u_{13} & 0 & u_{13} + 2u_{23} \\ 0 & -u_{13} + 2u_{33} & u_{13} \end{bmatrix} \cdot \begin{bmatrix} t_{1} \\ t_{2} \\ t_{9} \end{bmatrix} = \begin{bmatrix} 0 \\ 0 \\ 0 \end{bmatrix}.
\end{eqnarray}
Performing elementary matrix operations on the coefficient matrix of \eqref{vec_t_eq4b}, we obtain $\begin{bmatrix} -u_{23} & -u_{33} & u_{13} \\ -u_{13} & 0 & u_{23} \\ 0 & -u_{13} + 2u_{33} & u_{33} \end{bmatrix}$. We have two cases (a) and (b) in terms of $u_{33}$.

(a) $u_{33} \neq 0$. Similar to the last condition, the rank of the coefficient matrix is at least two.

(b) $u_{33} = 0$. The coefficient matrix above is $\begin{bmatrix} -u_{23} & 0 & u_{13} \\ -u_{13} & 0 & u_{23} \\ 0 & -u_{13} & 0 \end{bmatrix}$. Since $\vec{u_3}$ is a column vector of an orthogonal matrix, $u_{13}^2 + u_{23}^2 \neq 0$, which implies that the first and third columns are linearly independent if $u_{13} \neq \pm u_{23}$ and then the rank of the coefficient matrix is at least two. If $u_{13} = \pm u_{23}$, the second and third columns are linearly independent because $-u_{13} \neq 0$ and then the rank of the coefficient matrix is at least two.

So in case (a) and (b), using the rank-nullity theorem, the dimension of the solution space of $\begin{bmatrix} t_1 \\ t_2  \\ t_9 \end{bmatrix}$ is at most one. Hence the solution has the form $\begin{bmatrix} ka \\ kb  \\ kc \end{bmatrix}$ for any $k \in \mathbb{R}$, which means that $t_{1}, t_{2}, t_{9}$ must be all zero. In conclusion, if the boundary point $\r$ is unitarily equivalent to $\diag \{a, b, b, b, b, c, c, c, c\}$, it is an extreme point.
\end{proof}

Reviewing the proof above, we have examined a system of linear equations in terms of $t_1, \cdots  t_9$ under a unitarily equivalent form of a state $\r$ in \eqref{eq:r_abbbbc}. We performed elementary row operations on the coefficient matrix to simplify the system and then analyzed the possible values of its entries. Then we could study the rank of the matrix. Using the rank‑nullity theorem, we have analyzed the solution space and conclude that when its dimension is at most one, only the trivial solution exists, which means that boundary points imply extreme points. Next, we apply similar techniques to other cases of $\r$, with the proof given in appendix \ref{sec:abbcccccc}.
\begin{lemma}
\label{le:3cbbbbbccc}
Suppose $\r \in \mathcal{AP}_{3,3}$ is full-rank with three distinct eigenvalues, and the largest eigenvalue of $\r$ is simple. If $\r$ is a boundary point of $\app_{3,3}$, then $\r$ is an extreme point of $\app_{3,3}$ unless $\lambda(\r) = (3c, b, b, b, b, b, c, c, c) \quad (3c>b>c>0)$ with $l_1(\r)=0$.
\end{lemma} 
According to Lemma \ref{le:two distinct eigenvalues}, we then explore the convex combination decomposition of $\nu_{1,5,3}$ up to unitary similarity. 

\begin{lemma}
\label{le:nu(1,5,3)}
    Suppose $\nu_{1,5,3}$ is full-rank with $\lambda(\r) = (3c, b, b, b, b, b, c, c, c) \quad (3c>b>c>0)$. Then up to unitary similarity, the set of such $\nu_{1,5,3}$ is in one-to-one correspondence with the interior of convex hull of $\zeta_1$ in \eqref{zeta_1} and $\zeta_6$ in \eqref{zeta_6}, i.e.,
    \begin{eqnarray}
        \{\nu_{1,5,3}\} \leftrightarrow \{x\zeta_1 + (1-x)\zeta_6\} \ \text{where} \ x \in (0, 1). 
    \end{eqnarray}
\end{lemma}

\begin{proof}
(i) Given $\nu_{1,5,3}$, we use the method of undetermined coefficients to prove there exists a unique convex combination $x\zeta_1 + (1-x)\zeta_6$ where $x \in (0, 1)$ s.t. $\nu_{1,5,3} = x\zeta_1 + (1-x)\zeta_6$. 

Suppose $\nu_{1,5,3}$ is full-rank with \eqref{eq:r_abbbbbc} and $a=3c$. Since $3c+5b+3c = 1$, \eqref{eq:r_abbbbbc} can be reformulated as
\begin{eqnarray}
\label{eq:rho_a5bc.2}
    \lambda(\nu_{1,5,3}) = (3c, \frac{1-6c}{5}, \frac{1-6c}{5}, \frac{1-6c}{5}, \frac{1-6c}{5}, \frac{1-6c}{5}, c, c, c).
\end{eqnarray}
From \eqref{eq:rho_a5bc.2}, we also obtain that $0<c<\frac{1-6c}{5}<3c<1$. This implies $\frac{1}{21}<c<\frac{1}{11}$. The range of values for $c$ suggests a connection between the $\nu_{1,5,3}$ above and $\zeta_1$ in \eqref{zeta_1} and $\zeta_6$ in \eqref{zeta_6}. We consider the equation with undetermined coefficients
\begin{eqnarray}
    \nu_{1,5,3} = x \zeta_1 +y \zeta_6,
\end{eqnarray}
which can be reformulated as a system of linear equations: 
\begin{eqnarray}
    \frac{3}{11}x+\frac{1}{7}y = 3c,
\end{eqnarray}
\begin{eqnarray}
    \frac{1}{11}x+\frac{1}{7}y = \frac{1-6c}{5}.
\end{eqnarray}
Since the coefficient matrix of this system is of full rank, there is a unique solution 
\begin{eqnarray}
    x = \frac{11(21c-1)}{10}, \ y = \frac{21(1-11c)}{10}, 
\end{eqnarray}
where $\frac{1}{21}<c<\frac{1}{11}$. From the above analysis, for any full-rank $\nu_{1,5,3}$ satisfying the requirements of \eqref{eq:r_abbbbbc} and $a=3c>b>c>0$, there exists a unique $x \in (0, 1)$s.t.
\begin{eqnarray}
    \nu_{1,5,3} = x \zeta_1 + (1-x) \zeta_6,
\end{eqnarray}
where $\zeta_1 = \frac{1}{11}\diag\{3, 1, 1, 1, 1, 1, 1, 1, 1\}, \ \zeta_6 = \frac{1}{21}\diag\{3, 3, 3, 3, 3, 3, 1, 1, 1\}, \ x=\frac{11(21c-1)}{10}$. 

(ii) Given the convex combination of $\zeta_1$ in \eqref{zeta_1} and $\zeta_6$ in \eqref{zeta_6}, we have
\begin{eqnarray}
\label{eq:nv(1,5,3)=non-exe}
    \r = x \zeta_1 + (1-x) \zeta_6, \quad x\in (0,1). 
\end{eqnarray}
One can see $\r$ is full-rank with $\lambda(\r) = (3c, b, b, b, b, b, c, c, c) \quad (3c>b>c>0)$. As is demonstrated before, $\r$ is a $\nu_{1,5,3}$. 
\end{proof}
Now we apply the above technique to describe all extreme points in $\app_{3,3}$. The following fact is proven in appendix \ref{sec:all exe points of 3 distinct eigenvalues}.

\begin{theorem}
\label{all_3eigen}
    Suppose $\r \in \mathcal{AP}_{3,3}$ is full-rank with three distinct eigenvalues. If $\r$ is a boundary point of $\app_{3,3}$, then $\r$ is an extreme point of $\app_{3,3}$ unless $\r $ is unitarily equivalent to $ \nu_{1,5,3}$ in Lemma \ref{le:nu(1,5,3)}.
\end{theorem}

Using \eqref{eq:nv(1,5,3)=non-exe}, we conjecture that every non-extreme boundary point of $\app_{3,3}$ can be expressed as a convex combination of $\zeta_1, \cdots,\zeta_8$ in \eqref{zeta_1}-\eqref{zeta_8}. The analysis above shows that the three components of such special $\r$ are governed by a single parameter, with two of them being proportional. Consequently, if such a special $\r$ exists, the two distinct eigenvalues of the two states of \eqref{zeta_1}-\eqref{zeta_8} in the convex combination must share the same proportionality. For example, $\lambda_1 = 3\lambda_2$ in \eqref{zeta_1} and $\lambda_6 = 3\lambda_7$ in \eqref{zeta_6}. However, no other two states within \eqref{zeta_1}-\eqref{zeta_8} ensemble exhibit this same relationship again.

\section{Extreme points with maximum eigenvalue of multiplicity one}
\label{sec:mu(a)=1}

We have shown the existence of extreme points $\r$ of $\app_{3,3}$ with three distinct eigenvalues in Theorem \ref{all_3eigen}. From this section we analytically construct these points. We begin by considering the case of exactly one maximum eigenvalue, i.e., $\r = \nu_{1,k,8-k} = \diag \{a,b,\dots,b,c,\dots,c\}$ where $k\in[1,7]$, and the multiplicity of $b$ is $k$. The results are presented in Table \ref{tab:performance1}. In the rest of this section, we shall derive the results.

We shall often use the software Mathematica to compute the roots of constant-coefficient equations and polynomial-coefficient equations. We follow the "real roots first, in ascending numerical order" principle when sorting them. When the computations are straightforward, we discuss the acceptance or rejection of each solution in detail. For more complex computations, we present the results directly using "Mathematica", due to the similarity of the underlying reasoning.

The same principle is used in the following discussion. For brevity, only case (ii) is shown in the main text to illustrate this technique, while the rest are provided in appendix \ref{(i)-(iii)-(vii)_of_tab1}.

(ii) We consider the boundary point 
\begin{eqnarray}
\label{nu1,2,6}
    \nu_{1,2,6} = \diag \{a, b, b, c\cdots, c\} = \diag \{1-2b-6c, b, b, c\cdots, c\}.
\end{eqnarray}
This condition implies that $l_{1}(\lambda) = l_{2}(\lambda) = 0$ because $\lambda_6 = \lambda_7$. We have
\begin{eqnarray}
\label{eq:nu1,2,6}
    \left|\begin{matrix}2c&2b+7c-1&c-b\\2b+7c-1&2c&c-b\\c-b&c-b&2c\\\end{matrix}\right|
    = 0,
\end{eqnarray}
i.e. 
\begin{eqnarray}
    4 b^3 - 2 b^2 (1 + 3 c) + 12 b (1 - 6 c) c - 2 c (1 - 13 c + 40 c^2) = 0,
\end{eqnarray}
and from \eqref{nu1,2,6}, we obtain
\begin{eqnarray}
    b \in (c,\frac{1-6c}{3}), i.e.\ c \in (0, \frac{1}{9}).
\end{eqnarray}
The solution of \eqref{eq:nu1,2,6} are 
\begin{eqnarray}
    b_1 = \frac{1-5c}{2}, \ 
    b_2 = 2c + \sqrt{12c^2 - c}, \ 
    b_3 = 2c - \sqrt{12c^2 - c}. 
\end{eqnarray}

(ii.1) $b_1 = \frac{1-5c}{2}$. Then $a=1-2b_1-6c=-c<0$. This leads to a contradiction, so we discard this solution. 

Consequently, the solution of $b$ is restricted to either $b_2$ or $b_3$. Thus, $12c^2 - c \geq 0$. So 
\begin{eqnarray}
\label{eq:c_range1_126}
    \frac{1}{12} \le c < \frac{1}{9}. 
\end{eqnarray}

(ii.2) $b_2 = 2c + \sqrt{12c^2 - c}$. Then $a=1-2b_2-6c=1-10c-2 \sqrt{12c^2 - c}>2c + \sqrt{12c^2 - c}$. This implies $1-12c>3 \sqrt{12c^2 - c}\geq0$, and so $c< \frac{1}{12}$. However, this leads to a contradiction with \eqref{eq:c_range1_126}, so we discard this solution.

(ii.3) $b_3 = 2c - \sqrt{12c^2 - c}$. Then $a=1-2b_3-6c=1-10c+2 \sqrt{12c^2 - c}>2c - \sqrt{12c^2 - c}$, which always holds by \eqref{eq:c_range1_126} except $c=1/12$. Besides, from $b_3>c$, we obtain $c<\frac{1}{11}$. Thus
\begin{eqnarray}
\label{c_range2_126}
    0.083333\cdots = \frac{1}{12} < c < \frac{1}{11} = 0.090909\cdots. 
\end{eqnarray}
Therefore, the boundary point $\r = \nu_{1,2,6}$ in \eqref{nu1,2,6} with \eqref{c_range2_126} satisfies
\begin{eqnarray}
    a=1-10c+2 \sqrt{12c^2 - c}, \ b=2c- \sqrt{12c^2 - c}. 
\end{eqnarray}
So $\n_{1,2,6}$ is an extreme point of $\app_{3,3}$. 

In addition, consider the limit of the extreme point $\nu_{1,2,6}$. We obtain
\begin{eqnarray}
    \lim_{c\rightarrow\frac{1}{11}} \nu_{1,2,6} = \zeta_1, \ \lim_{c\rightarrow\frac{1}{12}} \nu_{1,2,6} = \zeta_3. 
\end{eqnarray}

Using the same technique as in (i) (ii), the following fact is proven in appendix \ref{(i)-(iii)-(vii)_of_tab1} and we summarize them in Table \ref{tab:performance1}. 

\newpage

\begin{table}[htbp]
\centering
\caption{Classification of extreme points of $\nu_{1,k,8-k}$ where $k\in[1,7]$ with Remark \ref{rem:nu1,4,4}}
\label{tab:performance1}
\renewcommand{\arraystretch}{5}
\begin{tabular}{|l|c|c|c|c|}
\toprule
$\nu_{1,k,8-k}$ & a & b & range of c & limits \\ \hline

$\nu_{1,1,7}$ & $\frac{1-7c+\sqrt{-73 c^2 + 18 c - 1}}{2}$ & $\frac{1-7c-\sqrt{-73 c^2 + 18 c - 1}}{2}$ & $0.084542\cdots = \frac{9-2\sqrt{2}}{73} < c < \frac{1}{11} = 0.090909\cdots$ & \shortstack{$\lim_{c\rightarrow\frac{1}{11}} \nu_{1,1,7} = \zeta_1$\\$\lim_{c\rightarrow\frac{9-2\sqrt{2}}{73}} \nu_{1,1,7} = \zeta_2$} \\ \hline

$\nu_{1,2,6}$ & $1-10c+2 \sqrt{12c^2 - c}$ & $2c- \sqrt{12c^2 - c}$ & $0.083333\cdots = \frac{1}{12} < c < \frac{1}{11} = 0.090909\cdots$ & \shortstack{$\lim_{c\rightarrow\frac{1}{11}} \nu_{1,2,6} = \zeta_1$\\$\lim_{c\rightarrow\frac{1}{12}} \nu_{1,2,6} = \zeta_3$} \\ \hline

$\nu_{1,3,5}$ & $\frac{1+10c-3\sqrt{108c^2-20c+1}}{4}$ & $\frac{1-10c+\sqrt{108c^2-20c+1}}{4}$ & $0.070805 \cdots = \frac{10-\sqrt{17}}{83} < c < \frac{1}{11} = 0.090909 \cdots$ & \shortstack{$\lim_{c\rightarrow\frac{1}{11}} \nu_{1,3,5} = \zeta_1$\\$\lim_{c\rightarrow\frac{10-\sqrt{17}}{83}} \nu_{1,3,5} = \zeta_4$} \\ \hline

$\nu_{1,4,4}$ & $1-4b^{(1)}-4c$ & $b^{(1)}$ & $0.056991\cdots = y < c < \frac{1}{11} = 0.090909 \cdots$ & \shortstack{$\lim_{c\rightarrow\frac{1}{11}} \nu_{1,4,4} = \zeta_1$\\$\lim_{c\rightarrow y} \nu_{1,4,4} = \zeta_5$} \\ \hline

$\nu_{1,6,2}$ & $\frac{2c+\sqrt{6c-17c^2}}{3}$ & $\frac{3-8c-\sqrt{6c-17c^2}}{18}$ & $0.023365 \cdots = \frac{23-14\sqrt{2}}{137} < c < \frac{1}{11} = 0.090909 \cdots$ & \shortstack{$\lim_{c\rightarrow\frac{1}{11}} \nu_{1,6,2} = \zeta_1$\\$\lim_{c\rightarrow\frac{23-14\sqrt{2}}{137}} \nu_{1,6,2} = \zeta_7$} \\ \hline

$\nu_{1,7,1}$ & $\frac{4-11c+7\sqrt{8c-7c^2}}{32}$ & $\frac{4-3c-\sqrt{8c-7c^2}}{32}$ & $0 < c < \frac{1}{11} = 0.090909 \cdots$ & \shortstack{$\lim_{c\rightarrow\frac{1}{11}} \nu_{1,7,1} = \zeta_1$\\$\lim_{c\rightarrow0} \nu_{1,7,1} = \zeta_8$} \\ \hline
\end{tabular}
\end{table}

\begin{remark}
\label{rem:nu1,4,4}
    In particular, the eigenvalue $b^{(1)}$ of $\nu_{1,4,4}$ is the second root of the cubic equation $16x^3+(41c-8)x^2+(19c^2-10c+1)x+c^3=0$ with $y < c < \frac{1}{11}$ where $y=0.056991\cdots$ is the second root of the cubic equation $481y^3-37y^2-17y+1=0$.
\end{remark}

\newpage

\begin{figure}[htbp]
    \centering
    \includegraphics[width=0.8\textwidth]{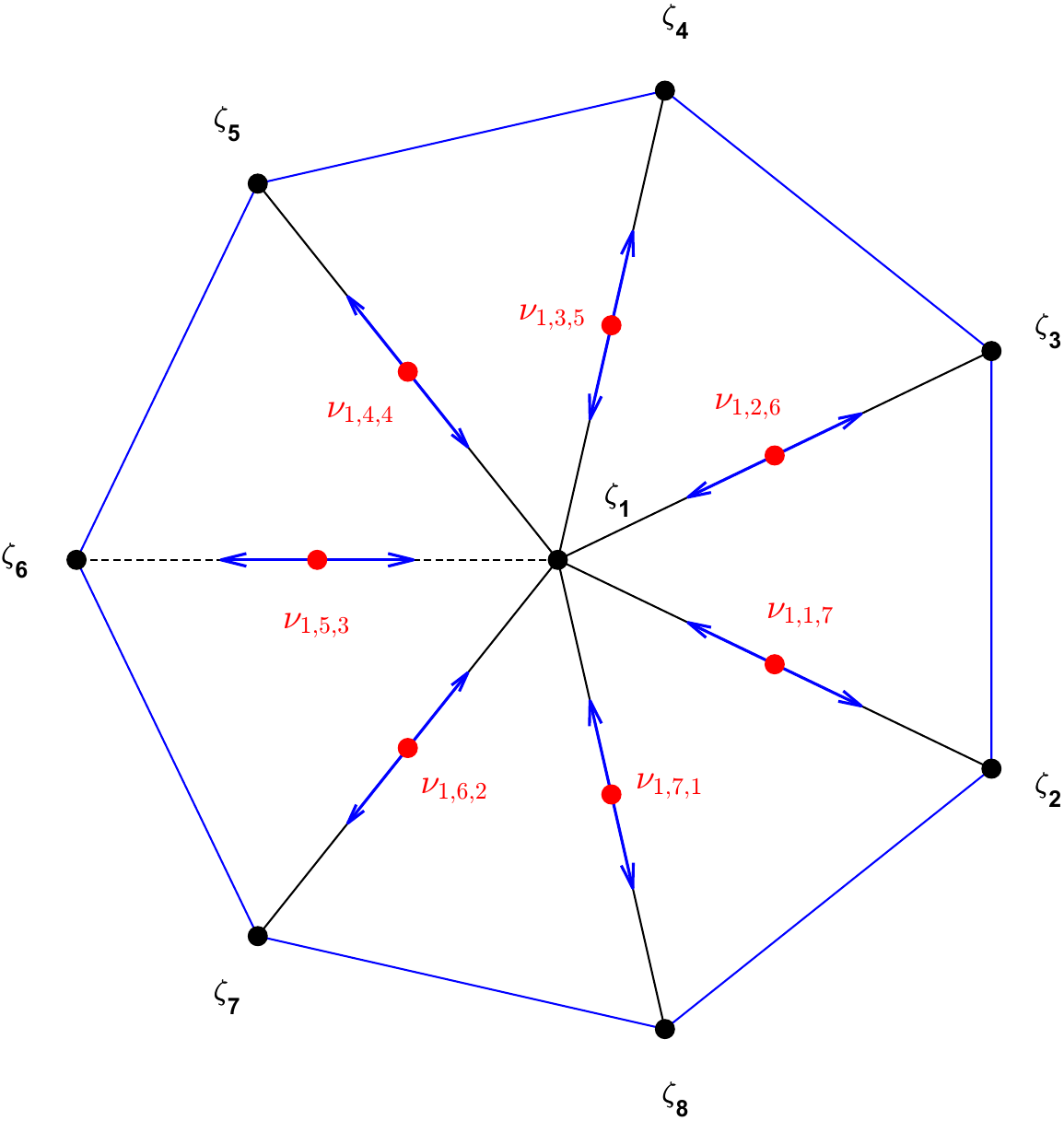}
    \caption{Classification of extreme points of $\nu_{1,k,8-k}$ where $k\in[1,7]$ with Remark \ref{rem:nu1,4,4}}
    \label{fig:fig1}
\end{figure}

\section{Extreme points with maximum eigenvalue of multiplicity two}
\label{sec:mu(a)=2}

In this section, we characterize the extreme points $\r$ of $\app_{3,3}$, when they have the maximum eigenvalue of multiplicity exactly two. Let $\r = \nu_{2,k,7-k} = \diag \{a,a,b,\dots,b,c,\dots,c\}$ where $k\in[1,6]$. The results are presented in Table \ref{tab:performance2}. In the rest of this section, we shall derive the results succinctly, as they are similar to the derivation of Table \ref{tab:performance1}. 

For brevity, while the remaining cases, i.e., (i)-(iii) and (v)-(vi), are presented in appendix \ref{other_of_tab2}.

(iv) We consider the boundary point 
\begin{eqnarray}
\label{nu2,4,3}
    \begin{aligned}
        \nu_{2,4,3} &= \diag \{a, a, b, b, b, b, c, c, c\} \\
        &= \diag \{\frac{1-4b-3c}{2}, \frac{1-4b-3c}{2}, b, b, b, b, c, c, c\}.
    \end{aligned}
\end{eqnarray}
Because $b>c$, we have two cases namely $l_1(\lambda)=0$ and $l_2(\lambda)=0$.

Firstly, if $l_1(\lambda)=0$ and $L_2(\lambda)$ is a positive semi-definite matrix, we assume $\nu_{2,4,3}=\nu_{2,4,3}^{(1)}$. We consider the system 
\begin{eqnarray}
\label{eq:nu2,4,3(1)}
    \begin{aligned}
        &\left|\begin{matrix}2c&\frac{4b+5c-1}{2}&\frac{6b+3c-1}{2}\\\frac{4b+5c-1}{2}&2c&0\\\frac{6b+3c-1}{2}&0&2b\\\end{matrix}\right| \\
        =&\frac{1}{2}[-16b^3+(8-76c)b^2+(-45c^2+22c-1)b-c(1-3c)^2]=0,
    \end{aligned}
\end{eqnarray}
\begin{eqnarray}
\label{neq:nu2,4,3(1)}
    L_2(\lambda)\geq O. 
\end{eqnarray}
By direct computation, using \eqref{nu2,4,3}, \eqref{eq:nu2,4,3(1)} and \eqref{neq:nu2,4,3(1)}, we obtain  
\begin{eqnarray}
    0.047619 \cdots=\frac{1}{21}<c\le\frac{85-14\sqrt{10}}{585}=0.069620\cdots. 
\end{eqnarray}
Furthermore, the value of $b$ is the second root of the cubic equation \eqref{eq:nu2,4,3(1)}, once $c$ is fixed. In addition, we consider the limit of the extreme point $\nu_{2,4,3}^{(1)}$. We obtain
\begin{eqnarray}
\lim_{c\rightarrow\frac{1}{21}} \nu_{2,4,3}^{(1)} = \zeta_6. 
\end{eqnarray}
When $c\ra\frac{85-14\sqrt{10}}{585}$, we will show below \eqref{eq:c->v243} that $l_1(\l)=l_2(\l)=0$, and $\n_{2,4,3}$ has three distinct eigenvalues in \eqref{eq:abc=nv243-1}-\eqref{eq:abc=nv243}.  

Secondly, if $l_2(\lambda)=0$ and $L_1(\lambda)$ is a positive semi-definite matrix, we assume $\nu_{2,4,3}=\nu_{2,4,3}^{(2)}$. We consider the system
\begin{eqnarray}
\label{eq:nu2,4,3(2)}
    \begin{aligned}
        &\left|\begin{matrix}2c&\frac{4b+5c-1}{2}&\frac{4b+5c-1}{2}\\\frac{4b+5c-1}{2}&2b&0\\\frac{4b+5c-1}{2}&0&2b\\\end{matrix}\right| \\
        =&-b[16b^2+8(4c-1)b+(1-5c)^2]=0,
    \end{aligned}
\end{eqnarray}
\begin{eqnarray}
\label{neq:nu2,4,3(2)}
    L_1(\lambda)\geq O. 
\end{eqnarray}
By direct computation and using \eqref{nu2,4,3}, \eqref{eq:nu2,4,3(2)} and \eqref{neq:nu2,4,3(2)}, we obtain that 
\begin{eqnarray}
    a=\frac{c+\sqrt{2c-9c^2}}{2}, \ 
    b=\frac{1-4c-\sqrt{2c-9c^2}}{4}
\end{eqnarray}
with
\begin{eqnarray}
    0.069620\cdots=\frac{85-14\sqrt{10}}{585}\le c <\frac{9-2\sqrt{2}}{73}=0.084542\cdots. 
\end{eqnarray}
In addition, we consider the limit of the extreme point $\nu_{2,4,3}^{(2)}$. We obtain
\begin{eqnarray}
\label{eq:c->v243}
\lim_{c\rightarrow\frac{9-2\sqrt{2}}{73}} \nu_{2,4,3}^{(2)} = \zeta_2. 
\end{eqnarray}
Moreover, from the above discussion one can also see that if $l_1(\lambda)=l_2(\lambda)=0$, there exists a unique solution
\begin{align}
\label{eq:abc=nv243-1}
    &a=\frac{1}{2}\left\{1+\frac{14\sqrt{10}-85}{195}-\frac{4\left[36-\frac{19}{117}(85-14\sqrt{10})\right]}{7(85-14\sqrt{10})}\right\}=0.189421\cdots, \\
    &b=\frac{36-\frac{19}{117}(85-14\sqrt{10})}{7(85-14\sqrt{10})}=0.103073\cdots, \\
    &c=\frac{85-14\sqrt{10}}{585}=0.069620\cdots.
    \label{eq:abc=nv243}
\end{align}
They generate a common extreme point $\nu_{2,4,3}^{(3)}=\diag(a,a,b,b,b,b,c,c,c)$ of the two systems \eqref{eq:nu2,4,3(1)}-\eqref{neq:nu2,4,3(1)} and \eqref{eq:nu2,4,3(2)}-\eqref{neq:nu2,4,3(2)}. 

Using the same technique as in (iv), the others are proven in appendix \ref{other_of_tab2} and we summarize them in Table \ref{tab:performance2}.

\newpage

\begin{table}[htbp]
\centering
\caption{Classification of extreme points of $\nu_{2,k,7-k}$ where $k\in[1,6]$ with Remark \ref{rem:table2}}
\label{tab:performance2}
\renewcommand{\arraystretch}{5}
\begin{tabular}{|l|c|c|c|c|}
\toprule
$\nu_{2,k,7-k}$ & a & b & range of c & limits \\ \hline

$\nu_{2,1,6}$ & $2c+\sqrt{12c^2-c}$ & $1-10c-2\sqrt{12c^2-c}$ & $0.083333\cdots = \frac{1}{12} < c < \frac{9-2\sqrt{2}}{73} = 0.084542\cdots$ & \shortstack{$\lim_{c\rightarrow\frac{9-2\sqrt{2}}{73}} \nu_{2,1,6} = \zeta_2$\\$\lim_{c\rightarrow\frac{1}{12}} \nu_{2,1,6} = \zeta_3$} \\ \hline

$\nu_{2,2,5}$ & $\frac{1-2b^{(1)}-5c}{2}$ & $b^{(1)}$ & $0.070805 \cdots=\frac{10 - \sqrt{17}}{83}<c<\frac{9-2\sqrt{2}}{73}=0.084542\cdots$ & \shortstack{$\lim_{c\rightarrow\frac{9-2\sqrt{2}}{73}} \nu_{2,2,5} = \zeta_2$\\$\lim_{c\rightarrow\frac{10 - \sqrt{17}}{83}} \nu_{2,2,5} = \zeta_4$} \\ \hline

$\nu_{2,3,4}$ & $\frac{1-3b^{(2)}-4c}{2}$ & $b^{(2)}$ & $0.056991\cdots=y<c<\frac{9-2\sqrt{2}}{73}=0.084542\cdots$ & \shortstack{$\lim_{c\rightarrow\frac{9-2\sqrt{2}}{73}} \nu_{2,3,4} = \zeta_2$\\$\lim_{c\rightarrow y} \nu_{2,3,4} = \zeta_5$} \\ \hline

$\nu_{2,4,3}^{(1)}$ & $\frac{1-4b^{(3)}-3c}{2}$ & $b^{(3)}$ & $0.047619 \cdots=\frac{1}{21}<c\le\frac{85-14\sqrt{10}}{585}=0.069620\cdots$ & $\lim_{c\rightarrow\frac{1}{21}} \nu_{2,4,3}^{(1)} = \zeta_6$ \\ 

$\nu_{2,4,3}^{(2)}$ & $\frac{c+\sqrt{2c-9c^2}}{2}$ & $\frac{1-4c-\sqrt{2c-9c^2}}{4}$ & $0.069620\cdots=\frac{85-14\sqrt{10}}{585}\le c <\frac{9-2\sqrt{2}}{73}=0.084542\cdots$ & $\lim_{c\rightarrow\frac{9-2\sqrt{2}}{73}} \nu_{2,4,3}^{(2)} = \zeta_2$ \\ 

$\nu_{2,4,3}^{(3)}$ & $1-4b^{(4)}-3c$ & \shortstack{$b^{(4)}:= $\\ $\frac{36-\frac{19}{117}(85-14\sqrt{10})}{7(85-14\sqrt{10})}$} & $\frac{85-14\sqrt{10}}{585}$ &  \\ \hline 

$\nu_{2,5,2}$ & $\frac{7 - 9 c + 5 \sqrt{-201 c^2 + 66 c - 1}}{37}$ & $\frac{6 - 13 c - \sqrt{-201 c^2 + 66 c - 1}}{37}$ & $0.023365 \cdots = \frac{23-14\sqrt{2}}{137}<c<\frac{9-2\sqrt{2}}{73}=0.084542\cdots$ & \shortstack{$\lim_{c\rightarrow\frac{9-2\sqrt{2}}{73}} \nu_{2,5,2} = \zeta_2$\\$\lim_{c\rightarrow\frac{23-14\sqrt{2}}{137}} \nu_{2,5,2} = \zeta_7$} \\ \hline

$\nu_{2,6,1}$ & $\frac{2 - 5c + 3\sqrt{4 c - 3 c^2}}{16}$ & $\frac{2 - c - \sqrt{4 c - 3 c^2}}{16}$ & $0<c<\frac{9-2\sqrt{2}}{73}=0.084542\cdots$ & \shortstack{$\lim_{c\rightarrow\frac{9-2\sqrt{2}}{73}} \nu_{2,6,1} = \zeta_2$\\$\lim_{c\rightarrow0} \nu_{2,6,1} = \zeta_8$} \\ \hline

\end{tabular}
\end{table}

\begin{remark}
\label{rem:table2}
Here, when $\lambda_6 \neq \lambda_7$, there are two conditions that may produce extreme points that as follows. $l_1(\lambda)=0$ and $L_2 \geq O$ ($\nu_{2,4,3}^{(1)}$); $l_2(\lambda)=0$ and $L_1 \geq O$ ($\nu_{2,4,3}^{(2)}$). Moreover, we have

1) $b^{(1)}$ of $\nu_{2,2,5}$ is the second root of the cubic equation $-4b^3+(4-30c)b^2+(-37c^2+14c-1)b-2c=0$. 

2) $b^{(2)}$ of $\nu_{2,3,4}$ is the second root of the cubic equation $-9b^3+(6-45c)b^2+(-56c^2+18c-1)b-c(1-6c)^2=0$ with $y<c<\frac{9-2\sqrt{2}}{73}$ where $y=0.056991\cdots$ is the second root of the equation $481y^3-37y^2-17y+1=0$.

3) $b^{(3)}$ of $\nu_{2,4,3}^{(1)}$ is the second root of the cubic equation $-16b^3+(8-76c)b^2+(-45c^2+22c-1)b-c(1-3c)^2=0$.    
\end{remark}

\newpage

\begin{figure}[htbp]
    \centering
    \includegraphics[width=0.8\textwidth]{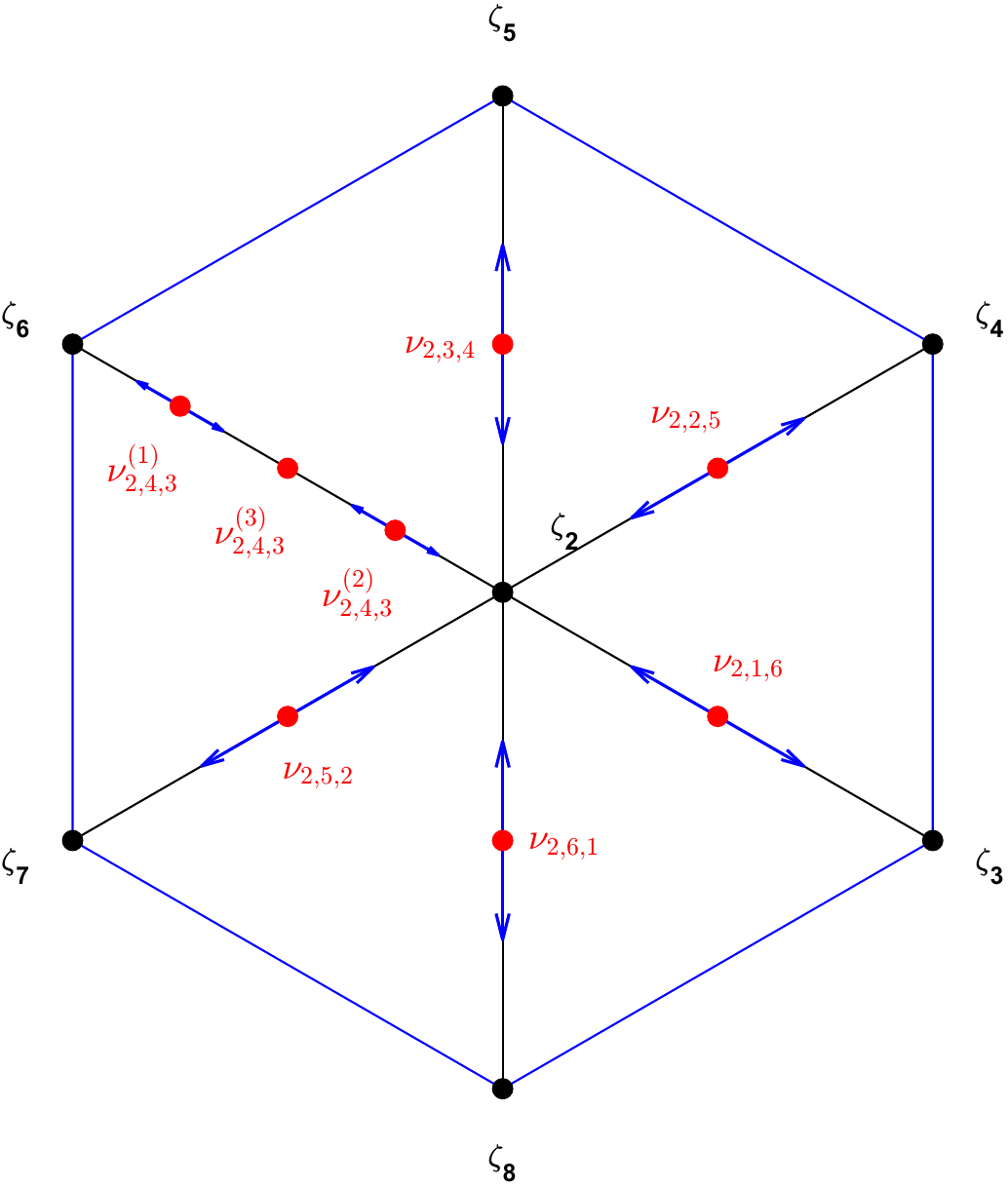}
    \caption{Classification of extreme points of $\nu_{2,k,7-k}$ where $k\in[1,6]$ with Remark \ref{rem:table2}}
    \label{fig:fig2}
\end{figure}

\section{Extreme points with maximum eigenvalue of multiplicity more than two}
\label{sec:mu(a)>2}



In this section, we classify extreme points of $\app_{3,3}$ when the maximum eigenvalue has the multiplicity more than two. Because the derivation is similar to those in the previous sections, we shall ignore the concrete derivation of most extreme points. We present our results on $\nu_{3,k,6-k}$ in Table \ref{tab:performance3}, $\nu_{4,k,5-k}$ in Table \ref{tab:performance4}, $\nu_{5,k,4-k}$ in Table \ref{tab:performance5}, $\nu_{6,k,3-k}$ in Table \ref{tab:performance6}, and $\nu_{7,1,1}$ in Table \ref{tab:performance7}.

\subsection{multiplicity three}
\label{sec:multi=three}

To illustrate Table \ref{tab:performance3} with a concrete example, we focus on the derivation of $\nu_{3,3,3}$. Similar to $\nu_{1,5,3}$ and $\nu_{2,4,3}$, the difficulty of deriving $\nu_{3,3,3}$ arises from the condition $\lambda_6 \neq \lambda_7$, yet its specific form is unique.
We consider the boundary point 
\begin{eqnarray}
\label{nu3,3,3}
    \begin{aligned}
        \nu_{3,3,3} &= \diag \{a, a, a, b, b, b, c, c, c\} \\
        &= \diag \{\frac{1-3b-3c}{3}, \frac{1-3b-3c}{3}, \frac{1-3b-3c}{3}, b, b, b, c, c, c\}.
    \end{aligned}
\end{eqnarray}
We have two cases, namely $l_1(\lambda)=0$, $l_2(\lambda)\ge0$ and $l_2(\lambda)=0$, $l_1(\lambda)\ge0$.

First if $l_1(\lambda)=0$ and $L_2(\lambda)$ is a positive semi-definite matrix, then we assume $\nu_{3,3,3}=\nu_{3,3,3}^{(1)}$. We consider the system 
\begin{eqnarray}
\label{eq:nu3,3,3(1)}
    \begin{aligned}
        &\left|\begin{matrix}2c&\frac{3b+6c-1}{3}&\frac{6b+3c-1}{3}\\\frac{3b+6c-1}{3}&2c&\frac{6b+3c-1}{3}\\\frac{6b+3c-1}{3}&\frac{6b+3c-1}{3}&2b\\\end{matrix}\right| \\
        =&\frac{2}{27} (3 b - 1) \left[27 b^2 - 9 b + (1 - 3 c)^2\right]=0,
    \end{aligned}
\end{eqnarray}
\begin{eqnarray}
\label{neq:nu3,3,3(1)}
    L_2(\lambda)\geq O. 
\end{eqnarray}
From \eqref{nu3,3,3}, we have 
\begin{eqnarray}
\label{neq:range_3,3,3(1)}
    \frac{1-3b-3c}{3}>b>c>0.
\end{eqnarray}
By solving the system of \eqref{eq:nu3,3,3(1)}-\eqref{neq:range_3,3,3(1)}, we obtain
\begin{eqnarray}
    a=\frac{3-18c+\sqrt{3(-1+24c-36c^2)}}{18}, \ 
    b=\frac{3-\sqrt{3(-1+24c-36c^2)}}{18}
\end{eqnarray}
with
\begin{eqnarray}
    0.047619 \cdots=\frac{1}{21} < c < \frac{1}{12} = 0.083333\cdots. \end{eqnarray}
In addition, one can verify that
\begin{eqnarray}
\lim_{c\rightarrow\frac{1}{12}} \nu_{3,3,3}^{(1)} = \zeta_3, \ 
\lim_{c\rightarrow\frac{1}{21}} \nu_{3,3,3}^{(1)} = \zeta_6. 
\end{eqnarray}  

Second if $l_2(\lambda)=0$ and $L_1(\lambda)$ is a positive semi-definite matrix, then we assume $\nu_{3,3,3}=\nu_{3,3,3}^{(2)}$. We consider the system
\begin{eqnarray}
\label{eq:nu3,3,3(2)}
    \begin{aligned}
        &\left|\begin{matrix}2c&\frac{3b+6c-1}{3}&\frac{3b+6c-1}{3}\\\frac{3b+6c-1}{3}&2b&\frac{6b+3c-1}{3}\\\frac{3b+6c-1}{3}&\frac{6b+3c-1}{3}&2b\\\end{matrix}\right| \\
        =&\frac{2}{27} (3 c - 1) (9 b^2 - 6 b + 27 c^2 - 9 c + 1)=0,
    \end{aligned}
\end{eqnarray}
\begin{eqnarray}
\label{neq:nu3,3,3(2)}
    L_1(\lambda)\geq O. 
\end{eqnarray}
However, the system consisting of \eqref{neq:range_3,3,3(1)}, \eqref{eq:nu3,3,3(2)}, and \eqref{neq:nu3,3,3(2)} has no solution. 

\newpage

\begin{table}[htbp]
\centering
\caption{Classification of extreme points of $\nu_{3,k,6-k}$ where $k\in[1,5]$ with Remark \ref{rem:table3}}
\label{tab:performance3}
\renewcommand{\arraystretch}{5}
\begin{tabular}{|l|c|c|c|c|}
\toprule
$\nu_{3,k,6-k}$ & a & b & range of c & limits \\ \hline

$\nu_{3,1,5}$ & $\frac{1 + 2 c - \sqrt{1 - 20 c + 132 c^2}}{4}$ & $\frac{1 - 26 c + 3\sqrt{1 - 20 c + 132 c^2}}{4}$ & $ 0.070805 \cdots=\frac{10 - \sqrt{17}}{83} < c < \frac{1}{12} = 0.083333\cdots$ & \shortstack{$\lim_{c\rightarrow\frac{1}{12}} \nu_{3,1,5} = \zeta_3 
$\\$\lim_{c\rightarrow\frac{10 - \sqrt{17}}{83}} \nu_{3,1,5} = \zeta_4$} \\ \hline

$\nu_{3,2,4}$ & $\frac{1-2b^{(1)}-4c}{4}$ & $b^{(1)}$ & $ y < c < \frac{1}{12} = 0.083333\cdots$ & \shortstack{$\lim_{c\rightarrow\frac{1}{12}} \nu_{3,2,4} = \zeta_3 
$\\$\lim_{c\rightarrow y} \nu_{3,2,4} = \zeta_5$} \\ \hline

$\nu_{3,3,3}^{(1)}$ & $\frac{3-18c+\sqrt{3(-1+24c-36c^2)}}{18}$ & $\frac{3-\sqrt{3(-1+24c-36c^2)}}{18}$ & $0.047619 \cdots=\frac{1}{21} < c < \frac{1}{12} = 0.083333\cdots$ & \shortstack{$\lim_{c\rightarrow\frac{1}{12}} \nu_{3,3,3}^{(1)} = \zeta_3 
$\\$\lim_{c\rightarrow\frac{1}{21}} \nu_{3,3,3}^{(1)} = \zeta_6$} \\ \hline

$\nu_{3,4,2}$ & $\frac{1-4b^{(2)}-2c}{4}$ & $b^{(2)}$ & $ 0.023365 \cdots = \frac{23-14\sqrt{2}}{137} < c < \frac{1}{12} = 0.083333\cdots$ & \shortstack{$\lim_{c\rightarrow\frac{1}{12}} \nu_{3,4,2} = \zeta_3 
$\\$\lim_{c\rightarrow\frac{23-14\sqrt{2}}{137}} \nu_{3,4,2} = \zeta_7$} \\ \hline

$\nu_{3,5,1}$ & $\frac{8-43c+5\sqrt{16c+33c^2}}{64}$ & $\frac{8+13c-3\sqrt{16c+33c^2}}{64}$ & $0 < c < \frac{1}{12} = 0.083333\cdots$ & \shortstack{$\lim_{c\rightarrow\frac{1}{12}} \nu_{3,5,1} = \zeta_3 
$\\$\lim_{c\rightarrow0} \nu_{3,5,1} = \zeta_8$} \\ \hline

\end{tabular}
\end{table}

\begin{remark}
\label{rem:table3}
When $\lambda_6 \neq \lambda_7$, there are two conditions that may produce extreme points, $l_1(\lambda)=0$ and $L_2 \geq O$; $l_2(\lambda)=0$ and $L_1 \geq O$. However, the latter system has no solution, while the former has a solution. We denote the solution as $\nu_{3,3,3}^{(1)}$. Moreover,

1) the value $b^{(1)}$ of $\nu_{3,2,4}$ is the first root
of the cubic equation $8b^3+(-12 - 15 c) b^2+ (6 - 30 c + 114 c^2)b-1 + 12 c -  39c^2 + c^3=0$ with $y<c<\frac{9-2\sqrt{2}}{73}$ where $y=0.056991\cdots$ is the second root of the equation $481y^3-37y^2-17y+1=0$.

2) The value $b^{(2)}$ of $\nu_{3,4,2}$ is the first root of the cubic equation $b^3+ (-39 + 114 c) b^2 +(12 - 30 c - 15 c^2) b -1 + 6 c - 12 c^2 + 8 c^3 =0$. 
\end{remark}

\newpage

\begin{figure}[htbp]
    \centering
    \includegraphics[width=0.8\textwidth]{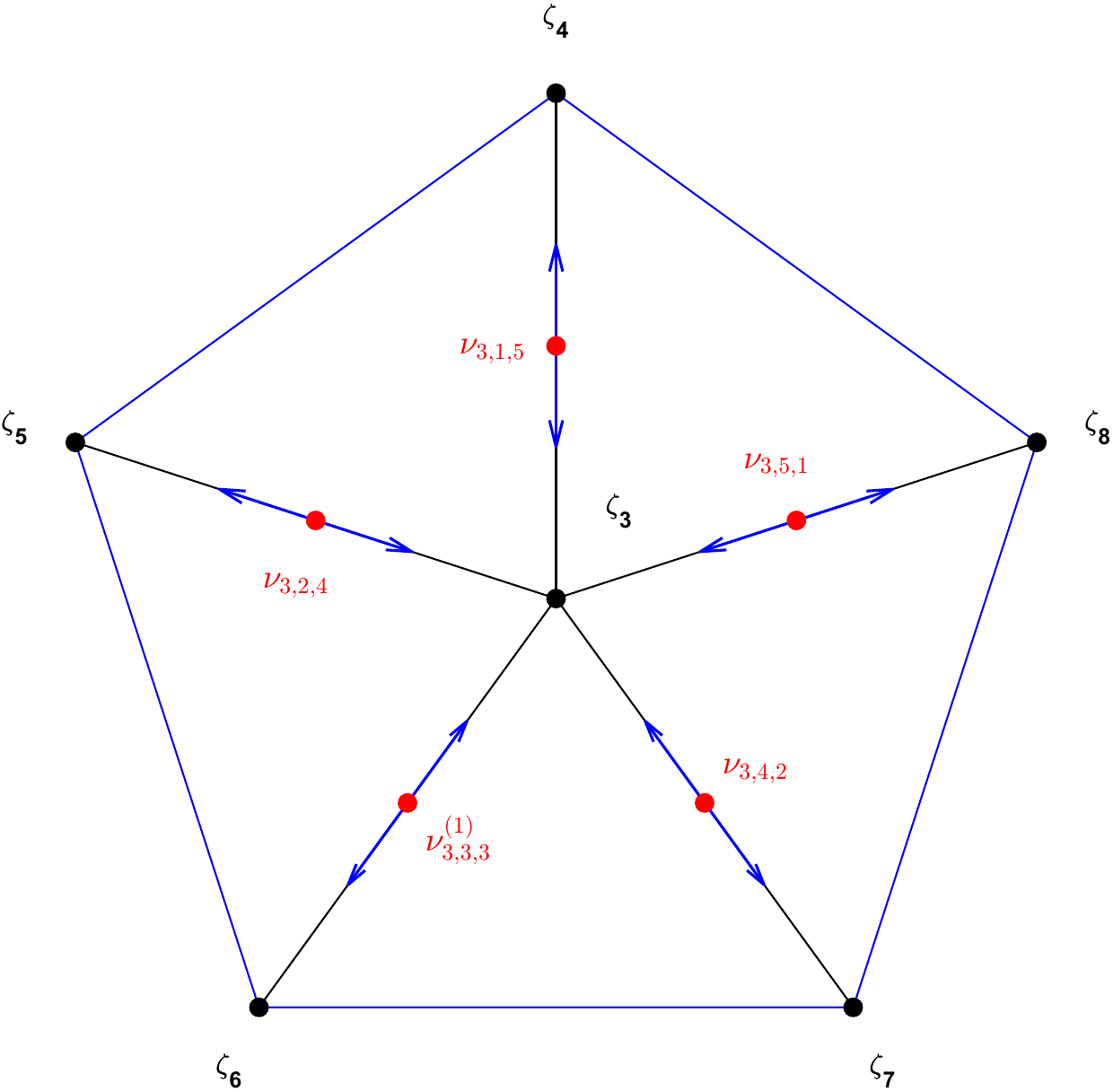}
    \caption{Classification of extreme points of $\nu_{3,k,6-k}$ where $k\in[1,5]$ with Remark \ref{rem:table3}.}
    \label{fig:fig3}
\end{figure}

\subsection{multiplicity four}
\label{sec:multi=four}

To illustrate Table \ref{tab:performance4} by a concrete example, we focus on the derivation of $\nu_{4,2,3}$. Similar to $\nu_{1,5,3}$, $\nu_{2,4,3}$ and $\nu_{3,3,3}$, the difficulty of deriving $\nu_{4,2,3}$ arises from the condition $\lambda_6 \neq \lambda_7$.
We consider the boundary point 
\begin{eqnarray}
\label{nu4,2,3}
    \begin{aligned}
        \nu_{4,2,3} &= \diag \{a, a, a, a, b, b, c, c, c\} \\
        &= \diag \{\frac{1-2b-3c}{4}, \frac{1-2b-3c}{4}, \frac{1-2b-3c}{4}, \frac{1-2b-3c}{4}, b, b, c, c, c\}.
    \end{aligned}
\end{eqnarray}

Firstly, if $l_1(\lambda)=0$ and $L_2(\lambda)$ is a positive semi-definite matrix, we assume $\nu_{4,2,3}=\nu_{4,2,3}^{(1)}$. We consider the system 
\begin{eqnarray}
\label{eq:nu4,2,3(1)}
    \begin{aligned}
        &\left|\begin{matrix}2c&\frac{2b+7c-1}{4}&\frac{6b+3c-1}{4}\\\frac{2b+7c-1}{4}&2c&\frac{6b+3c-1}{4}\\\frac{6b+3c-1}{4}&\frac{6b+3c-1}{4}&\frac{1-2b-3c}{2}\\\end{matrix}\right| \\
        =&\frac{1}{16} \left[40 b^3 + 4 (13 c - 9) b^2+ 2 (9 c^2 - 26 c + 5)b - 27 c^3 - 15 c^2 + 11 c - 1\right]=0,
    \end{aligned}
\end{eqnarray}
\begin{eqnarray}
\label{neq:nu4,2,3(1)}
    L_2(\lambda)\geq O. 
\end{eqnarray}
From \eqref{nu4,2,3}, we have 
\begin{eqnarray}
\label{neq:range_4,2,3(1)}
    \frac{1-2b-3c}{4}>b>c>0.
\end{eqnarray}
By solving the system of \eqref{eq:nu4,2,3(1)}-\eqref{neq:range_4,2,3(1)}, we obtain
\begin{eqnarray}
    a=\frac{3-6c+\sqrt{-1 + 24 c - 54 c^2}}{20}, \ b=\frac{2 - 9 c-\sqrt{-1 + 24 c - 54 c^2}}{10}
\end{eqnarray}
with
\begin{eqnarray}
    0.047619 \cdots=\frac{1}{21} < c < \frac{10 - \sqrt{17}}{83}=0.070805 \cdots. 
\end{eqnarray}
In addition, we have
\begin{eqnarray}
\lim_{c\rightarrow\frac{10 - \sqrt{17}}{83}} \nu_{4,2,3}^{(1)} = \zeta_4, \ 
\lim_{c\rightarrow \frac{1}{21}} \nu_{4,2,3}^{(1)} = \zeta_6. 
\end{eqnarray}  
Second, if $l_2(\lambda)=0$ and $L_1(\lambda)$ is a positive semi-definite matrix, then we assume $\nu_{4,2,3}=\nu_{4,2,3}^{(2)}$. We consider the system
\begin{eqnarray}
\label{eq:nu4,2,3(2)}
    \begin{aligned}
        &\left|\begin{matrix}2c&\frac{2b+7c-1}{4}&\frac{2b+7c-1}{4}\\\frac{2b+7c-1}{4}&2b&\frac{6b+3c-1}{4}\\\frac{2b+7c-1}{4}&\frac{6b+3c-1}{4}&\frac{1-2b-3c}{2}\\\end{matrix}\right| \\
        =&\frac{1}{16} \left[8 b^3 - 4 (17 c + 3) b^2 + 2 (7 c^2 - 6 c + 3)b + 129 c^3 - 79 c^2 + 15 c - 1\right]=0,
    \end{aligned}
\end{eqnarray}
\begin{eqnarray}
\label{neq:nu4,2,3(2)}
    L_1(\lambda)\geq O. 
\end{eqnarray}
However, the system consisting of \eqref{neq:range_4,2,3(1)}, \eqref{eq:nu4,2,3(2)}, and \eqref{neq:nu4,2,3(2)} has no solution. 

\newpage

\begin{table}
\centering
\caption{Classification of extreme points of $\nu_{4,k,5-k}$ where $k\in[1,4]$ with Remark \ref{rem:table4}}
\label{tab:performance4}
\renewcommand{\arraystretch}{5}
\begin{tabular}{|l|c|c|c|c|}
\toprule
$\nu_{4,k,5-k}$ & a & b & range of c & limits \\ \hline

$\nu_{4,1,4}$ & $\frac{1-b^{(1)}-4c}{4}$ & $b^{(1)}$ & $ y < c < \frac{10 - \sqrt{17}}{83}=0.070805 \cdots$ & \shortstack{$\lim_{c\rightarrow\frac{10 - \sqrt{17}}{83}} \nu_{4,1,4} = \zeta_4$
\\$\lim_{c\rightarrow y} \nu_{4,1,4} = \zeta_5$} \\ \hline

$\nu_{4,2,3}^{(1)}$ & $\frac{3-6c+\sqrt{-1 + 24 c - 54 c^2}}{20}$ & $\frac{2 - 9 c-\sqrt{-1 + 24 c - 54 c^2}}{10}$ & $0.047619 \cdots=\frac{1}{21} < c < \frac{10 - \sqrt{17}}{83}=0.070805 \cdots$ & \shortstack{$\lim_{c\rightarrow\frac{10 - \sqrt{17}}{83}} \nu_{4,2,3}^{(1)} = \zeta_4$
\\$\lim_{c\rightarrow \frac{1}{21}} \nu_{4,2,3}^{(1)} = \zeta_6$} \\ \hline

$\nu_{4,3,2}$ & $\frac{1-3b^{(2)}-2c}{4}$ & $b^{(2)}$ & $ 0.023365 \cdots = \frac{23-14\sqrt{2}}{137}  < c < \frac{10 - \sqrt{17}}{83}=0.070805 \cdots$ & \shortstack{$\lim_{c\rightarrow\frac{10 - \sqrt{17}}{83}} \nu_{4,3,2} = \zeta_4$
\\$\lim_{c\rightarrow \frac{23-14\sqrt{2}}{137} } \nu_{4,3,2} = \zeta_7$} \\ \hline

$\nu_{4,4,1}$ & $\frac{1-4b^{(3)}-c}{4}$ & $b^{(3)}$ & $0 < c < \frac{10 - \sqrt{17}}{83}=0.070805 \cdots$ & \shortstack{$\lim_{c\rightarrow\frac{10 - \sqrt{17}}{83}} \nu_{4,4,1} = \zeta_4$
\\$\lim_{c\rightarrow 0 } \nu_{4,4,1} = \zeta_8$} \\ \hline

\end{tabular}
\end{table}

\begin{remark}
\label{rem:table4}
When $\lambda_6 \neq \lambda_7$, there are two conditions that may produce extreme points, $l_1(\lambda)=0$ and $L_2 \geq O$; $l_2(\lambda)=0$ and $L_1 \geq O$. However, the latter system has no solution, while the former has a solution. We denote the solution as $\nu_{4,2,3}^{(1)}$. Moreover, we have

1) $b^{(1)}$ of $\nu_{4,1,4}$ is the first root of the cubic equation $ 3b^3  - 7b^2 + (5 - 48 c + 112 c^2)b -1 + 16 c - 48 c^2 - 32 c^3 
=0$ with $ y < c < \frac{10 - \sqrt{17}}{83}$ where $y=0.056991\cdots$ is the second root of the equation $481y^3-37y^2-17y+1=0$.

2) $b^{(2)}$ of $\nu_{4,3,2}$ is the second root of the cubic equation $ 11 b^3 + (31 - 2 c)b^2 + (-11 + 36 c - 52 c^2) b+1- 10 c + 36 c^2 - 
 40 c^3  =0$. 

 3) $b^{(3)}$ of $\nu_{4,4,1}$ is the first root of the cubic equation $ 256b^3+ (-128 - 128 c)b^2 + (20 + 24 c - 44 c^2)b -1 + c + c^2 - c^3 =0$. 
\end{remark}

\newpage

\begin{figure}[htbp]
    \centering
    \includegraphics[width=0.8\textwidth]{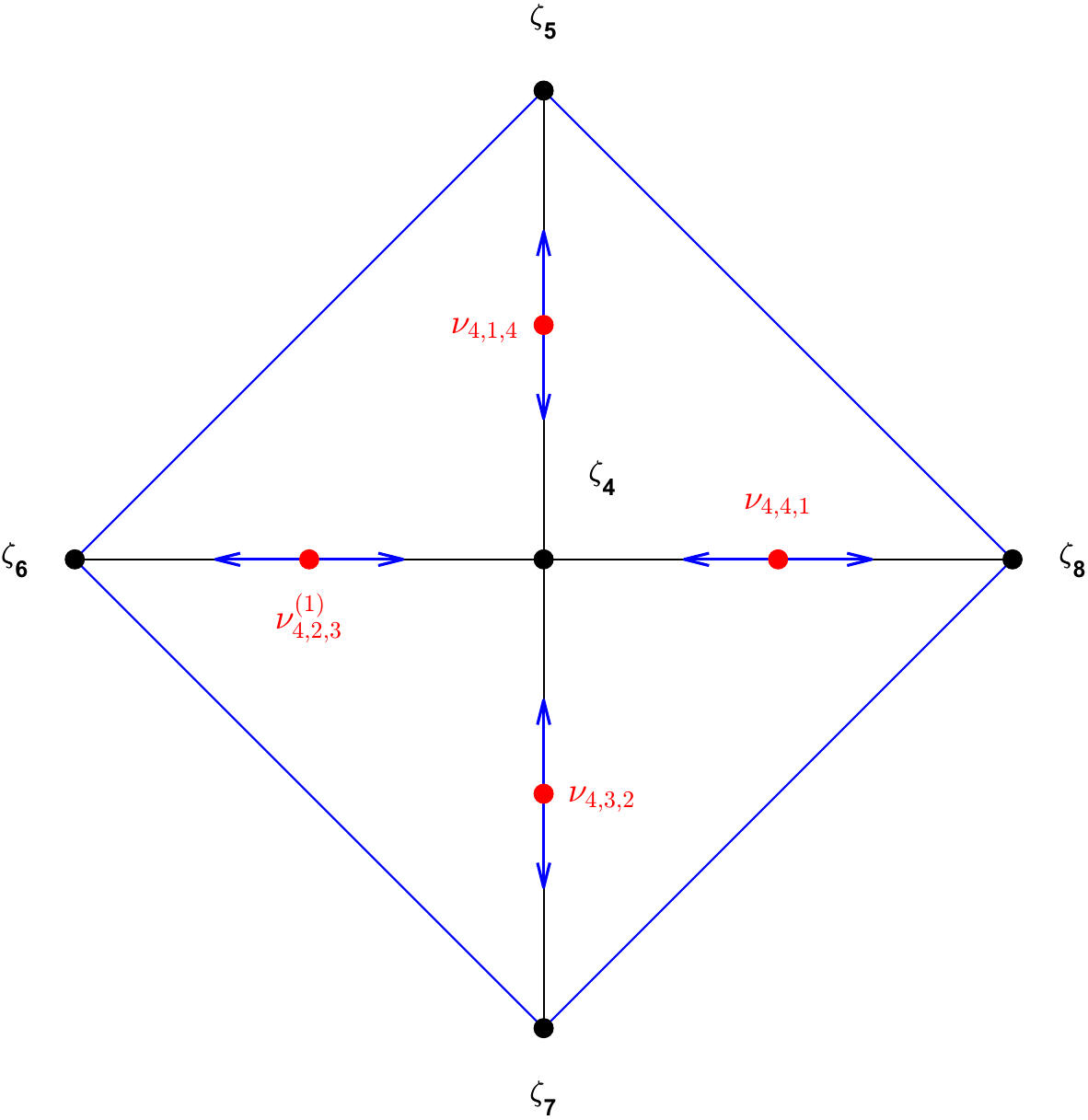}
    \caption{Classification of extreme points of $\nu_{4,k,5-k}$ where $k\in[1,4]$ with Remark \ref{rem:table4}}
    \label{fig:fig4}
\end{figure}

\subsection{multiplicity five}
\label{sec:multi=five}

To illustrate Table \ref{tab:performance5} by a concrete example, we focus on the derivation of $\nu_{5,1,3}$. Similar to $\nu_{1,5,3}$, $\nu_{2,4,3}$, $\nu_{3,3,3}$ and $\nu_{4,2,3}$, the difficulty of deriving $\nu_{5,1,3}$ arises from the condition $\lambda_6 \neq \lambda_7$.
We consider the boundary point 
\begin{eqnarray}
\label{nu5,1,3}
    \begin{aligned}
        \nu_{5,1,3} =& \diag \{a, a, a, a, a, b, c, c, c\} \\
        =& \diag \{\frac{1-b-3c}{5}, \cdots, \frac{1-b-3c}{5}, b, c, c, c\}.
    \end{aligned}
\end{eqnarray}
Firstly, if $l_1(\lambda)=0$ and $L_2(\lambda)$ is a positive semi-definite matrix, we assume $\nu_{5,1,3}=\nu_{5,1,3}^{(1)}$. We consider the system 
\begin{eqnarray}
\label{eq:nu5,1,3(1)}
    \begin{aligned}
        &\left|\begin{matrix}2c&\frac{b+8c-1}{5}&\frac{6b+3c-1}{5}\\\frac{b+8c-1}{5}&2c&0\\\frac{6b+3c-1}{5}&0&\frac{2(1-b-3c)}{5}\\\end{matrix}\right| =0,
    \end{aligned}
\end{eqnarray}
\begin{eqnarray}
\label{neq:nu5,1,3(1)}
    L_2(\lambda)\geq O. 
\end{eqnarray}
From \eqref{nu5,1,3}, we have 
\begin{eqnarray}
\label{neq:range_5,1,3(1)}
    \frac{1-b-3c}{5}>b>c>0.
\end{eqnarray}
By solving the system of \eqref{eq:nu5,1,3(1)}-\eqref{neq:range_5,1,3(1)}, we obtain
\begin{eqnarray}
    \frac{1}{21} < c < y 
\end{eqnarray}
where $y=0.056991\cdots$ is the second root of the equation $481y^3-37y^2-17y+1=0$. Furthermore, for a fixed $c$, $b$ is the first root of the cubic equation $ b^3  + (-3 - 161 c) b^2 + (3 + 22 c - 168 c^2) b -1 + 14 c + 18 c^2 - 153 c^3 =0$.

In addition, we have
\begin{eqnarray}
\lim_{c\rightarrow y} \nu_{5,1,3}^{(1)} = \zeta_5, \
\lim_{c\rightarrow \frac{1}{21}} \nu_{5,1,3}^{(1)} = \zeta_6. 
\end{eqnarray} 

Second, if $l_2(\lambda)=0$ and $L_1(\lambda)$ is a positive semi-definite matrix, then we assume $\nu_{5,1,3}=\nu_{5,1,3}^{(2)}$. We consider the system
\begin{eqnarray}
\label{eq:nu5,1,3(2)}
    \begin{aligned}
        &\left|\begin{matrix}2c&\frac{b+8c-1}{5}&\frac{b+8c-1}{5}\\\frac{b+8c-1}{5}&2b&0\\\frac{b+8c-1}{5}&0&\frac{2(1-b-3c)}{5}\\\end{matrix}\right|=0,
    \end{aligned}
\end{eqnarray}
\begin{eqnarray}
\label{neq:nu5,1,3(2)}
    L_1(\lambda)\geq O. 
\end{eqnarray}
However, the system consisting of \eqref{neq:range_5,1,3(1)}, \eqref{eq:nu5,1,3(2)}, and \eqref{neq:nu5,1,3(2)} has no solution.

\newpage

\begin{table}
    \centering
    \caption{Classification of extreme points of $\nu_{5,k,4-k}$ where $k\in[1,3]$ with Remark \ref{rem:table5}}
\label{tab:performance5}
    \renewcommand{\arraystretch}{5}
    \begin{tabular}{|l|c|c|c|c|}
    \toprule
    $\nu_{5,k,4-k}$ & a & b & range of c & limits \\ \hline

    $\nu_{5,1,3}^{(1)}$ & $\frac{1-b^{(1)}-3c}{4}$ & $ \quad b^{(1)} \quad$ & $ 0.047619 \cdots=\frac{1}{21} < c < y$ & \shortstack{$\lim_{c\rightarrow y} \nu_{5,1,3}^{(1)} = \zeta_5$\\$\lim_{c\rightarrow \frac{1}{21}} \nu_{5,1,3}^{(1)} = \zeta_6$} \\ \hline

    $\nu_{5,2,2}$ & $\frac{1-2b^{(2)}-2c}{4}$ & $b^{(2)}$ & $0.023365 \cdots = \frac{23-14\sqrt{2}}{137} < c < y$ & \shortstack{$\lim_{c\rightarrow y} \nu_{5,2,2}= \zeta_5$\\$\lim_{c\rightarrow \frac{23-14\sqrt{2}}{137} } \nu_{5,2,2} = \zeta_7$} \\ \hline

    $\nu_{5,3,1}$ & $\frac{1-3b^{(3)}-c}{4}$ & $b^{(3)}$ & $0< c < y$ & \shortstack{$\lim_{c\rightarrow y} \nu_{5,3,1}= \zeta_5$\\$\lim_{c\rightarrow 0 } \nu_{5,3,1} = \zeta_8$} \\ \hline
    
    \end{tabular}
    
\end{table}

\begin{remark}
\label{rem:table5}
When $\lambda_6 \neq \lambda_7$, there are two conditions that may produce extreme points, $l_1(\lambda)=0$ and $L_2 \geq O$; $l_2(\lambda)=0$ and $L_1 \geq O$. However, the latter system has no solution, while the former has a solution. We denote the solution as $\nu_{5,1,3}^{(1)}$. Moreover, we have

1) $b^{(1)}$ of $\nu_{5,1,3}^{(1)}$ is the first root of the cubic equation $ b^3  + (-3 - 161 c) b^2 + (3 + 22 c - 168 c^2) b -1 + 14 c + 18 c^2 - 153 c^3 =0$ with $\frac{1}{21} < c < y$ where $y=0.056991\cdots$ is the second root of the equation $481y^3-37y^2-17y+1=0$.

2) $b^{(2)}$ of $\nu_{5,2,2}$ is the second root of the cubic equation $ 237 b^3 + (-58 + 276 c)b^2+ (-1 - 56 c + 66 c^2)b +1 - 16 c + 77 c^2 - 98 c^3 =0$ with $\frac{23-14\sqrt{2}}{137} < c < y$ where $y=0.056991\cdots$ is the second root of the equation $481y^3-37y^2-17y+1=0$.

3) $b^{(3)}$ of $\nu_{5,3,1}$ is the second root of the cubic equation $128 b^3+ (32 + 268 c) b^2 +(-14 - 72 c + 86 c^2) b +1 - 3 c + 3 c^2 - c^3 =0$ with $\frac{23-14\sqrt{2}}{137} < c < y$ where $y=0.056991\cdots$ is the second root of the equation $481y^3-37y^2-17y+1=0$.
\end{remark}

\newpage

\begin{figure}[htbp]
    \centering
    \includegraphics[width=0.8\textwidth]{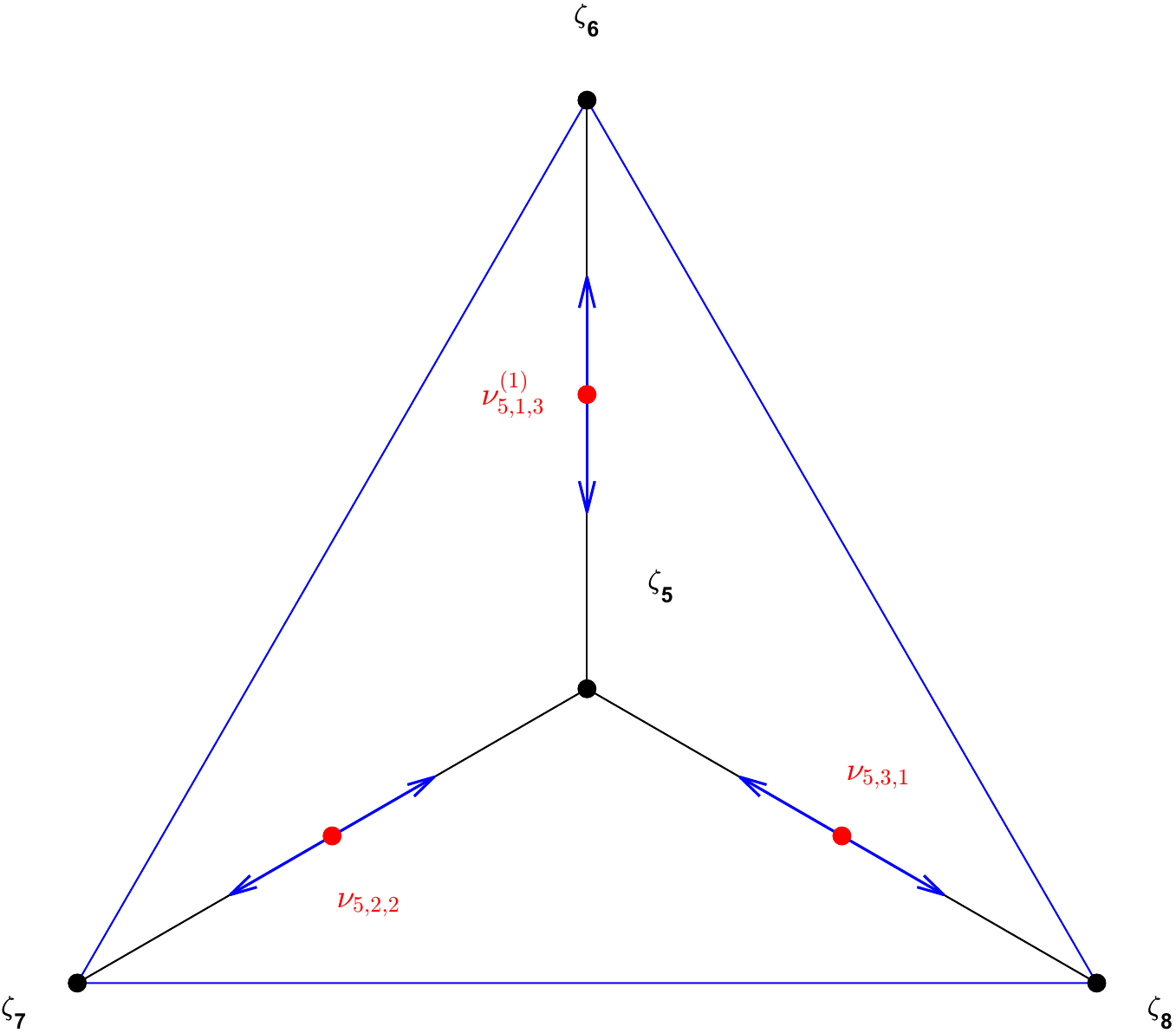}
    \caption{Classification of extreme points of $\nu_{5,k,4-k}$ where $k\in[1,3]$ with Remark \ref{rem:table5}}
    \label{fig:fig5}
\end{figure}

\subsection{multiplicity six}
\label{sec:multi=six}

To illustrate Table \ref{tab:performance6} by a concrete example, we focus on the derivation of $\nu_{6,1,2}$ and $\nu_{6,2,1}$. Similar to $\nu_{1,5,3}$, $\nu_{2,4,3}$, $\nu_{3,3,3}$, $\nu_{4,2,3}$ and $\nu_{5,1,3}$, the difficulty of deriving $\nu_{6,1,2}$ arises from the condition $\lambda_6 \neq \lambda_7$.
We consider the boundary point
\begin{eqnarray}
\label{nu6,1,2}
    \begin{aligned}
        \nu_{6,1,2} =& \diag \{a, a, a, a, a, a, b, c, c\} \\
        =& \diag \{\frac{1-b-2c}{6}, \cdots, \frac{1-b-2c}{6}, b, c, c\}.
    \end{aligned}
\end{eqnarray}

Firstly, if $l_1(\lambda)=0$ and $L_2(\lambda)$ is a positive semi-definite matrix, we assume $\nu_{6,1,2}=\nu_{6,1,2}^{(1)}$. We consider the system 
\begin{eqnarray}
\label{eq:nu6,1,2(1)}
    \begin{aligned}
        &\left|\begin{matrix}2c&\frac{b+8c-1}{6}&0\\\frac{b+8c-1}{6}&2b&0\\0&0&\frac{1-b-2c}{3}\\\end{matrix}\right|=0,
    \end{aligned}
\end{eqnarray}
\begin{eqnarray}
\label{neq:nu6,1,2(1)}
    L_2(\lambda)\geq O. 
\end{eqnarray}
From \eqref{nu6,1,2}, we have 
\begin{eqnarray}
\label{neq:range_6,1,2(1)}
    \frac{1-b-2c}{6}>b>c>0.
\end{eqnarray}
By solving the system of \eqref{eq:nu6,1,2(1)}-\eqref{neq:range_6,1,2(1)}, we obtain
\begin{eqnarray}
    a=-11c+2\sqrt{c + 28 c^2}, \
    b=1+64c-12\sqrt{c + 28 c^2}
\end{eqnarray}
with
\begin{eqnarray}
    0.023365 \cdots = \frac{23-14\sqrt{2}}{137} < c <\frac{1}{21}=0.047619 \cdots. 
\end{eqnarray}

In addition, we have
\begin{eqnarray}
\lim_{c\rightarrow \frac{1}{21}} \nu_{6,1,2}^{(1)} = \zeta_6, \
\lim_{c\rightarrow \frac{23-14\sqrt{2}}{137}} \nu_{6,1,2}^{(1)} = \zeta_7. 
\end{eqnarray} 

Second, if $l_2(\lambda)=0$ and $L_1(\lambda)$ is a positive semi-definite matrix, then we assume $\nu_{6,1,2}=\nu_{6,1,2}^{(2)}$. We consider the system
\begin{eqnarray}
\label{eq:nu6,1,2(2)}
    \begin{aligned}
        &\left|\begin{matrix}2c&\frac{b+8c-1}{6}&\frac{7b+2c-1}{6}\\\frac{b+8c-1}{6}&\frac{1-b-2c}{3}&0\\\frac{7b+2c-1}{6}&0&\frac{1-b-2c}{3}\\\end{matrix}\right|=0,
    \end{aligned}
\end{eqnarray}
\begin{eqnarray}
\label{neq:nu6,1,2(2)}
    L_1(\lambda)\geq O. 
\end{eqnarray}
However, the system consisting of \eqref{neq:range_6,1,2(1)}, \eqref{eq:nu6,1,2(2)}, and \eqref{neq:nu6,1,2(2)} has no solution.

The difficulty of deriving $\nu_{6,2,1}$ also arises from the condition $\lambda_6 \neq \lambda_7$.
We consider the boundary point 
\begin{eqnarray}
\label{nu6,2,1}
    \begin{aligned}
        \nu_{6,2,1} =& \diag \{a, a, a, a, a, a, b, b, c\} \\
        =& \diag \{\frac{1-2b-c}{6}, \cdots, \frac{1-2b-c}{6}, b, b, c\}.
    \end{aligned}
\end{eqnarray}

Firstly, if $l_1(\lambda)=0$ and $L_2(\lambda)$ is a positive semi-definite matrix, we assume $\nu_{6,2,1}=\nu_{6,2,1}^{(1)}$. We consider the system 
\begin{eqnarray}
\label{eq:nu6,2,1(1)}
    \begin{aligned}
        &\left|\begin{matrix}2c&\frac{8b+c-1}{6}&0\\\frac{8b+c-1}{6}&2b&0\\0&0&\frac{1-2b-c}{3}\\\end{matrix}\right| 
        =0,
    \end{aligned}
\end{eqnarray}
\begin{eqnarray}
\label{neq:nu6,2,1(1)}
    L_2(\lambda)\geq O. 
\end{eqnarray}
By direct computation and using \eqref{nu6,2,1}, \eqref{eq:nu6,2,1(1)} and \eqref{neq:nu6,2,1(1)}, we obtain that 
\begin{eqnarray}
    a=\frac{1-4c+\sqrt{2 c + 7 c^2}}{8}, \ 
    b=\frac{1+8c-3\sqrt{2 c + 7 c^2}}{8} 
\end{eqnarray}
with
\begin{eqnarray}
    0.017543 \cdots = \frac{1}{57} \le c <\frac{1}{21}=0.047619 \cdots. 
\end{eqnarray}
In addition, we consider the limit of the extreme point $\nu_{6,2,1}^{(1)}$. We obtain
\begin{eqnarray}
\lim_{c\rightarrow \frac{1}{21}} \nu_{6,2,1}^{(1)} = \zeta_6. 
\end{eqnarray}
When $c\ra\frac{1}{57}$, we will show below \eqref{eq:c->v621} that $l_1(\l)=l_2(\l)=0$, and $\n_{6,2,1}$ has three distinct eigenvalues in \eqref{eq:abc=nv621-1}-\eqref{eq:abc=nv621}.  

Secondly, if $l_2(\lambda)=0$ and $L_1(\lambda)$ is a positive semi-definite matrix, we assume $\nu_{6,2,1}=\nu_{6,2,1}^{(2)}$. We consider the system
\begin{eqnarray}
\label{eq:nu6,2,1(2)}
    \begin{aligned}
        &\left|\begin{matrix}2c&\frac{8b+c-1}{6}&\frac{8b+c-1}{6}\\\frac{8b+c-1}{6}&\frac{1-2b-c}{3}&0\\\frac{8b+c-1}{6}&0&\frac{1-2b-c}{3}\\\end{matrix}\right| 
        =0,
    \end{aligned}
\end{eqnarray}
\begin{eqnarray}
\label{neq:nu6,2,1(2)}
    L_1(\lambda)\geq O. 
\end{eqnarray}
By direct computation and using \eqref{nu6,2,1}, \eqref{eq:nu6,2,1(2)} and \eqref{neq:nu6,2,1(2)}, we obtain that 
\begin{eqnarray}
    a=\frac{2-c+\sqrt{4 c - 3 c^2}}{16}, \ 
    b=\frac{2-5c-3\sqrt{4 c - 3 c^2}}{16} 
\end{eqnarray}
with
\begin{eqnarray}
    0< c \le \frac{1}{57}=0.017543 \cdots. 
\end{eqnarray}
In addition, we consider the limit of the extreme point $\nu_{6,2,1}^{(2)}$. We obtain
\begin{eqnarray}
\label{eq:c->v621}
\lim_{c\rightarrow 0} \nu_{6,2,1}^{(2)} = \zeta_8. 
\end{eqnarray}
Moreover, from the above discussion one can also see that if $l_1(\lambda)=l_2(\lambda)=0$, there exists a unique solution
\begin{align}
\label{eq:abc=nv621-1}
    &a=\frac{8}{57}=0.140350 \cdots, \\
    &b=\frac{4}{57}=0.070175 \cdots, \\
    &c=\frac{1}{57}=0.017543 \cdots.
    \label{eq:abc=nv621}
\end{align}
They generate a common extreme point $\nu_{6,2,1}^{(3)}={1\over57}\diag(8,8,8,8,8,8,4,4,1)$ of the two systems \eqref{eq:nu6,2,1(1)}-\eqref{neq:nu6,2,1(1)} and \eqref{eq:nu6,2,1(2)}-\eqref{neq:nu6,2,1(2)}. 

\newpage

\begin{table}
    \centering
    \caption{Classification of extreme points of $\nu_{6,k,3-k}$ where $k\in[1,2]$ with Remark \ref{rem:table6}}
\label{tab:performance6}
    \renewcommand{\arraystretch}{5}
    \begin{tabular}{|l|c|c|c|c|}
    \toprule
    $\nu_{6,k,3-k}$ & a & b & range of c & limits \\ \hline

    $\nu_{6,1,2}^{(1)}$ & $-11c+2\sqrt{c + 28 c^2}$ & $ 1+64c-12\sqrt{c + 28 c^2}$ & $0.023365 \cdots = \frac{23-14\sqrt{2}}{137} < c <\frac{1}{21}=0.047619 \cdots$ & \shortstack{$\lim_{c\rightarrow \frac{1}{21}} \nu_{6,1,2}^{(1)} = \zeta_6$\\$\lim_{c\rightarrow \frac{23-14\sqrt{2}}{137}} \nu_{6,1,2}^{(1)} = \zeta_7$} \\ \hline

    $\nu_{6,2,1}^{(1)}$ & $\frac{1-4c+\sqrt{2 c + 7 c^2}}{8}$ & $ \frac{1+8c-3\sqrt{2 c + 7 c^2}}{8}$ & $0.017543 \cdots = \frac{1}{57} \le c <\frac{1}{21}=0.047619 \cdots$ & $\lim_{c\rightarrow \frac{1}{21}} \nu_{6,2,1}^{(1)} = \zeta_6$ \\ 

    $\nu_{6,2,1}^{(2)}$ & $\frac{2-c+\sqrt{4 c - 3 c^2}}{16}$ & $ \frac{2-5c-3\sqrt{4 c - 3 c^2}}{16}$ & $ 0< c \le \frac{1}{57}=0.017543 \cdots$ & $\lim_{c\rightarrow 0} \nu_{6,2,1}^{(2)} = \zeta_8$ \\ 

    $\nu_{6,2,1}^{(3)}$ & $\frac{8}{57}=0.140350 \cdots$ & $\frac{4}{57}=0.070175 \cdots$ & $\frac{1}{57}=0.017543 \cdots$ &  \\ \hline
    
    \end{tabular}
\end{table}

\begin{remark}
\label{rem:table6}
When $\lambda_6 \neq \lambda_7$, there are two conditions that may produce extreme points, $l_1(\lambda)=0$ and $L_2 \geq O$; $l_2(\lambda)=0$ and $L_1 \geq O$. 

For $\nu_{6,1,2}$, the latter system has no solution, while the former has a solution. We denote the solution as $\nu_{6,1,2}^{(1)}$. 

For $\nu_{6,2,1}$, there are two conditions that may produce extreme points: $l_1(\lambda)=0$ and $L_2 \geq O$ ($\nu_{6,2,1}^{(1)}$); $l_2(\lambda)=0$ and $L_1 \geq O$ ($\nu_{6,2,1}^{(2)}$).
\end{remark}

\newpage

\begin{figure}[htbp]
    \centering
    \includegraphics[width=0.8\textwidth]{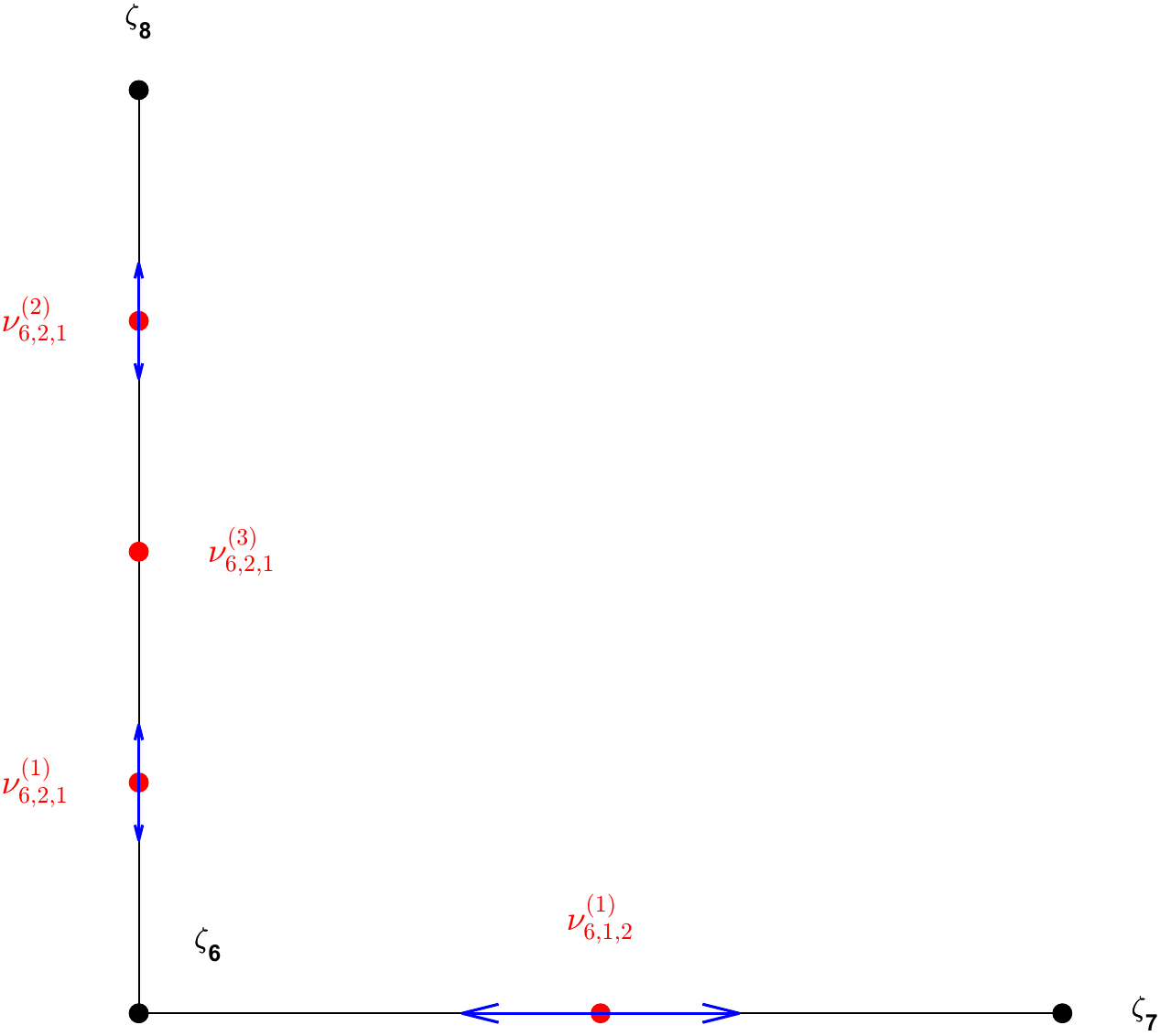}
    \caption{Classification of extreme points of $\nu_{6,k,3-k}$ where $k\in[1,2]$ with Remark \ref{rem:table6}}
    \label{fig:fig6}
\end{figure}

\subsection{multiplicity seven}
\label{sec:multi=seven}

\newpage

\begin{table}
    \centering
    \caption{Classification of extreme point $\nu_{7,1,1}$}
\label{tab:performance7}
    \renewcommand{\arraystretch}{5}
    \begin{tabular}{|l|c|c|c|c|}
    \toprule
    $\nu_{7,1,1}$ & a & b & range of c & limits \\ \hline

    $\nu_{7,1,1}$ & $\frac{4-3c+\sqrt{8 c - 7 c^2}}{32}$ & $ \frac{4-11c-7\sqrt{8 c - 7 c^2}}{32}$ & $0 < c <\frac{23-14\sqrt{2}}{137}=0.023365 \cdots $ & \shortstack{$\lim_{c\rightarrow \frac{23-14\sqrt{2}}{137}} \nu_{7,1,1} = \zeta_7$\\$\lim_{c\rightarrow0} \nu_{7,1,1} = \zeta_8$} \\ \hline
    \end{tabular}
\end{table}

\begin{figure}[htbp]
    \centering
    \includegraphics[width=0.8\textwidth]{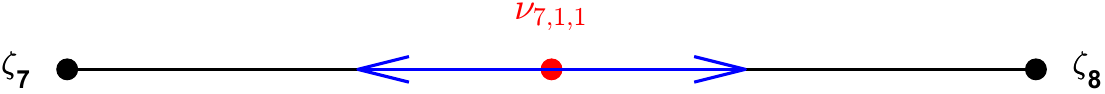}
    \caption{Classification of extreme point $\nu_{7,1,1}$}
    \label{fig:fig7}
\end{figure}

So far we have observed that for cases with three distinct eigenvalues, Lemma \ref{boun-extr} and corresponding methodology can be applied to determine whether a boundary point is an extreme point and, if so, to derive its extreme point form. Moreover, by Lemma \ref{AP3,3_judge_positive_semi-definite}, the necessary and sufficient condition for "belonging to $\mathcal{AP}_{3,3}$" can ignore the inequality \eqref{an inequality_2}. This simplification allows us to conveniently obtain a wide range of extreme points. While this method seems tedious, it is required by Lemma \ref{boun-extr}. Since each case involves the distinct behavior of each of its eigenvalue, a unified parametric treatment seems infeasible. Nevertheless, a more concise proof might be achieved by leveraging the one-to-one correspondence similar to what established in Lemma \ref{le:nu(1,5,3)}, as illustrated after this lemma. These arguments might be extended to the further characterization of two-qutrit extreme points of $\app_{3,3}$ with more than three distinct eigenvalues. We use an integrated three-dimensional Figure \ref{fig:umbrella} to represent the main results in this paper. 

\newpage
\begin{figure}[htbp]
    \centering
    \includegraphics[width=1.1\textwidth]{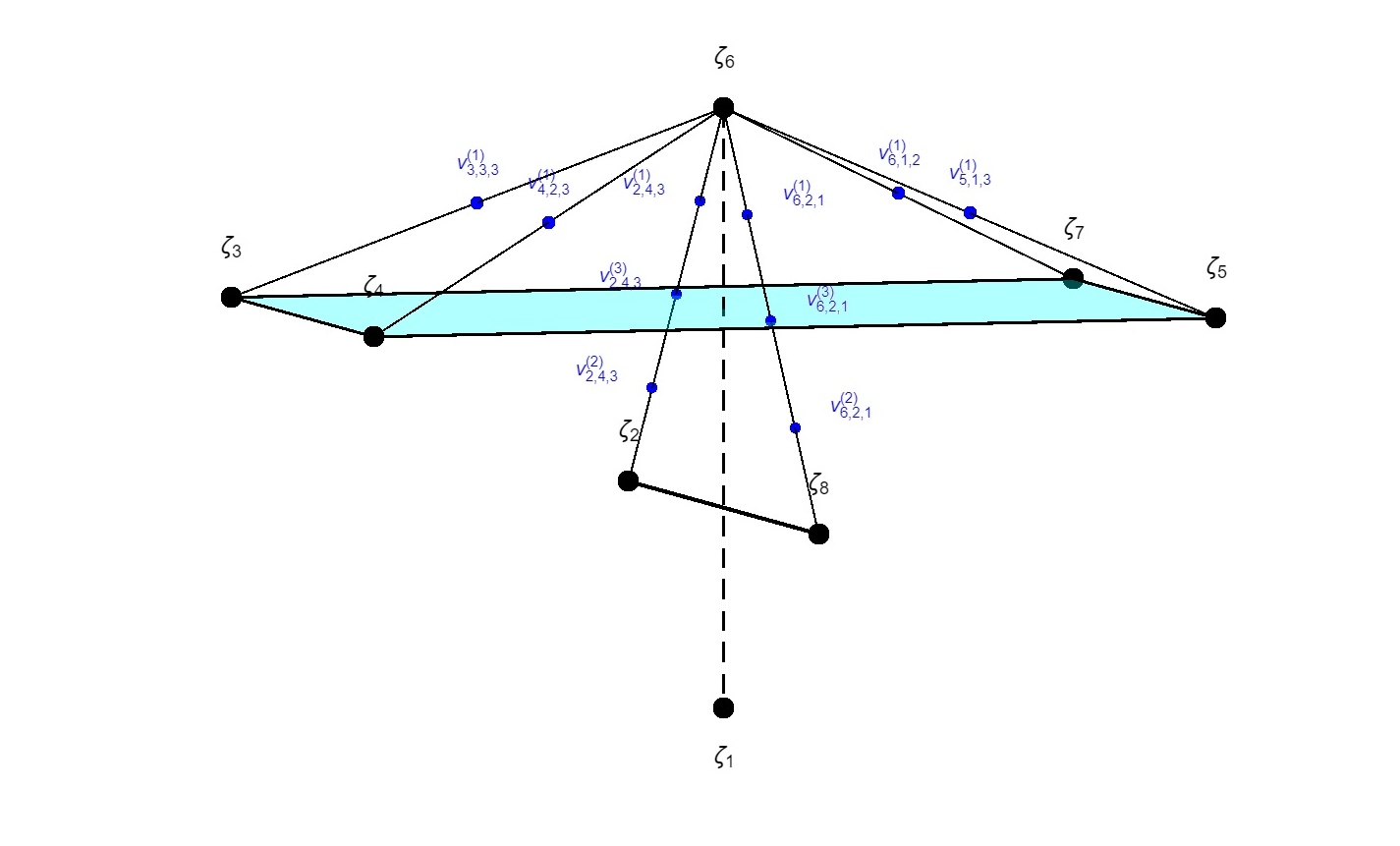}
    \caption{Umbrella Model. The black dots are extreme points with two distinct eigenvalues, while the blue dots are extreme points with three distinct eigenvalues, and there are no extreme points on the dashed line. The blue extreme points are obtained under the condition that $\lambda_6 \neq \lambda_7$, and therefore they all have superscripts. Between $\zeta_6$ and $\zeta_2$, and between $\zeta_6$ and $\zeta_8$, the eigenvalue $c$ has two different ranges, causing the extreme points to tend toward $\nu_{2,4,3}^{(3)}$ and $\nu_{6,2,1}^{(3)}$, respectively. This is a 3D schematic visualization of the hierarchical structure of the extreme points, which we term the 'Umbrella Model' due to its geometric shape.  }
    \label{fig:umbrella}
\end{figure}

\section{Conclusions}
\label{sec:con}

We have explicitly characterized all boundary and extreme points of two-qutrit AP states with exactly three distinct eigenvalues. Most of them contain parameters in intervals, and the extreme points at the two ends of intervals become some extreme points of exactly two eigenvalues. This work raises several open questions for future research. First, it is not hard to see that, if the points in the open subset of interval are separable, then so are the points at the ends of interval, by the fact that every PPT entangled state is detected by a non-decomposable entanglement witness. Is the converse true? Second, can we simplify the findings in this paper, especially those in the tables constructed for the expressions of extreme points and their limits? This is possible due to that many bounds in the tables appear similar. Third, can we extend the techniques and results of this paper to construct two-qutrit extreme AP states of more than three distinct eigenvalues? Fourth, what is the achievable upper bound of distinct eigenvalues of extreme two-qutrit AP states?  

\section*{Acknowledgments}
\label{sec:ack}	

Authors were supported by the NNSF of China (Grant No. 12471427), and the Fundamental Research Funds for the Central Universities (Grant Nos. ZG216S2110).

\appendix 
\section{ } 
\label{Appendix}

In the appendix we provide the detailed proof of Lemma \ref{le:3cbbbbbccc}, Theorem \ref{all_3eigen}, cases (i) and (iii)-(vii) in Table \ref{tab:performance1}, and other cases in Table \ref{tab:performance2}. We shall divide the proof into different section.

\subsection{The proof of Lemma \ref{le:3cbbbbbccc} }
\label{sec:abbcccccc}

(i) As mentioned before, the boundary point $\r \in \mathcal{AP}_{3,3}$ with $\lambda(\r) = (a, b, c, c, c, c, c, c, c)$ $(a>b>c>0)$ is an extreme point.

(ii) Given the boundary point 
\begin{eqnarray}
\label{eq:r_abbc}
    \lambda(\r) = (a, b, b, c, c, c, c, c, c) \quad (a>b>c>0),
\end{eqnarray}
we have $l_{1}(\lambda) = l_{2}(\lambda) = 0$ because $\lambda_6 = \lambda_7$. According to Lemma \ref{boun-extr}, to prove $\r$ in \eqref{eq:r_abbc} is an extreme point, it suffices to prove that the following equations have only the trivial solution:
\begin{eqnarray}
    t_{1} + 2t_{2} + 6t_{9} = 0, \ t_2=t_3, \ t_4 = \cdots =t_9, 
\end{eqnarray}
\begin{eqnarray}
\label{vec_u3_eq2b}
    \begin{bmatrix} 2t_{9} & t_{9} - t_{1} & t_{9} - t_{2} \\ t_{9} - t_{1} & 2t_{9} & t_{9} - t_{2} \\ t_{9} - t_{2} & t_{9} - t_{2} & 2t_{9} \end{bmatrix} \cdot \vec{u_3} = \begin{bmatrix} 0 \\ 0 \\ 0 \end{bmatrix}. 
\end{eqnarray}
Let $\vec{u_3} = \begin{bmatrix} u_{13} \\ u_{23} \\ u_{33} \end{bmatrix}$. Eq. \eqref{vec_u3_eq2b} can be reformulated as
\begin{eqnarray}
\label{vec_t_eq2b}
    \begin{bmatrix} -u_{23} & -u_{33} & 2u_{13} + u_{23} + u_{33} \\ -u_{13} & -u_{33} & u_{13} + 2u_{23} + u_{33} \\ 0 & -u_{13} -u_{23} & u_{13} + u_{23} + 2u_{33} \end{bmatrix} \cdot \begin{bmatrix} t_{1} \\ t_{2} \\ t_{9} \end{bmatrix} = \begin{bmatrix} 0 \\ 0 \\ 0 \end{bmatrix}.
\end{eqnarray}
Performing elementary matrix operations on the coefficient matrix of \eqref{vec_t_eq2b}, we obtain
\begin{align*}
     \begin{bmatrix} -u_{23} & -u_{33} & 2u_{13} + u_{23} + u_{33} \\ -u_{13} & -u_{33} & u_{13} + 2u_{23} + u_{33} \\ 0 & -u_{13} -u_{23} & u_{13} + u_{23} + 2u_{33} \end{bmatrix}
     &\xrightarrow[(1) + (3), (2) + (3)]{}
     \begin{bmatrix} -u_{23} & -u_{33} & 2u_{13} \\ -u_{13} & -u_{33} & 2u_{23} \\ 0 & -u_{13} -u_{23} & 2u_{33} \end{bmatrix}\\
     &\xrightarrow[\frac{1}{2} (3)]{}
     \begin{bmatrix} -u_{23} & -u_{33} & u_{13} \\ -u_{13} & -u_{33} & u_{23} \\ 0 & -u_{13} -u_{23} & u_{33} \end{bmatrix}. 
\end{align*}
We have two cases (a) and (b).

(a) $u_{33} \neq 0$. If at least one of $u_{13}$ and $u_{23}$ is not zero, the first and third columns of the coefficient matrix are linearly independent. If both $u_{13}$ and $u_{23}$ are zero, the second and third columns are linearly independent. Therefore, the rank of the coefficient matrix of $t_{1}, t_{2}, t_{9}$ is at least two.

(b) $u_{33} = 0$. The coefficient matrix above is $\begin{bmatrix} -u_{23} & 0 & u_{13} \\ -u_{13} & 0 & u_{23} \\ 0 & -u_{13} -u_{23} & 0 \end{bmatrix}$. Since $\vec{u_3}$ is a column vector of an orthogonal matrix, $u_{13}^2 + u_{23}^2 \neq 0$, which implies that the second and third columns are linearly independent if $u_{13} + u_{23} \neq 0$. Therefore, the rank of the coefficient matrix of $t_{1}, t_{2}, t_{9}$ is at least two unless $\vec{u_3} = \begin{bmatrix} \pm \frac{\sqrt2}{2} \\ \mp \frac{\sqrt2}{2}  \\ 0 \end{bmatrix}$. 

So in condition (a) and (b), using the rank-nullity theorem, the dimension of the solution space of $\begin{bmatrix} t_1 \\ t_2  \\ t_9 \end{bmatrix}$ is at most one. Hence the solution has the form $\begin{bmatrix} ka \\ kb  \\ kc \end{bmatrix}$ for any $k \in \mathbb{R}$, which means that $t_{1}, t_{2}, t_{9}$ must be all zero. In conclusion, if the boundary point $\r$ is unitarily equivalent to $\diag \{a, b, b, c, c, c, c, c, c\}$ in \eqref{eq:r_abbc} and $\vec{u_3} \neq \begin{bmatrix} \pm \frac{\sqrt2}{2} \\ \mp \frac{\sqrt2}{2}  \\ 0 \end{bmatrix}$ in \eqref{U-orth}, it is an extreme point. 

However, substituting $a, b, c$ and $\vec{u_3} = \begin{bmatrix} \pm \frac{\sqrt2}{2} \\ \mp \frac{\sqrt2}{2}  \\ 0 \end{bmatrix}$ into equation \eqref{U-orth} and computing the matrix multiplication, we obtain
\begin{eqnarray}
    \begin{bmatrix} 2c & c - a & c - b \\ c - a & 2c & c - b \\ c - b & c - b & 2c \end{bmatrix} \cdot \begin{bmatrix} \frac{\sqrt2}{2} \\ -\frac{\sqrt2}{2}  \\ 0 \end{bmatrix} = \begin{bmatrix} 0 \\ 0 \\ 0 \end{bmatrix}, 
\end{eqnarray}
which means that 
\begin{eqnarray}
    \begin{bmatrix} \frac{\sqrt2}{2}(a+c) \\ - \frac{\sqrt2}{2}(a+c)  \\ 0 \end{bmatrix} = \begin{bmatrix} 0 \\ 0 \\ 0 \end{bmatrix}.
\end{eqnarray}
This implies that $a = c = 0$, which leads to a contradiction. Similarly, this kind of discussion can be applied to $\vec{u_3} = \begin{bmatrix} - \frac{\sqrt2}{2} \\ \frac{\sqrt2}{2}  \\ 0 \end{bmatrix}$. Thus, $\vec{u_3} \neq \begin{bmatrix} \pm \frac{\sqrt2}{2} \\ \mp \frac{\sqrt2}{2}  \\ 0 \end{bmatrix}$ in \eqref{U-orth}. Therefore, If $\r$ is a boundary point of $\app_{3,3}$ with \eqref{eq:r_abbc}, then $\r$ is an extreme point of $\app_{3,3}$.

(iii) Given the boundary point 
\begin{eqnarray}
\label{eq:r_abbbc}
    \lambda(\r) = (a, b, b, b, c, c, c, c, c) \quad (a>b>c>0),
\end{eqnarray}
the condition implies that $l_{1}(\lambda) = l_{2}(\lambda) = 0$ because $\lambda_6 = \lambda_7$. According to Lemma \ref{boun-extr}, to prove $\r$ in \eqref{eq:r_abbbc} is an extreme point, it suffices to prove that the following equations have only the trivial solution:
\begin{eqnarray}
    t_{1} + 3t_{2} + 5t_{9} = 0, t_2 = t_3 = t_4, t_5 = \cdots =t_9, 
\end{eqnarray}
\begin{eqnarray}
\label{vec_u3_eq3b}
    \begin{bmatrix} 2t_{9} & t_{9} - t_{1} & t_{9} - t_{2} \\ t_{9} - t_{1} & 2t_{9} & t_{9} - t_{2} \\ t_{9} - t_{2} & t_{9} - t_{2} & 2t_{2} \end{bmatrix} \cdot \vec{u_3} = \begin{bmatrix} 0 \\ 0 \\ 0 \end{bmatrix}. 
\end{eqnarray}
Let $\vec{u_3} = \begin{bmatrix} u_{13} \\ u_{23} \\ u_{33} \end{bmatrix}$, \eqref{vec_u3_eq3b} can thus be reformulated as
\begin{eqnarray}
\label{vec_t_eq3b}
    \begin{bmatrix} -u_{23} & -u_{33} & 2u_{13} + u_{23} + u_{33} \\ -u_{13} & -u_{33} & u_{13} + 2u_{23} + u_{33} \\ 0 & -u_{13} -u_{23} + 2u_{33} & u_{13} + u_{23} \end{bmatrix} \cdot \begin{bmatrix} t_{1} \\ t_{2} \\ t_{9} \end{bmatrix} = \begin{bmatrix} 0 \\ 0 \\ 0 \end{bmatrix}.
\end{eqnarray}
Similarly, performing elementary matrix operations on the coefficient matrix of \eqref{vec_t_eq3b}, we obtain $\begin{bmatrix} -u_{23} & -u_{33} & u_{13} \\ -u_{13} & -u_{33} & u_{23} \\ 0 & -u_{13} -u_{23} + 2u_{33} & u_{33} \end{bmatrix}$. 

(a) $u_{33} \neq 0$. Similar to the last condition, the rank of the coefficient matrix is at least two.

(b) $u_{33} = 0$. The coefficient matrix above is $\begin{bmatrix} -u_{23} & 0 & u_{13} \\ -u_{13} & 0 & u_{23} \\ 0 & -u_{13} -u_{23} & 0 \end{bmatrix}$. Similarly, the rank of the coefficient matrix of $t_{1}, t_{2}, t_{9}$ is at least two unless $\vec{u_3} = \begin{bmatrix} \pm \frac{\sqrt2}{2} \\ \mp \frac{\sqrt2}{2}  \\ 0 \end{bmatrix}$.

So in condition (a) and (b), using the rank-nullity theorem, the dimension of the solution space of $\begin{bmatrix} t_1 \\ t_2  \\ t_9 \end{bmatrix}$ is at most one. Hence the solution has the form $\begin{bmatrix} ka \\ kb  \\ kc \end{bmatrix}$ for any $k \in \mathbb{R}$, which means that $t_{1}, t_{2}, t_{9}$ must be all zero. In conclusion, if the boundary point $\r$ is unitarily equivalent to $\diag \{a, b, b, b, c, c, c, c, c\}$ and $\vec{u_3} \neq \begin{bmatrix} \pm \frac{\sqrt2}{2} \\ \mp \frac{\sqrt2}{2}  \\ 0 \end{bmatrix}$ in \eqref{U-orth}, it is an extreme point. 

On the other hand, substituting $a, b, c$ and $\vec{u_3} = \begin{bmatrix} \pm \frac{\sqrt2}{2} \\ \mp \frac{\sqrt2}{2}  \\ 0 \end{bmatrix}$ into equation \eqref{U-orth} and computing the matrix multiplication, we obtain
\begin{eqnarray}
    \begin{bmatrix} 2c & c - a & c - b \\ c - a & 2c & c - b \\ c - b & c - b & 2b \end{bmatrix} \cdot \begin{bmatrix} \frac{\sqrt2}{2} \\ -\frac{\sqrt2}{2}  \\ 0 \end{bmatrix} = \begin{bmatrix} 0 \\ 0 \\ 0 \end{bmatrix}, 
\end{eqnarray}
which means that 
\begin{eqnarray}
    \begin{bmatrix} \frac{\sqrt2}{2}(a+c) \\ - \frac{\sqrt2}{2}(a+c)  \\ 0 \end{bmatrix} = \begin{bmatrix} 0 \\ 0 \\ 0 \end{bmatrix}.
\end{eqnarray}
This implies that $a = c = 0$, which leads to a contradiction. Similarly, this kind of discussion can be applied to $\vec{u_3} = \begin{bmatrix} - \frac{\sqrt2}{2} \\ \frac{\sqrt2}{2}  \\ 0 \end{bmatrix}$. Thus, $\vec{u_3} \neq \begin{bmatrix} \pm \frac{\sqrt2}{2} \\ \mp \frac{\sqrt2}{2}  \\ 0 \end{bmatrix}$ in \eqref{U-orth}. Therefore, If $\r$ is a boundary point of $\app_{3,3}$ with \eqref{eq:r_abbbc}, then $\r$ is an extreme point of $\app_{3,3}$.

(iv) Given the boundary point $\r$ in \eqref{eq:r_abbbbc}, it is an extreme point. 

(v) Given the boundary point 
\begin{eqnarray}
\label{eq:r_abbbbbc}
    \lambda(\r) = (a, b, b, b, b, b, c, c, c) \quad (a>b>c>0),
\end{eqnarray}
the condition implies that $l_{1}(\lambda) \neq l_{2}(\lambda)$ because $\lambda_6 \neq \lambda_7$. According to Lemma \ref{boun-extr}, there are two cases (v.1) $l_{1}(\lambda) = 0$ and (v.2) $l_{2}(\lambda) = 0$. We respectively discuss them as follows.

(v.1) Suppose $l_{1}(\lambda) = 0$. We consider the following linear equations in terms of $t_1, \cdots t_9$: 
\begin{eqnarray}
    t_{1} + 5t_{2} + 3t_{9} = 0, t_2 = \cdots t_6, t_7 = t_8 =t_9, 
\end{eqnarray}
\begin{eqnarray}
\label{vec_u3_eq5b}
    \begin{bmatrix} 2t_{9} & t_{9} - t_{1} & 0 \\ t_{9} - t_{1} & 2t_{9} & 0 \\ 0 & 0 & 2t_{2} \end{bmatrix} \cdot \vec{u_3} = \begin{bmatrix} 0 \\ 0 \\ 0 \end{bmatrix}. 
\end{eqnarray}
Let $\vec{u_3} = \begin{bmatrix} u_{13} \\ u_{23} \\ u_{33} \end{bmatrix}$, \eqref{vec_u3_eq5b} can thus be reformulated as
\begin{eqnarray}
\label{vec_t_eq5b}
    \begin{bmatrix} -u_{23} & 0 & 2u_{13} + u_{23} \\ -u_{13} & 0 & u_{13} + 2u_{23} \\ 0 & 2u_{33} & 0 \end{bmatrix} \cdot \begin{bmatrix} t_{1} \\ t_{2} \\ t_{9} \end{bmatrix} = \begin{bmatrix} 0 \\ 0 \\ 0 \end{bmatrix}.
\end{eqnarray}
Performing elementary matrix operations on the coefficient matrix of \eqref{vec_t_eq5b}, we obtain
\begin{align*}
     \begin{bmatrix} -u_{23} & 0 & 2u_{13} + u_{23} \\ -u_{13} & 0 & u_{13} + 2u_{23} \\ 0 & 2u_{33} & 0 \end{bmatrix}
     &\xrightarrow[(1) + (3)]{}
     \begin{bmatrix} -u_{23} & 0 & 2u_{13} \\ -u_{13} & 0 & 2u_{23} \\ 0 & 2u_{33} & 0 \end{bmatrix}\\
     &\xrightarrow[\frac{1}{2} (2), \frac{1}{2} (3)]{}
     \begin{bmatrix} -u_{23} & 0 & u_{13} \\ -u_{13} & 0 & u_{23} \\ 0 & u_{33} & 0 \end{bmatrix}. 
\end{align*}
We have  two cases (a) and (b).

(a) $u_{33} \neq 0$. If at least one of $u_{13}$ and $u_{23}$ is not zero, the first and second columns of the coefficient matrix are linearly independent. If both $u_{13}$ and $u_{23}$ are zero, the rank of the coefficient matrix of $t_{1}, t_{2}, t_{9}$ is one. Therefore, if $u_{33} \neq 0$ and $\vec{u_3} \neq \begin{bmatrix} 0 \\ 0 \\ \pm 1 \end{bmatrix}$ the rank of the coefficient matrix of $t_{1}, t_{2}, t_{9}$ is at least two.

(b) $u_{33} = 0$. The coefficient matrix above is $\begin{bmatrix} -u_{23} & 0 & u_{13} \\ -u_{13} & 0 & u_{23} \\ 0 & 0 & 0 \end{bmatrix}$. Therefore, if $u_{33} = 0$, $\vec{u_3} \neq \begin{bmatrix} \pm \frac{\sqrt2}{2} \\ \pm \frac{\sqrt2}{2}  \\ 0 \end{bmatrix}$ and $\vec{u_3} \neq \begin{bmatrix} \pm \frac{\sqrt2}{2} \\ \mp \frac{\sqrt2}{2}  \\ 0 \end{bmatrix}$, the first and third columns are linearly independent, and then the rank of the coefficient matrix of $t_{1}, t_{2}, t_{9}$ is at least two. 

So in case (a) and (b), using the rank-nullity theorem, the dimension of the solution space of $\begin{bmatrix} t_1 \\ t_2  \\ t_9 \end{bmatrix}$ is at most one. Hence the solution has the form $\begin{bmatrix} ka \\ kb  \\ kc \end{bmatrix}$ for any $k \in \mathbb{R}$, which $t_{1}, t_{2}, t_{9}$ must be all zero. In conclusion, if $l_{1}(\lambda) = 0$, and we exclude the following three exceptional cases in \eqref{U-orth},  
\begin{eqnarray}
\label{eq:v=three exceptional cases}    
\vec{u_3} \neq \begin{bmatrix} 0 \\ 0 \\ \pm 1 \end{bmatrix}, \ \begin{bmatrix} \pm \frac{\sqrt2}{2} \\ \mp \frac{\sqrt2}{2}  \\ 0 \end{bmatrix}, \ \text{and} \begin{bmatrix} \pm \frac{\sqrt2}{2} \\ \pm \frac{\sqrt2}{2}  \\ 0 \end{bmatrix},
\end{eqnarray}
then the boundary point $\r$ in \eqref{eq:r_abbbbbc} is an extreme point. 

Next we investigate the three exceptional cases in \eqref{eq:v=three exceptional cases}. First, substituting $a, b, c$ and $\vec{u_3} = \begin{bmatrix} 0 \\ 0  \\ \pm1 \end{bmatrix}$ into equation \eqref{U-orth} and computing the matrix multiplication, we obtain
\begin{eqnarray}
    \begin{bmatrix} 2c & c - a & 0 \\ c - a & 2c & 0 \\ 0 & 0 & 2b \end{bmatrix} \cdot \begin{bmatrix} 0 \\ 0  \\ 1 \end{bmatrix} = \begin{bmatrix} 0 \\ 0 \\ 0 \end{bmatrix}, 
\end{eqnarray}
which means that 
\begin{eqnarray}
    \begin{bmatrix} 0 \\ 0  \\ 2b \end{bmatrix} = \begin{bmatrix} 0 \\ 0 \\ 0 \end{bmatrix}.
\end{eqnarray}
This implies that $b = 0$, which leads to a contradiction. Similarly, this kind of discussion can be applied to $\vec{u_3} = \begin{bmatrix} 0 \\ 0  \\ -1 \end{bmatrix}$. Thus, $\vec{u_3} \neq \begin{bmatrix} 0 \\ 0 \\ \pm 1 \end{bmatrix}$ in \eqref{U-orth}. 

Next, we study the second exceptional point in \eqref{eq:v=three exceptional cases}. Substituting $a, b, c$ and $\vec{u_3} = \begin{bmatrix} \pm \frac{\sqrt2}{2} \\ \mp \frac{\sqrt2}{2}  \\ 0 \end{bmatrix}$ into equation \eqref{U-orth} and computing the matrix multiplication, we obtain
\begin{eqnarray}
    \begin{bmatrix} 2c & c - a & 0 \\ c - a & 2c & 0 \\ 0 & 0 & 2b \end{bmatrix} \cdot \begin{bmatrix} \frac{\sqrt2}{2} \\ -\frac{\sqrt2}{2}  \\ 0 \end{bmatrix} = \begin{bmatrix} 0 \\ 0 \\ 0 \end{bmatrix}, 
\end{eqnarray}
which means that 
\begin{eqnarray}
    \begin{bmatrix} \frac{\sqrt2}{2}(a+c) \\ - \frac{\sqrt2}{2}(a+c)  \\ 0 \end{bmatrix} = \begin{bmatrix} 0 \\ 0 \\ 0 \end{bmatrix}.
\end{eqnarray}
This implies that $a = c = 0$, which leads to a contradiction. Similarly, this kind of discussion can be applied to $\vec{u_3} = \begin{bmatrix} - \frac{\sqrt2}{2} \\ \frac{\sqrt2}{2}  \\ 0 \end{bmatrix}$. Thus, $\vec{u_3} \neq \begin{bmatrix} \pm \frac{\sqrt2}{2} \\ \mp \frac{\sqrt2}{2}  \\ 0 \end{bmatrix}$ in \eqref{U-orth}. 

We have excluded the existence of $\r$ in \eqref{eq:r_abbbbbc} in the first two exceptional cases in \eqref{eq:v=three exceptional cases}. It remains to study the third case. Substituting $a, b, c$ and $\vec{u_3} = \begin{bmatrix} \pm \frac{\sqrt2}{2} \\ \pm \frac{\sqrt2}{2}  \\ 0 \end{bmatrix}$ into equation \eqref{U-orth} and computing the matrix multiplication, we obtain
\begin{eqnarray}
\label{eq:third exceptional case}
    \begin{bmatrix} 2c & c - a & 0 \\ c - a & 2c & 0 \\ 0 & 0 & 2b \end{bmatrix} \cdot \begin{bmatrix} \frac{\sqrt2}{2} \\ \frac{\sqrt2}{2}  \\ 0 \end{bmatrix} = \begin{bmatrix} 0 \\ 0 \\ 0 \end{bmatrix}, 
\end{eqnarray}
which means that 
\begin{eqnarray}
    \begin{bmatrix} \frac{3\sqrt2}{2}c - \frac{\sqrt2}{2}a \\ - \frac{\sqrt2}{2}a + \frac{3\sqrt2}{2}c  \\ 0 \end{bmatrix} = \begin{bmatrix} 0 \\ 0 \\ 0 \end{bmatrix}.
\end{eqnarray}
This implies that 
\begin{eqnarray}
\label{eq:a=3c}    
a= 3c.
\end{eqnarray}
In the following, we verify whether both $L_1(\lambda)$ in \eqref{U-orth} and $L_2(\lambda)$ in \eqref{V-orth} are positive semi-definite matrices under this condition.

Firstly, $L_1(\lambda) = \begin{bmatrix} 2c & c - 3c & 0 \\ c - 3c & 2c & 0 \\ 0 & 0 & 2b \end{bmatrix} = \begin{bmatrix} 2c & -2c & 0 \\ -2c & 2c & 0 \\ 0 & 0 & 2b \end{bmatrix}$, which means that the characteristic polynomial of $L_1(\lambda)$ is 
\begin{eqnarray}
    \left|x I_3 - L_1(\lambda)\right| = \left|\begin{matrix}x-2c&2c&0\\2c&x-2c&0\\0&0&x-2b\\\end{matrix}\right| 
    = (x - 2b)(x - 4c)x, 
\end{eqnarray}
where $I_3$ is a $3 \times 3$ identity matrix. Thus, there exists an order-three real orthogonal matrix $U$ s.t. 
\begin{eqnarray}
    U^TL_1(\lambda)U = \left[\begin{matrix}2b&&\\&4c&\\&&0\\\end{matrix}\right]. 
\end{eqnarray}
So $L_1(\lambda)$ is a positive semi-definite matrix. 

Secondly, $L_2(\lambda) = \begin{bmatrix} 2c & c - 3c & c - b \\ c - 3c & 2b & 0 \\ c -b & 0 & 2b \end{bmatrix} = \begin{bmatrix} 2c & -2c & c - b\\ -2c & 2b & 0 \\ c - b & 0 & 2b \end{bmatrix}$, which means that the characteristic polynomial of $L_2(\lambda)$ is
\begin{align*}
    \left|x I_3 - L_2(\lambda)\right| &=
    \left|\begin{matrix}x-2c&2c&b-c\\2c&x-2b&0\\b-c&0&x-2b\\\end{matrix}\right|\\
    &=(x - 2b)[x^2 - 2(b + c)x -b^2 + 6bc - 5c^2],
\end{align*}
where $I_3$ is a $3 \times 3$ identity matrix. So the three eigenvalues of $L_2(\lambda)$ are $2b, \ b+c \pm \sqrt{2[(b-c)^2+2c^2]}$. Thus, there exists an order-three real orthogonal matrix $V$ s.t. 
\begin{eqnarray}
    V^TL_2(\lambda)V = \left[\begin{matrix}2b&&\\&b+c + \sqrt{2[(b-c)^2+2c^2]}&\\&&b+c - \sqrt{2[(b-c)^2+2c^2]}\\\end{matrix}\right]. 
\end{eqnarray}
In this case, $L_2(\lambda)$ is a positive semi-definite matrix if and only if $b+c - \sqrt{2[(b-c)^2+2c^2]} \geq 0$ that means $c \le b \le 5c$. But we also have $a>b>c>0$ and $a=3c$ in \eqref{eq:a=3c}. So $L_2(\lambda)$ a positive semi-definite matrix. To conclude, we have shown that $\r$ in \eqref{eq:r_abbbbbc} is a two-qutrit AP state. It remains to show whether the boundary $\r$ is an extreme point. 

Substituting the specific $\vec{u_3}$ into the above equation \eqref{vec_u3_eq5b}, we obtain the following linear equations in terms of $t_1, \cdots t_9$, 
\begin{eqnarray}
    t_{1} + 5t_{2} + 3t_{9} = 0, t_2 = \cdots t_6, t_7 = t_8 =t_9, 
\end{eqnarray}
\begin{eqnarray}
    \begin{bmatrix} 2t_{9} & t_{9} - t_{1} & 0 \\ t_{9} - t_{1} & 2t_{9} & 0 \\ 0 & 0 & 2t_{2} \end{bmatrix} \cdot \begin{bmatrix} \frac{\sqrt2}{2} \\ \frac{\sqrt2}{2}  \\ 0 \end{bmatrix} = \begin{bmatrix} 0 \\ 0 \\ 0 \end{bmatrix}. 
\end{eqnarray}
Thus $t_1 = 3t_9, \ t_2 = -\frac{6}{5}t_9$ i.e. the solution has a form of $\begin{bmatrix} 15 \\ -6  \\ 5 \end{bmatrix}k, \ k \in \mathbb{R}$. Using Lemma \ref{boun-extr}, the state $\r$ in \eqref{eq:r_abbbbbc} is not an extreme point. To give a more clear counter-example, we assume
\begin{eqnarray}
\label{eq:alpha}
    \alpha = \diag\{3c+15\epsilon, b-6\epsilon, b-6\epsilon, b-6\epsilon, b-6\epsilon, b-6\epsilon, c+5\epsilon, c+5\epsilon, c+5\epsilon\},
\end{eqnarray}
\begin{eqnarray}
\label{eq:beta}
    \beta = \diag\{3c-15\epsilon, b+6\epsilon, b+6\epsilon, b+6\epsilon, b+6\epsilon, b+6\epsilon, c-5\epsilon, c-5\epsilon, c-5\epsilon\}, 
\end{eqnarray}
where $\epsilon$ is sufficiently small s.t. the diagonal elements (eigenvalues) of $\alpha, \ \beta$ above are still arranged in descending order and consist of three different eigenvalues. According to 
Lemma \ref{AP3,3_judge_positive_semi-definite}, we have
$\alpha, \ \beta \in \app_{3,3}$. For the state $\alpha$, the characteristic polynomial of $L_1(\lambda)$ is 
\begin{align*}
     \left|x I_3 - L_1(\lambda)\right| 
     &= \left|\begin{matrix}x-2(c+5\epsilon)&2(c+5\epsilon)&0\\2(c+5\epsilon)&x-2(c+5\epsilon)&0\\0&0&x-2(b-6\epsilon)\\\end{matrix}\right| ,
\end{align*}
where $I_3$ is a $3 \times 3$ identity matrix. Thus, there exists an order-three real orthogonal matrix $U$ s.t. 
\begin{eqnarray}
    U^TL_1(\lambda)U = \left[\begin{matrix}2(b-6\epsilon)&&\\&4(c+5\epsilon)\epsilon&\\&&0\\\end{matrix}\right]. 
\end{eqnarray}
Then the characteristic polynomial of $L_2(\lambda)$ is 
\begin{align*}
    \left|x I_3 - L_2(\lambda)\right| 
    &=\left|\begin{matrix}x-2(c+5\epsilon)&2(c+5\epsilon)&b-c-11\epsilon\\2(c+5\epsilon)&x-2(b-6\epsilon)&0\\b-c-11\epsilon&0&x-2(b-6\epsilon)\\\end{matrix}\right|,
\end{align*}
where $I_3$ is a $3 \times 3$ identity matrix. Thus, there exists an order-three real orthogonal matrix $V$ s.t. 
\begin{eqnarray}
    V^TL_2(\lambda)V = \left[\begin{matrix}\lambda_1&&\\&\lambda_2&\\&&\lambda_3\\\end{matrix}\right],
\end{eqnarray}
where $\lambda_1, \lambda_2, \lambda_3$ are 
\begin{align*}
    &2(b-6\epsilon), \\
    &b+c-\epsilon + \sqrt{2[(b-c)^2+2c^2]+2(171\epsilon^2-22b\epsilon+42c\epsilon)}, \\
    &b+c-\epsilon - \sqrt{2[(b-c)^2+2c^2]+2(171\epsilon^2-22b\epsilon+42c\epsilon)},
\end{align*}
respectively. Therefore, by appropriately choosing a sufficiently small $\epsilon$, $L_1(\lambda), L_2(\lambda)$ are both positive semi-definite matrices. For the state $\beta$ in \eqref{eq:beta}, we simply change the sign in front of $\epsilon$ and the rest of the proof follows similarly. $L_1(\lambda), L_2(\lambda)$ are both positive semi-definite matrices, too. 

Up to unitary similarity, we have $\r = \frac{1}{2}(\alpha+\beta)$ by \eqref{eq:alpha} and \eqref{eq:beta}. The two-qutrit AP state $\r$ in the third exceptional case in \eqref{eq:v=three exceptional cases} is not an extreme point. 

(v.2) Suppose $l_2(\lambda) = 0$. We consider the following linear equations in terms of $t_1, \cdots t_9$: 
\begin{eqnarray}
    t_{1} + 5t_{2} + 3t_{9} = 0, t_2 = \cdots t_6, t_7 = t_8 =t_9, 
\end{eqnarray}
\begin{eqnarray}
\label{vec_u3_eq5b'}
    \begin{bmatrix} 2t_{9} & t_{9} - t_{1} & t_{9} - t_{2} \\ t_{9} - t_{1} & 2t_{2} & 0 \\ t_{9} - t_{2} & 0 & 2t_{2} \end{bmatrix} \cdot \vec{v_3} = \begin{bmatrix} 0 \\ 0 \\ 0 \end{bmatrix}. 
\end{eqnarray}
Setting $\vec{v_3} = \begin{bmatrix} v_{13} \\ v_{23} \\ v_{33} \end{bmatrix}$, \eqref{vec_u3_eq5b'} can thus be reformulated as
\begin{eqnarray}
\label{vec_t_eq5b'}
    \begin{bmatrix} -v_{23} & -v_{33} & 2v_{13} + v_{23} + v_{33} \\ -v_{13} & 2v_{23} & v_{13} \\ 0 & -v_{13} + 2v_{33} & v_{13} \end{bmatrix} \cdot \begin{bmatrix} t_{1} \\ t_{2} \\ t_{9} \end{bmatrix} = \begin{bmatrix} 0 \\ 0 \\ 0 \end{bmatrix}.
\end{eqnarray}
Performing elementary matrix operations on the coefficient matrix of \eqref{vec_t_eq5b'}, we obtain
\begin{align*}
     \begin{bmatrix} -v_{23} & -v_{33} & 2v_{13} + v_{23} + v_{33} \\ -v_{13} & 2v_{23} & v_{13} \\ 0 & -v_{13} + 2v_{33} & v_{13} \end{bmatrix}
     &\xrightarrow[(1) + (3), (2) + (3)]{}
     \begin{bmatrix} -v_{23} & -v_{33} & 2v_{13} \\ -v_{13} & 2v_{23} & 2v_{23} \\ 0 & -v_{13} + 2v_{33} & 2v_{33} \end{bmatrix}\\
     &\xrightarrow[\frac{1}{2} (2), \frac{1}{2} (3)]{}
     \begin{bmatrix} -v_{23} & -v_{33} & v_{13} \\ -v_{13} & 2v_{23} & v_{23} \\ 0 & -v_{13} + v_{33} & 2v_{33} \end{bmatrix}. 
\end{align*}
(a) $v_{33} \neq 0$. If at least one of $v_{13}$ and $v_{23}$ is not zero, the first and third columns of the coefficient matrix are linearly independent. If both $v_{13}$ and $v_{23}$ are zero i.e. the coefficient matrix is $\begin{bmatrix} 0 & -v_{33} & 0 \\ 0 & 0 & 0 \\ 0 & v_{33} & 2v_{33} \end{bmatrix}$, the second and third columns of the coefficient matrix are linearly independent. Therefore,  the rank of the coefficient matrix of $t_{1}, t_{2}, t_{9}$ is at least two.

(b) $u_{33} = 0$. The coefficient matrix above is $\begin{bmatrix} -v_{23} & 0 & v_{13} \\ -v_{13} & 2v_{23} & v_{23} \\ 0 & -v_{13} & 0 \end{bmatrix}$. Obviously, the rank of the coefficient matrix of $t_{1}, t_{2}, t_{9}$ is at least two. 

So in condition (a) and (b), using the rank-nullity theorem, the dimension of the solution space of $\begin{bmatrix} t_1 \\ t_2  \\ t_9 \end{bmatrix}$ is at most one. Hence the solution has the form $\begin{bmatrix} ka \\ kb  \\ kc \end{bmatrix}$ for any $k \in \mathbb{R}$, which means $t_{1}, t_{2}, t_{9}$ must be all zero. In conclusion, if $l_{2}(\lambda) = 0$ and the boundary point $\r$ is unitarily equivalent to $\diag \{a, b, b, b, b, b, c, c, c\}$, $\r$ is an extreme point.

(vi) Given the boundary point 
\begin{eqnarray}
\label{eq:r_a6bc}
    \lambda(\r) = (a, b, b, b, b, b, b, c, c) \quad (a>b>c>0),
\end{eqnarray}
we consider the equations
\begin{eqnarray}
    t_{1} + 6t_{2} + 2t_{9} = 0, t_2 \cdots = t_7, t_8 =t_9, 
\end{eqnarray}
\begin{eqnarray}
\label{vec_u3_eq6b}
    \begin{bmatrix} 2t_{9} & t_{9} - t_{1} & 0 \\ t_{9} - t_{1} & 2t_{2} & 0 \\ 0 & 0 & 2t_{2} \end{bmatrix} \cdot \begin{bmatrix} u_{13} \\ u_{23} \\ u_{33} \end{bmatrix} = \begin{bmatrix} 0 \\ 0 \\ 0 \end{bmatrix},
\end{eqnarray}
it can be reformulated as
\begin{eqnarray}
\label{vec_t_eq6b}
    \begin{bmatrix} -u_{23} & 0 & 2u_{13} + u_{23} \\ -u_{13} & 2u_{23} & u_{13} \\ 0 & 2u_{33} & 0 \end{bmatrix} \cdot \begin{bmatrix} t_{1} \\ t_{2} \\ t_{9} \end{bmatrix} = \begin{bmatrix} 0 \\ 0 \\ 0 \end{bmatrix}.
\end{eqnarray}
Similarly, performing elementary matrix operations on the coefficient matrix of \eqref{vec_t_eq6b}, we obtain $\begin{bmatrix} -u_{23} & 0 & u_{13} \\ -u_{13} & 2u_{23} & u_{23} \\ 0 & 2u_{33} & u_{33} \end{bmatrix}$. 

In this condition, if the boundary point $\r$ is unitarily equivalent to $\diag \{a, b, b, b, b, b, b, c, c\}$ and $\vec{u_3} \neq \begin{bmatrix} 0 \\ 0  \\ \pm 1 \end{bmatrix}$ in \eqref{U-orth}, it is an extreme point. 

However, substituting $a, b, c$ and $\vec{u_3} = \begin{bmatrix} 0 \\ 0  \\ \pm 1 \end{bmatrix}$ into equation \eqref{U-orth} and computing the matrix multiplication, we obtain $b = 0$, which leads to a contradiction. Therefore, If $\r$ is a boundary point of $\app_{3,3}$ with \eqref{eq:r_a6bc}, then $\r$ is an extreme point of $\app_{3,3}$.

(vii) Given the boundary point 
\begin{eqnarray}
\label{eq:r_a7bc}
    \lambda(\r) = (a, b, b, b, b, b, b, b, c) \quad (a>b>c>0),
\end{eqnarray}
we consider the equation
\begin{eqnarray}
    t_{1} + 7t_{2} + t_{9} = 0, t_2 \cdots = t_8,
\end{eqnarray}
\begin{eqnarray}
\label{vec_u3_eq7b}
    \begin{bmatrix} 2t_{9} & t_{2} - t_{1} & 0 \\ t_{2} - t_{1} & 2t_{2} & 0 \\ 0 & 0 & 2t_{2} \end{bmatrix} \cdot \begin{bmatrix} u_{13} \\ u_{23} \\ u_{33} \end{bmatrix} = \begin{bmatrix} 0 \\ 0 \\ 0 \end{bmatrix},
\end{eqnarray}
it can be reformulated as
\begin{eqnarray}
\label{vec_t_eq7b}
    \begin{bmatrix} -u_{23} & u_{23} & 2u_{13} \\ -u_{13} & u_{13} + 2u_{23} & 0 \\ 0 & 2u_{33} & 0 \end{bmatrix} \cdot \begin{bmatrix} t_{1} \\ t_{2} \\ t_{9} \end{bmatrix} = \begin{bmatrix} 0 \\ 0 \\ 0 \end{bmatrix}.
\end{eqnarray}
Similarly, performing elementary matrix operations on the coefficient matrix of \eqref{vec_t_eq7b}, we obtain $\begin{bmatrix} -u_{23} & 0 & 2u_{13} \\ -u_{13} &  2u_{23} & 0 \\ 0 & 2u_{33} & 0 \end{bmatrix}$. 
In this condition, if the boundary point $\r$ is unitarily equivalent to $\diag \{a, b, b, b, b, b, b, b, c\}$ and $\vec{u_3} \neq \begin{bmatrix} 0 \\ 0  \\ \pm 1 \end{bmatrix}$ in \eqref{U-orth}, it is an extreme point. Substituting $a, b, c$ and $\vec{u_3} = \begin{bmatrix} 0 \\ 0  \\ \pm 1 \end{bmatrix}$ into equation \eqref{U-orth} and computing the matrix multiplication, we obtain $b = 0$, which leads to a contradiction. Therefore, If $\r$ is a boundary point of $\app_{3,3}$ with \eqref{eq:r_a7bc}, then $\r$ is an extreme point of $\app_{3,3}$.

To conclude, for full-rank $\r \in \mathcal{AP}_{3,3}$ with three distinct eigenvalues whose largest eigenvalue is simple, if it is a boundary point of $\app_{3,3}$, then it is an extreme point of $\app_{3,3}$ unless $\lambda(\r) = (3c, b, b, b, b, b, c, c, c) \quad (3c>b>c>0)$.

\subsection{The proof of Theorem \ref{all_3eigen} }
\label{sec:all exe points of 3 distinct eigenvalues}
\begin{proof}
The technique described above can be applied directly to show that the rank of the coefficient matrix is at least two. In a few exceptional cases where it is less than two, the results would imply that $a$, $b$, or $c$ must be zero, which leads to a contradiction. All conditions are as follows. Without loss of generality, we explore the diagonal states. That is, we shall respectively analyze the cases $a=1,2,...,7$ in the following. 

(i) $\nu_{1,k,8-k}$ where $k=1, \cdots , 7$. The boundary point $\r$ is an extreme point of $\app_{3,3}$ unless $\r $ is unitarily equivalent to $ \nu_{1,5,3}$ in Lemma \ref{le:nu(1,5,3)}. 

(ii) $\nu_{2,k,7-k}$ where $k=1, \cdots , 6$. The spectrum vector of $\r$ is 
\begin{eqnarray}
\label{eq:rho216}
    \lambda(\r) = (a, a, b, c, c, c, c, c, c) \quad (a>b>c>0).
\end{eqnarray}
Using an approach similar to that of Lemma \ref{le:nu(1,5,3)}, we obtain a linear system in $t_1$, $t_3$, and $t_9$, whose coefficient matrix has rank at least two, unless $\vec{u_3} = \begin{bmatrix} 0 \\ \pm \frac{\sqrt2}{2}  \\ \mp \frac{\sqrt2}{2} \end{bmatrix}$. But this implies that $b = c = 0$, which leads to a contradiction. Consequently, coefficient matrix always has rank at least two and then only the trivial solution exists. Thus for $\r$ in \eqref{eq:rho216}, the conclusion of the theorem holds.

\begin{eqnarray}
\label{eq:rho225}
    \lambda(\r) = (a, a, b, b, c, c, c, c, c) \quad (a>b>c>0).
\end{eqnarray}
Similarly, we obtain a linear system in $t_1$, $t_3$, and $t_9$, whose coefficient matrix always has rank at least two. And then only the trivial solution exists. Thus for $\r$ in \eqref{eq:rho225}, the conclusion of the theorem holds.

\begin{eqnarray}
\label{eq:rho234}
    \lambda(\r) = (a, a, b, b, b, c, c, c, c) \quad (a>b>c>0).
\end{eqnarray}
Similarly, we obtain a linear system in $t_1$, $t_3$, and $t_9$, whose coefficient matrix always has rank at least two. And then only the trivial solution exists. Thus for $\r$ in \eqref{eq:rho234}, the conclusion of the theorem holds.

\begin{eqnarray}
\label{eq:rho243}
    \lambda(\r) = (a, a, b, b, b, b, c, c, c) \quad (a>b>c>0).
\end{eqnarray}
Firstly, suppose $l_{1}(\lambda) = 0$. Similarly, we obtain a linear system in $t_1$, $t_3$, and $t_9$, whose coefficient matrix always has rank at least two. And then only the trivial solution exists. Secondly, suppose $l_{2}(\lambda) = 0$. Similarly, we obtain a linear system in $t_1$, $t_3$, and $t_9$, whose coefficient matrix has rank at least two, unless $\vec{u_3} = \begin{bmatrix} 0 \\ \pm \frac{\sqrt2}{2}  \\ \mp \frac{\sqrt2}{2} \end{bmatrix}$. But this implies that $b = 0$, which leads to a contradiction. Consequently, coefficient matrix always has rank at least two and then only the trivial solution exists. Thus for $\r$ in \eqref{eq:rho243}, the conclusion of the theorem holds.

\begin{eqnarray}
\label{eq:rho252}
    \lambda(\r) = (a, a, b, b, b, b, b, c, c) \quad (a>b>c>0).
\end{eqnarray}
Similarly, we obtain a linear system in $t_1$, $t_3$, and $t_9$, whose coefficient matrix always has rank at least two. And then only the trivial solution exists. Thus for $\r$ in \eqref{eq:rho252}, the conclusion of the theorem holds.

\begin{eqnarray}
\label{eq:rho261}
    \lambda(\r) = (a, a, b, b, b, b, b, b, c) \quad (a>b>c>0).
\end{eqnarray}
Similarly, we obtain a linear system in $t_1$, $t_3$, and $t_9$, whose coefficient matrix has rank at least two, unless $\vec{u_3} = \begin{bmatrix} 0 \\ \pm \frac{\sqrt2}{2}  \\ \mp \frac{\sqrt2}{2} \end{bmatrix}$. This implies that $b = 0$, which leads to a contradiction. Consequently, coefficient matrix always has rank at least two and then only the trivial solution exists. Thus for $\r$ in \eqref{eq:rho261}, the conclusion of the theorem holds.

(iii) $\nu_{3,k,6-k}$ where $k=1, \cdots , 5$. Firstly we consider the state with $k=1$,
\begin{eqnarray}
\label{eq:rho315}
    \lambda(\r) = (a, a, a, b, c, c, c, c, c) \quad (a>b>c>0).
\end{eqnarray}
Similarly, we obtain a linear system in $t_1$, $t_4$, and $t_9$, whose coefficient matrix has rank at least two, unless $\vec{u_3} = \begin{bmatrix} \pm \frac{\sqrt2}{2} \\ \mp \frac{\sqrt2}{2}  \\ 0 \end{bmatrix}$. It implies that $a = c = 0$, which leads to a contradiction. Consequently, coefficient matrix always has rank at least two and then only the trivial solution exists. Thus for $\r$ in \eqref{eq:rho315}, the conclusion of the theorem holds. Next we consider the state with $k=2$,
\begin{eqnarray}
\label{eq:rho324}
    \lambda(\r) = (a, a, a, b, b, c, c, c, c) \quad (a>b>c>0).
\end{eqnarray}
Similarly, we obtain a linear system in $t_1$, $t_4$, and $t_9$, whose coefficient matrix always has rank at least two. And then only the trivial solution exists. Thus for $\r$ in \eqref{eq:rho324}, the conclusion of the theorem holds. Third we consider the state with $k=3$,
\begin{eqnarray}
\label{eq:rho333}
    \lambda(\r) = (a, a, a, b, b, b, c, c, c) \quad (a>b>c>0).
\end{eqnarray}
Firstly, suppose $l_{1}(\lambda) = 0$. Similarly, we obtain a linear system in $t_1$, $t_4$, and $t_9$, whose coefficient matrix always has rank at least two. And then only the trivial solution exists. Secondly, suppose $l_{2}(\lambda) = 0$. Similarly, we obtain a linear system in $t_1$, $t_4$, and $t_9$, whose coefficient matrix has rank at least two, unless $\vec{u_3} = \begin{bmatrix} 0 \\ \pm \frac{\sqrt2}{2}  \\ \mp \frac{\sqrt2}{2} \end{bmatrix}$. But this implies that $a = b = 0$, which leads to a contradiction. Consequently, coefficient matrix always has rank at least two and then only the trivial solution exists. Thus for $\r$ in \eqref{eq:rho333}, the conclusion of the theorem holds. Fourth we consider the state with $k=4$,
\begin{eqnarray}
\label{eq:rho342}
    \lambda(\r) = (a, a, a, b, b, b, b, c, c) \quad (a>b>c>0).
\end{eqnarray}
Similarly, we obtain a linear system in $t_1$, $t_4$, and $t_9$, whose coefficient matrix always has rank at least two. And then only the trivial solution exists. Thus for $\r$ in \eqref{eq:rho342}, the conclusion of the theorem holds. Fifth we consider the state with $k=5$,
\begin{eqnarray}
\label{eq:rho351}
    \lambda(\r) = (a, a, a, b, b, b, b, b, c) \quad (a>b>c>0).
\end{eqnarray}
Similarly, we obtain a linear system in $t_1$, $t_4$, and $t_9$, whose coefficient matrix has rank at least two, unless $\vec{u_3} = \begin{bmatrix} 0 \\ \pm \frac{\sqrt2}{2}  \\ \mp \frac{\sqrt2}{2} \end{bmatrix}$. It implies $a = b = 0$, which leads to a contradiction. Consequently, coefficient matrix always has rank at least two and then only the trivial solution exists. Thus for $\r$ in \eqref{eq:rho351}, the conclusion of the theorem holds.

(vi) $\nu_{4,k,5-k}$ where $k=1, \cdots , 4$.
\begin{eqnarray}
\label{eq:rho414}
    \lambda(\r) = (a, a, a, a, b, c, c, c, c) \quad (a>b>c>0).
\end{eqnarray}
Similarly, we obtain a linear system in $t_1$, $t_5$, and $t_9$, whose coefficient matrix always has rank at least two. And then only the trivial solution exists. Thus for $\r$ in \eqref{eq:rho414}, the conclusion of the theorem holds.

\begin{eqnarray}
\label{eq:rho423}
    \lambda(\r) = (a, a, a, a, b, b, c, c, c) \quad (a>b>c>0).
\end{eqnarray}
Firstly, suppose $l_{1}(\lambda) = 0$. Similarly, we obtain a linear system in $t_1$, $t_5$, and $t_9$, whose coefficient matrix has rank at least two, unless $\vec{u_3} = \begin{bmatrix} \pm \frac{\sqrt2}{2} \\ \mp \frac{\sqrt2}{2}  \\ 0 \end{bmatrix}$. But this implies that $a = c = 0$, which leads to a contradiction. Consequently, coefficient matrix always has rank at least two and then only the trivial solution exists. Secondly, suppose $l_{2}(\lambda) = 0$. Similarly, we obtain a linear system in $t_1$, $t_5$, and $t_9$, whose coefficient matrix always has rank at least two. And then only the trivial solution exists. Thus for $\r$ in \eqref{eq:rho423}, the conclusion of the theorem holds.

\begin{eqnarray}
\label{eq:rho432}
    \lambda(\r) = (a, a, a, a, b, b, b, c, c) \quad (a>b>c>0).
\end{eqnarray}
Similarly, we obtain a linear system in $t_1$, $t_5$, and $t_9$, whose coefficient matrix always has rank at least two. And then only the trivial solution exists. Thus for $\r$ in \eqref{eq:rho432}, the conclusion of the theorem holds.

\begin{eqnarray}
\label{eq:rho441}
    \lambda(\r) = (a, a, a, a, b, b, b, b, c) \quad (a>b>c>0).
\end{eqnarray}
Similarly, we obtain a linear system in $t_1$, $t_5$, and $t_9$, whose coefficient matrix always has rank at least two. And then only the trivial solution exists. Thus for $\r$ in \eqref{eq:rho441}, the conclusion of the theorem holds.

(v) $\nu_{5,k,4-k}$ where $k=1,2,3$.
\begin{eqnarray}
\label{eq:rho513}
    \lambda(\r) = (a, a, a, a, a, b, c, c, c) \quad (a>b>c>0).
\end{eqnarray}
Firstly, suppose $l_{1}(\lambda) = 0$. Similarly, we obtain a linear system in $t_1$, $t_6$, and $t_9$, whose coefficient matrix always has rank at least two. And then only the trivial solution exists. Secondly, suppose $l_{2}(\lambda) = 0$. Similarly, we obtain a linear system in $t_1$, $t_6$, and $t_9$, whose coefficient matrix always has rank at least two. And then only the trivial solution exists. Thus for $\r$ in \eqref{eq:rho513}, the conclusion of the theorem holds.

\begin{eqnarray}
\label{eq:rho522}
    \lambda(\r) = (a, a, a, a, a, b, b, c, c) \quad (a>b>c>0).
\end{eqnarray}
Similarly, we obtain a linear system in $t_1$, $t_6$, and $t_9$, whose coefficient matrix always has rank at least two. And then only the trivial solution exists. Thus for $\r$ in \eqref{eq:rho522}, the conclusion of the theorem holds.

\begin{eqnarray}
\label{eq:rho531}
    \lambda(\r) = (a, a, a, a, a, b, b, b, c) \quad (a>b>c>0).
\end{eqnarray}
Similarly, we obtain a linear system in $t_1$, $t_6$, and $t_9$, whose coefficient matrix has rank at least two, unless $\vec{u_3} = \begin{bmatrix} 0 \\ \pm \frac{\sqrt2}{2}  \\ \mp \frac{\sqrt2}{2} \end{bmatrix}$. But this implies that $b = 0$, which leads to a contradiction. Consequently, coefficient matrix always has rank at least two and then only the trivial solution exists. Thus for $\r$ in \eqref{eq:rho531}, the conclusion of the theorem holds.

(vi) $\nu_{6,k,3-k}$ where $k=1,2$.
\begin{eqnarray}
\label{eq:rho612}
    \lambda(\r) = (a, a, a, a, a, a, b, c, c) \quad (a>b>c>0).
\end{eqnarray}
Firstly, suppose $l_{1}(\lambda) = 0$. Similarly, we obtain a linear system in $t_1$, $t_7$, and $t_9$, whose coefficient matrix has rank at least two, unless $\vec{u_3} = \begin{bmatrix} 0 \\ 0  \\ \pm 1 \end{bmatrix}$. But this implies that $a = 0$, which leads to a contradiction. Consequently, coefficient matrix always has rank at least two and then only the trivial solution exists. Secondly, suppose $l_{2}(\lambda) = 0$. Similarly, we obtain a linear system in $t_1$, $t_7$, and $t_9$, whose coefficient matrix always has rank at least two. And then only the trivial solution exists. Thus for $\r$ in \eqref{eq:rho612}, the conclusion of the theorem holds.

\begin{eqnarray}
\label{eq:rho621}
    \lambda(\r) = (a, a, a, a, a, a, b, b, c) \quad (a>b>c>0).
\end{eqnarray}
Firstly, suppose $l_{1}(\lambda) = 0$. Similarly, we obtain a linear system in $t_1$, $t_7$, and $t_9$, whose coefficient matrix has rank at least two, unless $\vec{u_3} = \begin{bmatrix} 0 \\ 0  \\ \pm 1 \end{bmatrix}$. But this implies that $a = 0$, which leads to a contradiction. Consequently, coefficient matrix always has rank at least two and then only the trivial solution exists. Secondly, suppose $l_{2}(\lambda) = 0$. Similarly, we obtain a linear system in $t_1$, $t_7$, and $t_9$, whose coefficient matrix has rank at least two, unless $\vec{u_3} = \begin{bmatrix} 0 \\ \pm \frac{\sqrt2}{2}  \\ \mp \frac{\sqrt2}{2} \end{bmatrix}$. But this implies that $a = 0$, which leads to a contradiction. Consequently, coefficient matrix always has rank at least two and then only the trivial solution exists. Thus for $\r$ in \eqref{eq:rho612}, the conclusion of the theorem holds.

(vii) $\nu_{7,1,1}$.
\begin{eqnarray}
\label{eq:rho711}
    \lambda(\r) = (a, a, a, a, a, a, a, b, c) \quad (a>b>c>0).
\end{eqnarray}
Similarly, we obtain a linear system in $t_1$, $t_7$, and $t_9$, whose coefficient matrix has rank at least two, unless $\vec{u_3} = \begin{bmatrix} 0 \\ 0  \\ \pm 1 \end{bmatrix}$. But this implies that $a = 0$, which leads to a contradiction. Consequently, coefficient matrix always has rank at least two and then only the trivial solution exists. Thus for $\r$ in \eqref{eq:rho711}, the conclusion of the theorem holds.

Therefore, Theorem \ref{all_3eigen} is proved. 
\end{proof}

\subsection{Cases (i) and (iii)-(vii) in the proof of Table \ref{tab:performance1}}
\label{(i)-(iii)-(vii)_of_tab1}

(i) We consider the boundary point 
\begin{eqnarray}
\label{nu1,1,7}
    \nu_{1,1,7} = \diag \{a, b, c\cdots, c\} = \diag \{1-b-7c, b, c\cdots, c\}. 
\end{eqnarray}
This condition implies that $l_{1}(\lambda) = l_{2}(\lambda) = 0$ because $\lambda_6 = \lambda_7$. We have
\begin{eqnarray}
\label{eq:nu1,1,7}
    \left|\begin{matrix}2c&b+8c-1&c-b\\b+8c-1&2c&0\\c-b&0&2c\\\end{matrix}\right|
    = -2 c (2 b^2 + 2 b (7 c - 1) + 61 c^2 - 16 c + 1)
    = 0.
\end{eqnarray}
Using \eqref{nu1,1,7}, we obtain
\begin{eqnarray}
    b \in (c,\frac{1-7c}{2}), i.e.\ c \in (0, \frac{1}{9}).
\end{eqnarray}
The solutions $b$ of \eqref{eq:nu1,1,7} are 
\begin{eqnarray}
    b_1 = \frac{1-7c+\sqrt{-73 c^2 + 18 c - 1}}{2}, \ 
    b_2 = \frac{1-7c-\sqrt{-73 c^2 + 18 c - 1}}{2}.
\end{eqnarray}
So we have 
\begin{eqnarray}
\label{c_range1}
    0.084542\cdots = \frac{9-2\sqrt{2}}{73} \le c < \frac{1}{9} = 0.111111\cdots. 
\end{eqnarray}
It is clear that $b_1 + b_2 = 1-7c$, which is also the sum of the largest and the second largest eigenvalue of $\nu_{1,1,7}$ in \eqref{nu1,1,7}. We have
\begin{eqnarray}
\label{a,b_of_nu1,1,7}
    a = \frac{1-7c+\sqrt{-73 c^2 + 18 c - 1}}{2}, \ 
    b = \frac{1-7c-\sqrt{-73 c^2 + 18 c - 1}}{2}. 
\end{eqnarray}
From \eqref{a,b_of_nu1,1,7}, $-73 c^2 + 18 c - 1 > 0$ because $a>b$. Thus, $c$ also satisfies
\begin{eqnarray}
\label{c<b<a_reform}
    c < \frac{1-7c-\sqrt{-73 c^2 + 18 c - 1}}{2} < \frac{1-7c+\sqrt{-73 c^2 + 18 c - 1}}{2}. 
\end{eqnarray}
Considering both \eqref{c_range1} and \eqref{c<b<a_reform}, we obtain 
\begin{eqnarray}
\label{eq:c_range2}
    0.084542\cdots =\frac{1}{9+2\sqrt{2}} =  \frac{9-2\sqrt{2}}{73} < c < \frac{1}{11} = 0.090909\cdots. 
\end{eqnarray}
Therefore, the boundary point $\r = \nu_{1,1,7}$ has the form of
\begin{eqnarray}
    \diag \{\frac{1-7c+\sqrt{-73 c^2 + 18 c - 1}}{2}, \frac{1-7c-\sqrt{-73 c^2 + 18 c - 1}}{2}, c, \cdots, c\} 
\end{eqnarray}
with \eqref{eq:c_range2}, which is an extreme point, too. 

In addition, we consider the limit of the extreme point $\nu_{1,1,7}$. We obtain
\begin{eqnarray}
    \lim_{c\rightarrow\frac{1}{11}} \nu_{1,1,7} = \zeta_1, \ \lim_{c\rightarrow\frac{9-2\sqrt{2}}{73}} \nu_{1,1,7} = \zeta_2. 
\end{eqnarray}

(iii) We consider the boundary point 
\begin{eqnarray}
\label{nu1,3,5}
    \nu_{1,3,5} = \diag \{a, b, b, b, c\cdots, c\} = \diag \{1-3b-5c, b, b, b, c\cdots, c\}.
\end{eqnarray}
This condition implies that $l_{1}(\lambda) = l_{2}(\lambda) = 0$ because $\lambda_6 = \lambda_7$. We have
\begin{eqnarray}
\label{eq:nu1,3,5}
    \left|\begin{matrix}2c&3b+6c-1&c-b\\3b+6c-1&2c&c-b\\c-b&c-b&2b\\\end{matrix}\right|
    = 0,
\end{eqnarray}
i.e. 
\begin{eqnarray}
    -2 [6 b^3 + (38 c - 5) b^2 + (37 c^2 - 14 c + 1) b - 4 c^3 + c^2] = 0,
\end{eqnarray}
and from \eqref{nu1,3,5}, we obtain
\begin{eqnarray}
\label{eq:c in(0,1/9)}
    b \in (c,\frac{1-5c}{4}), \text{i.e.}, \ c \in (0, \frac{1}{9}).
\end{eqnarray}
The solutions of \eqref{eq:nu1,3,5} are 
\begin{eqnarray}
    b_1 = \frac{1-4c}{3},
\end{eqnarray}
\begin{eqnarray}
    b_2 = \frac{1-10c+\sqrt{108c^2-20c+1}}{4},
\end{eqnarray}
\begin{eqnarray}
    b_3 = \frac{1-10c-\sqrt{108c^2-20c+1}}{4}. 
\end{eqnarray}

(iii.1) $b_1 = \frac{1-4c}{3}$. Then $a=1-3b_1-5c=-c<0$. This leads to a contradiction, so we discard this solution. 

Consequently, the solution of $b$ is either $b_2$ or $b_3$. Note that $108c^2-20c+1 \geq 0$ holds for all $c \in \mathbb{R}$. 

(iii.2) $b_2 = \frac{1-10c+\sqrt{108c^2-20c+1}}{4}$. Then $a=1-3b_2-5c=\frac{1+10c-3\sqrt{108c^2-20c+1}}{4}>\frac{1-10c+\sqrt{108c^2-20c+1}}{4}$. This implies
\begin{eqnarray}
\label{c_range1_135}
    0.070805 \cdots = \frac{10-\sqrt{17}}{83} < c < \frac{10+\sqrt{17}}{83} = 0.170157 \cdots. 
\end{eqnarray}
Besides from $b_2>c>0$, we obtain 
\begin{eqnarray}
\label{c_range2_135}
    0<c<\frac{1}{11} = 0.090909 \cdots. 
\end{eqnarray}
From \eqref{eq:c in(0,1/9)},  \eqref{c_range1_135} and \eqref{c_range2_135}, we have
\begin{eqnarray}
\label{c_range3_135}
    0.070805 \cdots = \frac{1}{10+\sqrt{17}} = \frac{10-\sqrt{17}}{83} < c < \frac{1}{11} = 0.090909 \cdots.
\end{eqnarray}

(iii.3) $b_3 = \frac{1-10c-\sqrt{108c^2-20c+1}}{4}$. However, from $b_3 > 0$, we obtain that $1-10c-\sqrt{108c^2-20c+1} > 0$, which implies $100c^2>108c^2$. This leads to a contradiction, so we discard this solution. 

To conclude, the boundary point $\r = \nu_{1,3,5}$ in \eqref{nu1,3,5} with \eqref{c_range3_135} satisfies
\begin{eqnarray}
    a=\frac{1+10c-3\sqrt{108c^2-20c+1}}{4}, \ b=\frac{1-10c+\sqrt{108c^2-20c+1}}{4}.  
\end{eqnarray}
So $\r=\n_{1,3,5}$ is an extreme point of $\app_{3,3}$. 

In addition, we consider the limit of the extreme point $\nu_{1,3,5}$. We obtain
\begin{eqnarray}
\lim_{c\rightarrow\frac{1}{11}} \nu_{1,3,5} = \zeta_1, \ \lim_{c\rightarrow\frac{10-\sqrt{17}}{83}} \nu_{1,3,5} = \zeta_4. 
\end{eqnarray}

(iv) We consider the boundary point 
\begin{eqnarray}
\label{nu1,4,4}
    \nu_{1,4,4} = \diag \{a, b, b, b, b, c, c, c, c\} = \diag \{1-4b-4c, b, b, b, b, c, c, c, c\}.
\end{eqnarray}
From \eqref{nu1,4,4}, we obtain
\begin{eqnarray}
\label{eq:144c in(0,1/9)}
    b \in (c,\frac{1-4c}{5}), \text{i.e.}, \ c \in (0, \frac{1}{9}).
\end{eqnarray}
This condition implies that $l_{1}(\lambda) = l_{2}(\lambda) = 0$ because $\lambda_6 = \lambda_7$. We have
\begin{eqnarray}
&\left|\begin{matrix}2c&4b+5c-1&c-b\\4b+5c-1&2c&0\\c-b&0&2b\\\end{matrix}\right|=0,
\end{eqnarray}
i.e.,
\begin{eqnarray}
\label{eq:nu1,4,4}
    -2[16b^3+(41c-8)b^2+(19c^2-10c+1)b+c^3]=0
\end{eqnarray}
By direct computation, we obtain that $y<c<\frac{1}{11}$, where $y=0.056991\cdots$ is the second root of the equation $481y^3-37y^2-17y+1=0$.
Furthermore, the value of $b$ is the second root of the cubic equation $16x^3+(41c-8)x^2+(19c^2-10c+1)x+c^3=0$, once $c$ is fixed.

In addition, we consider the limit of the extreme point $\nu_{1,4,4}$. We obtain
\begin{eqnarray}
\lim_{c\rightarrow\frac{1}{11}} \nu_{1,4,4} = \zeta_1, \ \lim_{c\rightarrow y} \nu_{1,4,4} = \zeta_5. 
\end{eqnarray}

(v) We consider the boundary point 
\begin{eqnarray}
\label{nu1,5,3}
    \nu_{1,5,3} = \diag \{a, b, \cdots, b, c, c, c\} = \diag \{1-5b-3c, b, \cdots, b, c, c, c\}.
\end{eqnarray}
Based on Lemma \ref{boun-extr}, the following analysis considers two cases: $l_1(\lambda)=0$ and $l_2(\lambda)=0$. 

Firstly, if $l_1(\lambda)=0$, which implies 
\begin{eqnarray}
    \label{eq:nu1,5,3}
    \left|\begin{matrix}2c&5b+4c-1&0\\5b+4c-1&2c&0\\0&0&2b\\\end{matrix}\right|
    = 0.
\end{eqnarray}
\eqref{eq:nu1,5,3} yields two solutions. The first solution $b_1=\frac{1-6c}{5}$ implies $a=3c$, resulting in a boundary point $\nu_{1,5,3}$. However, according to Lemma \ref{le:nu(1,5,3)} stated earlier, this must be discarded. The second solution $b_1=\frac{1-2c}{5}$ implies $a=-c$, which is discarded, too.

Secondly, if $l_2(\lambda)=0$, which implies 
\begin{eqnarray}
    \label{eq:nu1,5,3(2)}
    \left|\begin{matrix}2c&5b+4c-1&c-b\\5b+4c-1&2b&0\\c-b&0&2b\\\end{matrix}\right|
    = 0.
\end{eqnarray}
Taking the only admissible solution of \eqref{eq:nu1,5,3(2)}, we obtain an extreme point in \eqref{nu1,5,3} with
\begin{eqnarray}
\label{a,b_of_nu1,5,3}
    a=\frac{1+7c+5\sqrt{-153c^2+38c-1}}{26},\ b=\frac{5-17c-\sqrt{-153c^2+38c-1}}{26},
\end{eqnarray}
where
\begin{eqnarray}
    0.039382 \cdots = \frac{5-2\sqrt{3}}{39} < c < \frac{1}{11} = 0.090909 \cdots. 
\end{eqnarray}
However, this solution \eqref{a,b_of_nu1,5,3} does not make $L_1(\lambda)$ a positive semi-definite matrix. Therefore, $\nu_{1,5,3}$ in \eqref{nu1,5,3} cannot be an extreme point. In fact, we consider the system: 
\begin{eqnarray}
\label{L1(lambda)>=O}
    2b \geq 0, \ -5b-2c+1 \geq 0, \ 5b+6c-1 \geq 0,
\end{eqnarray}
\begin{eqnarray}
    l_2(\lambda) = 0,
\end{eqnarray}
where \eqref{L1(lambda)>=O} is an equivalent condition for matrix $L_1(\lambda)$ to be positive semi-definite. Therefore, there is no boundary point with the form of \eqref{nu1,5,3} that satisfies condition $l_2(\lambda)=0$ according to Lemma \ref{AP3,3_judge_positive_semi-definite}. 

(vi) We consider the boundary point 
\begin{eqnarray}
\label{nu1,6,2}
    \nu_{1,6,2} = \diag \{a, b, \cdots, b, c, c\} = \diag \{1-6b-2c, b, \cdots, b, c, c\}.
\end{eqnarray}
This condition implies that $l_{1}(\lambda) = l_{2}(\lambda) = 0$ because $\lambda_6 = \lambda_7$. We have
\begin{eqnarray}
\label{eq:nu1,6,2}
    \left|\begin{matrix}2c&6b+3c-1&0\\6b+3c-1&2b&0\\0&0&2b\\\end{matrix}\right|
    = 0.
\end{eqnarray}
Selecting an appropriate solution from \eqref{nu1,6,2} and \eqref{eq:nu1,6,2}, we obtain
\begin{eqnarray}
    a=\frac{2c+\sqrt{6c-17c^2}}{3}, \ b=\frac{3-8c-\sqrt{6c-17c^2}}{18},
\end{eqnarray}
where 
\begin{eqnarray}
    0.023365 \cdots = \frac{23-14\sqrt{2}}{137} < c < \frac{1}{11} = 0.090909 \cdots. 
\end{eqnarray}

In addition, we consider the limit of the extreme point $\nu_{1,6,2}$. We obtain
\begin{eqnarray}
\lim_{c\rightarrow\frac{1}{11}} \nu_{1,6,2} = \zeta_1, \ \lim_{c\rightarrow\frac{23-14\sqrt{2}}{137}} \nu_{1,6,2} = \zeta_7. 
\end{eqnarray}

(vii) We consider the boundary point 
\begin{eqnarray}
\label{nu1,7,1}
    \nu_{1,7,1} = \diag \{a, b, \cdots, b, c\} = \diag \{1-7b-c, b, \cdots, b, c\}.
\end{eqnarray}
This condition implies that $l_{1}(\lambda) = l_{2}(\lambda) = 0$ because $\lambda_6 = \lambda_7$. We have
\begin{eqnarray}
\label{eq:nu1,7,1}
    \left|\begin{matrix}2c&8b+c-1&0\\8b+c-1&2b&0\\0&0&2b\\\end{matrix}\right|
    = 0.
\end{eqnarray}
Selecting an appropriate solution by \eqref{nu1,7,1} and \eqref{eq:nu1,7,1}, we obtain
\begin{eqnarray}
    a=\frac{4-11c+7\sqrt{8c-7c^2}}{32}, \ b=\frac{4-3c-\sqrt{8c-7c^2}}{32},
\end{eqnarray}
where 
\begin{eqnarray}
    0 < c < \frac{1}{11} = 0.090909 \cdots. 
\end{eqnarray}
In addition, we consider the limit of the extreme point $\nu_{1,7,1}$. We obtain
\begin{eqnarray}
\lim_{c\rightarrow\frac{1}{11}} \nu_{1,7,1} = \zeta_1, \ \lim_{c\rightarrow0} \nu_{1,7,1} = \zeta_8. 
\end{eqnarray}

\subsection{Other cases in the proof of Table \ref{tab:performance2}}
\label{other_of_tab2}
(i) We consider the boundary point 
\begin{eqnarray}
\label{nu2,1,6}
    \nu_{2,1,6} = \diag \{a, a, b, c, c, c, c, c, c\} = \diag \{\frac{1-b-6c}{2}, \frac{1-b-6c}{2}, b, c, c, c, c, c, c\}.
\notag\\
\end{eqnarray}
This condition implies that $l_{1}(\lambda) = l_{2}(\lambda) = 0$ because $\lambda_6 = \lambda_7$. We have
\begin{eqnarray}
    \left|\begin{matrix}2c&\frac{b+8c-1}{2}&\frac{b+8c-1}{2}\\\frac{b+8c-1}{2}&2c&c-b\\\frac{b+8c-1}{2}&c-b&2c\\\end{matrix}\right|
    = 0.
\end{eqnarray}
i.e.,
\begin{eqnarray}
\label{eq:nu2,1,6}
    \frac{1}{2}[-b^3+(2-21c)b^2+(-72c^2+18c-1)b+c(-52c^2+16c-1)]=0.
\end{eqnarray}
By direct computation using \eqref{nu2,1,6} and \eqref{eq:nu2,1,6}, we obtain 
\begin{eqnarray}
    a=2c+\sqrt{12c^2-c},\ b=1-10c-2\sqrt{12c^2-c}
\end{eqnarray}
with 
\begin{eqnarray}
    0.083333\cdots=\frac{1}{12}<c<\frac{9-2\sqrt{2}}{73}=0.084542\cdots. 
\end{eqnarray}

In addition, we consider the limit of the extreme point $\nu_{2,1,6}$. We obtain
\begin{eqnarray}
\lim_{c\rightarrow\frac{9-2\sqrt{2}}{73}} \nu_{2,1,6} = \zeta_2, \ 
\lim_{c\rightarrow\frac{1}{12}} \nu_{2,1,6} = \zeta_3. 
\end{eqnarray}

(ii) We consider the boundary point 
\begin{eqnarray}
\label{nu2,2,5}
    \begin{aligned}
        \nu_{2,2,5} &= \diag \{a, a, b, b, c, c, c, c, c\} \\
        &= \diag \{\frac{1-2b-5c}{2}, \frac{1-2b-5c}{2}, b, b, c, c, c, c, c\}.
    \end{aligned}
\end{eqnarray}
This condition implies that $l_{1}(\lambda) = l_{2}(\lambda) = 0$ because $\lambda_6 = \lambda_7$. We have
\begin{eqnarray}
    \left|\begin{matrix}2c&\frac{2b+7c-1}{2}&\frac{2b+7c-1}{2}\\\frac{2b+7c-1}{2}&2c&c-b\\\frac{2b+7c-1}{2}&c-b&2b\\\end{matrix}\right|
    = 0.
\end{eqnarray}
i.e.,
\begin{eqnarray}
\label{eq:nu2,2,5}
   -4b^3+(4-30c)b^2+(-37c^2+14c-1)b-2c^3=0.
\end{eqnarray}
From \eqref{nu2,2,5}, we have 
\begin{eqnarray}
\label{neq:range_2,2,5}
    \frac{1-2b-5c}{2}>b>c>0.
\end{eqnarray}
By solving the system of \eqref{eq:nu2,2,5} and \eqref{neq:range_2,2,5}, we obtain
\begin{eqnarray}
    0.070805 \cdots=\frac{10 - \sqrt{17}}{83}<c<\frac{9-2\sqrt{2}}{73}=0.084542\cdots. 
\end{eqnarray}
Furthermore, the value of $b$ is the second root of the cubic equation \eqref{eq:nu2,2,5}, once $c$ is fixed.

In addition, we consider the limit of the extreme point $\nu_{2,2,5}$. We obtain
\begin{eqnarray}
\lim_{c\rightarrow\frac{9-2\sqrt{2}}{73}} \nu_{2,2,5} = \zeta_2, \ 
\lim_{c\rightarrow\frac{10 - \sqrt{17}}{83}} \nu_{2,2,5} = \zeta_4. 
\end{eqnarray}

(iii) We consider the boundary point 
\begin{eqnarray}
\label{nu2,3,4}
    \begin{aligned}
        \nu_{2,3,4} &= \diag \{a, a, b, b, b, c, c, c, c\} \\
        &= \diag \{\frac{1-3b-4c}{2}, \frac{1-3b-4c}{2}, b, b, b, c, c, c, c\}.
    \end{aligned}
\end{eqnarray}
This condition implies that $l_{1}(\lambda) = l_{2}(\lambda) = 0$ because $\lambda_6 = \lambda_7$. We have
\begin{eqnarray}
    \left|\begin{matrix}2c&\frac{3b+6c-1}{2}&\frac{3b+6c-1}{2}\\\frac{3b+6c-1}{2}&2c&0\\\frac{3b+6c-1}{2}&0&2b\\\end{matrix}\right|
    = 0.
\end{eqnarray}
i.e.,
\begin{eqnarray}
\label{eq:nu2,3,4}
   \frac{1}{2}[-9b^3+(6-45c)b^2+(-56c^2+18c-1)b-c(1-6c)^2]=0.
\end{eqnarray}
By direct computation using \eqref{nu2,3,4} and \eqref{eq:nu2,3,4}, we obtain that $y<c<\frac{9-2\sqrt{2}}{73}=0.084542\cdots$, where $y=0.056991\cdots$ is the second root of the equation $481y^3-37y^2-17y+1=0$.
Furthermore, the value of $b$ is the second root of the cubic equation \eqref{eq:nu2,3,4}, once $c$ is fixed.

In addition, we consider the limit of the extreme point $\nu_{2,3,4}$. We obtain
\begin{eqnarray}
\lim_{c\rightarrow\frac{9-2\sqrt{2}}{73}} \nu_{2,3,4} = \zeta_2, \ 
\lim_{c\rightarrow y} \nu_{2,3,4} = \zeta_5. 
\end{eqnarray}

(v) We consider the boundary point 
\begin{eqnarray}
\label{nu2,5,2}
    \begin{aligned}
        \nu_{2,5,2} &= \diag \{a, a, b, b, b, b, b, c, c\} \\
        &= \diag \{\frac{1-5b-2c}{2}, \frac{1-5b-2c}{2}, b, b, b, b, b, c, c\}.
    \end{aligned}
\end{eqnarray}
This condition implies that $l_{1}(\lambda) = l_{2}(\lambda) = 0$ because $\lambda_6 = \lambda_7$. We have
\begin{eqnarray}
    \left|\begin{matrix}2c&\frac{5b+4c-1}{2}&\frac{7b+2c-1}{2}\\\frac{5b+4c-1}{2}&2b&0\\\frac{7b+2c-1}{2}&0&2b\\\end{matrix}\right|
    = 0.
\end{eqnarray}
i.e.,
\begin{eqnarray}
\label{eq:nu2,5,2}
   -b [37 b^2 + 2 (13 c - 6) b + 10 c^2 - 6 c + 1]=0.
\end{eqnarray}
By direct computation and using \eqref{nu2,5,2} and \eqref{eq:nu2,5,2}, we obtain that 
\begin{eqnarray}
    a= \frac{7 - 9 c + 5 \sqrt{-201 c^2 + 66 c - 1}}{37}, \ 
    b= \frac{6 - 13 c - \sqrt{-201 c^2 + 66 c - 1}}{37}
\end{eqnarray}
with
\begin{eqnarray}
    0.023365 \cdots = \frac{23-14\sqrt{2}}{137} < c < \frac{9-2\sqrt{2}}{73}=0.084542\cdots. 
\end{eqnarray}

In addition, we consider the limit of the extreme point $\nu_{2,5,2}$. We obtain
\begin{eqnarray}
\lim_{c\rightarrow\frac{9-2\sqrt{2}}{73}} \nu_{2,5,2} = \zeta_2, \ 
\lim_{c\rightarrow\frac{23-14\sqrt{2}}{137}} \nu_{2,5,2} = \zeta_7. 
\end{eqnarray}

(vi) We consider the boundary point 
\begin{eqnarray}
\label{nu2,6,1}
    \begin{aligned}
        \nu_{2,6,1} &= \diag \{a, a, b, b, b, b, b, b, c\} \\
        &= \diag \{\frac{1-6b-c}{2}, \frac{1-6b-c}{2}, b, b, b, b, b, b, c\}.
    \end{aligned}
\end{eqnarray}
This condition implies that $l_{1}(\lambda) = l_{2}(\lambda) = 0$ because $\lambda_6 = \lambda_7$. We have
\begin{eqnarray}
    \left|\begin{matrix}2c&\frac{8b+c-1}{2}&\frac{8b+c-1}{2}\\\frac{8b+c-1}{2}&2b&0\\\frac{8b+c-1}{2}&0&2b\\\end{matrix}\right|
    = 0.
\end{eqnarray}
i.e.,
\begin{eqnarray}
\label{eq:nu2,6,1}
   -b [64 b^2 + 8 (c - 2) b + (c - 1)^2]=0.
\end{eqnarray}
By direct computation and using \eqref{nu2,6,1} and \eqref{eq:nu2,6,1}, we obtain that 
\begin{eqnarray}
    a= \frac{2 - 5c + 3\sqrt{4 c - 3 c^2}}{16}, \ 
    b= \frac{2 - c - \sqrt{4 c - 3 c^2}}{16}
\end{eqnarray}
with
\begin{eqnarray}
    0 < c < \frac{9-2\sqrt{2}}{73}=0.084542\cdots. 
\end{eqnarray}

In addition, we consider the limit of the extreme point $\nu_{2,5,2}$. We obtain
\begin{eqnarray}
\lim_{c\rightarrow\frac{9-2\sqrt{2}}{73}} \nu_{2,6,1} = \zeta_2, \ 
\lim_{c\rightarrow0} \nu_{2,6,1} = \zeta_8. 
\end{eqnarray}

\bibliographystyle{unsrt}

\bibliography{Thefinalanswer} 

\end{document}